%% file: main.tex
\documentclass[english]{article}
\PassOptionsToPackage{natbib=true}{biblatex}
\usepackage[T1]{fontenc}
\usepackage[latin9]{inputenc}
\usepackage{geometry}
\usepackage{color}
\usepackage{babel}
\usepackage{float}
\usepackage{mathrsfs}
\usepackage{amsmath}
\usepackage{amsthm}
\usepackage{amssymb}
\usepackage{stmaryrd}
\usepackage[unicode=true,pdfusetitle,
 bookmarks=true,bookmarksnumbered=false,bookmarksopen=false,
 breaklinks=false,pdfborder={0 0 1},backref=false,colorlinks=true]
 {hyperref}

\makeatletter

\floatstyle{ruled}
\newfloat{algorithm}{tbp}{loa}
\providecommand{\algorithmname}{Algorithm}
\floatname{algorithm}{\protect\algorithmname}

\newcommand{\lyxaddress}[1]{
	\par {\raggedright #1
	\vspace{1.4em}
	\noindent\par}
}
\theoremstyle{definition}
\newtheorem{defn}{\protect\definitionname}
\theoremstyle{remark}
\newtheorem{rem}{\protect\remarkname}
\theoremstyle{plain}
\newtheorem{thm}{\protect\theoremname}
\theoremstyle{plain}
\newtheorem{prop}{\protect\propositionname}
\theoremstyle{definition}
 \newtheorem{example}{\protect\examplename}
\theoremstyle{plain}
\newtheorem{lem}{\protect\lemmaname}
\theoremstyle{plain}
\newtheorem{cor}{\protect\corollaryname}

\makeatother

\usepackage[style=authoryear]{biblatex}
\providecommand{\corollaryname}{Corollary}
\providecommand{\definitionname}{Definition}
\providecommand{\examplename}{Example}
\providecommand{\lemmaname}{Lemma}
\providecommand{\propositionname}{Proposition}
\providecommand{\remarkname}{Remark}
\providecommand{\theoremname}{Theorem}

\begin{document}
\title{A general perspective on the Metropolis--Hastings kernel}
\author{Christophe Andrieu$^{*}$, Anthony Lee$^{*}$ and Sam Livingstone$^{\dagger}$}
\maketitle

\lyxaddress{$^{*}$School of Mathematics, University of Bristol, U.K. \\
$^{\dagger}$Department of Statistical Science, University College
London, U.K.}
\begin{abstract}
Since its inception the Metropolis--Hastings kernel has been applied
in sophisticated ways to address ever more challenging and diverse
sampling problems. Its success stems from the flexibility brought
by the fact that its verification and sampling implementation rests
on a local ``detailed balance'' condition, as opposed to a global
condition in the form of a typically intractable integral equation.
While checking the local condition is routine in the simplest scenarios,
this proves much more difficult for complicated applications involving
auxiliary structures and variables. Our aim is to develop a framework
making establishing correctness of complex Markov chain Monte Carlo
kernels a purely mechanical or algebraic exercise, while making communication
of ideas simpler and unambiguous by allowing a stronger focus on essential
features --- a choice of embedding distribution, an involution and
occasionally an acceptance function --- rather than the induced,
boilerplate structure of the kernels that often tends to obscure what
is important. This framework can also be used to validate kernels
that do not satisfy detailed balance, i.e. which are not reversible,
but a modified version thereof.
\end{abstract}
\tableofcontents{}

\include{generic}

\include{stop}
\include{dr}

\printbibliography

\appendix
\include{appendices}
\end{document}

%% file: generic.tex
\section{Introduction}

Assume one is interested in sampling from a probability distribution
$\pi$, defined on some probability space $(\mathsf{Z},\mathscr{Z})$.
A Markov chain Monte Carlo algorithm (MCMC) consists of simulating
a realization of a time-homogeneous Markov chain $(Z_{0},Z_{1},Z_{2}\ldots)$,
of say kernel $P$, with the property that the distribution of $Z_{n}$
becomes arbitrarily close to $\pi$ as $n\to\infty$ irrespective
of the distribution of $Z_{0}$. A property the kernel $P$, or its
components in the case of mixtures or composition of kernels, must
satisfy is to leave the distribution $\pi$ invariant, that is $\pi$
should be a fixed point of the Markov kernel. This is often referred
to as a ``global balance'' condition in the physics literature and
is most often not tractable to verify. Instead one can consider the
stronger ``detailed balance'' condition, or reversibility, a more
tractable property due to its local character which has led in particular
to the celebrated Metropolis-Hastings (MH) kernel \citep{metropolis1953equation,hastings1970monte},
the cornerstone of MCMC simulations, and a multitude of successful
variations. It is difficult to overstate the importance of detailed
balance when discussing the widespread application of MH kernels:
one can view such a kernel as being defined by a pair $(\pi,Q)$,
where $Q$ is a proposal Markov kernel, and the algorithm requires
only simulation according to $Q$ and computing densities associated
with $\pi$ and $Q$. This ease of use has lead to MH algorithms being
used in increasingly sophisticated contexts, leading to sometimes
spectacular practical improvements but also increased complexity when
establishing correctness (which we will take throughout to mean ensure
that $\pi$ is left invariant by $P$) and communicating their structure.
The aim of this paper is to develop a simple and general framework
to address these issues. In particular, the proposed framework defines
an invariant MH kernel $\Pi$ using a triple $(\mu,\phi,a)$, where
$\mu$ is the invariant distribution of $\Pi$, $\phi$ is an involution
and $a$ is an acceptance function, and retains similar ease-of-use
properties to those described above: one is required only to be able
to simulate from an appropriate conditional distribution of $\mu$,
calculate $\phi$ and ratios of densities involving $\mu$ and $\phi$.

\subsection{Contributions}

We consider a framework, extending \citet{tierney1998}, for defining
a $\mu$-reversible Markov kernel $\Pi$ of the Metropolis-Hastings
type, which only requires the specification of a triplet $(\mu,\phi,a)$
where $\mu$ is a probability measure on some space $(\mathsf{E},\mathscr{E})$,
$\phi\colon\mathsf{E}\rightarrow\mathsf{E}$ an involution, and $a\colon\mathbb{R}_{+}\rightarrow[0,1]$
an acceptance function. As we shall see, this covers most scenarios
of interest where sampling from $\pi$ as above is of interest by
letting $\pi$ be a marginal of $\mu$. More specifically for $\xi=(\xi_{0},\xi_{-0})\sim\mu$
such that $\xi_{0}\sim\pi$, $\xi_{-0}$ is a set of instrumental
random variables involved in the design of MH kernels--often referred
to as ``proposals'' for standard algorithm, but we refrain from
using this reductive terminology. Then the involution $\phi$ is applied,
defining $\xi':=\phi(\xi_{0})$, and $\xi'_{0}$ is the next state
of the Markov chain with a probability entirely determined by the
triplet $(\mu,\phi,a)$, or the Markov chain remains at $\xi_{0}$.
What is remarkable is that a correct algorithm is mathematically entirely
determined by this triplet--in particular there is, again at a theoretical
level, no need to determine an expression for the ``acceptance ratio'':
it exists! 

Practical implementation requires determining a tractable expression
for the acceptance ratio which is, fundamentally, of a measure theoretic
nature. Measure theoretic arguments are often overlooked in the literature
and indeed do not need to be considered in detail in most simple scenarios.
However this is not the case for more involved cases, where such issues
can lead to excruciating and \emph{ad hoc} contortions, and we have
made an effort here not to ignore them. We hope to convince the reader
that doing so is truly valuable and brings both generality and clarity
to the arguments. The background required is minimal and we provide
key results in the text: extensive knowledge of measure theory is
not a prerequisite to read the manuscript.

As we shall see we focus primarily on the choice of $(\mu,\phi)$
since the choice of $a$ is, at least theoretically, independent of
the choice of $(\mu,\phi)$ and can be determined optimally thanks
to the results of \citet{peskun1973optimum,tierney1998} in the reversible
setup and \citet{Andrieu2019} for nonreversible extensions. We revisit
numerous examples, some particularly simple for pedagogical purposes,
but also dedicate full sections (Sections~\ref{sec:NUTS} and \ref{sec:Delayed-rejection})
to popular examples which, we know, have baffled more than one researcher
before. This includes the No U-Turn Sampler \citep{hoffman2014no},
the extra-chance algorithm \citep{sohl2014hamiltonian,campos2015extra}
or event chain algorithms \citep{michel:2016}. In fact, We provide
generalizations and in some cases completely novel versions of these
algorithms.

We neither address the issues of convergence to equilibrium or ergodic
averages, nor answer the question of what is the best possible involution.
These are completely separate issues but we note that the ideas of
\citet{thin2020metflow}, or more generally adaptive MCMC \citep{andrieu2008tutorial},
could be used for the latter purpose while \citet{durmus2017convergence,thin:2020b}
provide some ideas concerning general results to establish irreducibility
and aperiodicity, the additional sufficient ingredients needed to
ensure convergence. There are in our view too many degrees of freedom
involved in the choice of good involutions, auxiliary variables and
their distributions and we do not believe that a theorem can, yet,
replace intuition, creativity and commonsense when designing good
MCMC schemes. Our aim here is rather to make checking that one's intuition
is correct a purely algebraic exercise, removing in particular the
need to revisit common points every time the question of correctness
arises, while helping with efficient and unambiguous communication
of potentially very complex schemes--see \citet{andrieu-doucet-yildirim-chopin-2020}
for an attempt at implementing this point of view.

We limit probabilistic arguments and notation to a minimum and, in
contrast with accepted common wisdom, most often use lower case fonts
for both random variables and their realizations in order to alleviate
notation. We hope this does not cause confusion.

\subsection{MCMC and involutions in the literature}

This work is strongly influenced by \citet{tierney1998} where the
possibility of using involutions as ``deterministic proposals''
is suggested, but not developed as a unifying tool as in the present
paper, and the treatment of densities therein is the direct source
of inspiration for our own treatment. The papers \citet{fang2014compressible,campos2015extra}
were complementary, and revealed to us the importance and generality
of the involution point of view, both in the reversible and nonreversible
setups, although not always in an explicit manner. A statement of
the main abstract result (Theorem~\ref{thm:invo-rev}) was given
in \citet[Proposition 3.5]{Andrieu2019} and presented in a series
of lectures organized at the Higher School of Economics lectures in
St. Petersburg in August 2019 \citep{andrieu:2019}, together with
various applications, while a preliminary version of the results concerned
with NUTS were presented at BayesComp 2020 in Florida in January 2020.
We have recently become aware of \citet[p. 64]{graham2018auxiliary}
where the possibility of using an involution as an update was suggested,
drawing on an analogy to \citet{green1995reversible}, but not developed.
In fact the involutive framework underpins \citet{green1995reversible}
but is not made explicit. The term ``Involutive MCMC'', perhaps
a tautology, was coined in \citet{neklyudov2020involutive} where
classical algorithms are revisited in turn following this perspective,
but no connection to earlier literature was made; we also note \citet{cusumanotowner2020automating}
with earlier claims and the interesting very recent contribution by
\citet{glattholtz2020acceptreject}. \citet{thin2020metflow} exploit
this type of representation of the MH kernel to design normalising
flows and \citet{thin:2020b} establish necessary conditions mirroring
\citet{tierney1998} in the skew detailed balance scenario, but also
general conditions ensuring aperiodicity and periodicity. 

\subsection{Notation and definitions}
\begin{itemize}
\item All real-valued functions we consider are Borel measurable.
\item If $\mu$ is a measure on $(E,\mathscr{E})$ and $f:\mathsf{E}\to\mathbb{R}$
is a $\mu$-integrable function then we denote the integral $\mu(f):=\int_{E}f(x)\mu({\rm d}x)$.
\item $\min\{a,b\}=a\wedge b$, $\max\{a,b\}=a\vee b$.
\item $f\cdot g$ is pointwise product $f\cdot g=x\mapsto f(x)g(x)$, $f/g=x\mapsto f(x)/g(x)$.
\item For a set $A\subset E,$ the function $\mathbf{1}_{A}$ is the indicator
function of set $A$, i.e. 
\[
\mathbf{1}_{A}(x)=\begin{cases}
1 & \text{if }x\in A,\\
0 & \text{otherwise}.
\end{cases}
\]
We also use the notation $\mathbb{I}\{x\in A\}:=\mathbf{1}_{A}(x)$
when the definition of $A$ is explicit and long.
\item $\mathbf{1}$ used to denote the constant function $x\mapsto\mathbf{1}$,
usage is clear from context.
\item For a given $x$, $\delta_{x}$ is the Dirac measure at $x$: $\delta_{x}(A)={\bf 1}_{A}(x)$.
\item If $(E,\mathscr{E})$ and $(F,\mathscr{F})$ are measurable spaces,
the product measurable space is $(E\times F,\mathscr{E}\otimes\mathscr{F})$
where $\mathscr{E}\otimes\mathscr{F}$ is the product $\sigma$-algebra
$\sigma(\{A\times B:A\in\mathscr{E},B\in\mathscr{F}\})$. If $\mu$
is a measure on $(E,\mathscr{E})$ and $\nu$ a measure on $(F,\mathscr{F})$
then their product measure on $(E\times F,\mathscr{E}\otimes\mathscr{F})$
is $\mu\otimes\nu$ where $(\mu\otimes\nu)(A,B)=\mu(A)\nu(B)$ and
define recursively $\mu^{\otimes n}=\mu^{\otimes(n-1)}\otimes\mu$
for $n\in\mathbb{N}_{*}$.
\item If $\mu$ is a measure on $(E,\mathscr{E})$ then the restriction
of $\mu$ to $C\in\mathscr{E}$ is a measure $\mu_{C}$ on $(E,\mathscr{E})$
satisfying $\mu_{C}(A):=\mu(A\cap C)$ for any $A\in\mathscr{E}$.
\item If $\mu({\rm d}x,{\rm d}y)$ is a probability measure, we write $\mu_{x}$
to refer to a conditional probability measure for $Y$ given $X=x$.
(Polish space)
\item A cycle of two Markov kernels $P:E\times\mathscr{E}\to[0,1]$ and
$Q:E\times\mathscr{E}\to[0,1]$ is the Markov kernel
\[
PQ(x,A)=\int P(x,{\rm d}y)Q(y,A),\qquad x\in E,A\in\mathscr{E}.
\]
\item We adopt the standard conventions for products and sums that for $b<a$
, $\prod_{i=a}^{b}\cdot=1$ and $\sum_{i=a}^{b}\cdot=0$ whatever
the nature of the argument.
\item For $x\in\mathbb{R}$, ${\rm sgn}(x)\in\{-1,0,1\}$ is the sign of
$x$.
\item We define $\mathbb{N}=\{0,1,\ldots,\}$ and $\mathbb{N}_{*}=\{1,2,\ldots\}$.
\item We define $\left\llbracket i,j\right\rrbracket =\{i,i+1,\ldots,j\}$
for integers $i\leq j$, and $\left\llbracket i\right\rrbracket =\left\llbracket 1,i\right\rrbracket $
for $i\in\mathbb{N}_{*}$.
\end{itemize}

\section{Motivating example \label{sec:Motivating-example}}

Assume one is interested in sampling from a probability distribution
$\pi$, defined on some probability space $(\mathsf{Z},\mathscr{Z})$.
A Markov chain Monte Carlo (MCMC) algorithm consists of simulating
a realization $\{Z_{i};i\geq0\}$ of a Markov chain such that 
\[
\mathbb{P}(Z_{n}\in A)\rightarrow\pi(A),\qquad A\in\mathscr{Z},
\]
as $n\rightarrow\infty$ and/or for functions $f\in L_{1}(\mathsf{Z},\pi)$,
\[
\lim_{n\rightarrow\infty}\frac{1}{n}\sum_{i=1}^{n}f(Z_{i})=\pi(f),
\]
One of the fundamental properties required to ensure the above is
that, with $P$ denoting the transition probability of the Markov
chain, $\pi$ is left invariant by $P$. That is, the ``global balance''
condition holds:
\begin{equation}
\int\pi({\rm d}z)P(z,A)=\pi(A),\qquad z\in\mathsf{Z},A\in\mathscr{Z}.\label{eq:global-balance}
\end{equation}
It is very difficult to verify (\ref{eq:global-balance}) directly,
complicating the design of Markov kernels satisfying this property.
A successful approach often consists instead of verifying the stronger,
local property of ``detailed balance'' or $\pi-$reversibility.
\begin{defn}[Reversible Markov kernel]
For a finite measure $\mu$ on $(E,\mathscr{E})$, a Markov kernel
$P:E\times\mathscr{E}\to[0,1]$ is $\mu$-reversible if the measures
$\mu({\rm d}\xi)P(\xi,{\rm d}\xi')$ and $\mu({\rm d}\xi')P(\xi',{\rm d}\xi)$
are identical. That is, if,
\[
\int_{A}\mu({\rm d}\xi)P(\xi,B)=\int_{B}\mu({\rm d}\xi)P(\xi,A),\qquad A,B\in\mathscr{E}.
\]
\end{defn}
It is straightforward to deduce that (\ref{eq:global-balance}) holds
if $P$ is $\pi$-reversible by taking $A=E$ in the definition.
\begin{rem}
The definition of $\mu$-reversibility is equivalent to: for all measurable
$F,G:E\to[0,1]$,
\begin{equation}
\int F(\xi)G(\xi')\mu({\rm d}\xi)P(\xi,{\rm d}\xi')=\int G(\xi)F(\xi')\mu({\rm d}\xi)P(\xi,{\rm d}\xi').\label{eq:rev-equiv}
\end{equation}
In particular, we recover the definition with $F={\bf 1}_{A}$ and
$G={\bf 1}_{B}$, and for the other direction, we use the identity
$\mu({\rm d}\xi)P(\xi,{\rm d}\xi')=\mu({\rm d}\xi')P(\xi',{\rm d}\xi)$.
\end{rem}
Metropolis--Hastings (MH) kernels are a flexible class of $\pi-$reversible
Markov kernels for which simulation of the corresponding Markov chain
can often be implemented on a computer. A textbook derivation is as
follows. Assume that $\mathsf{Z}=\mathbb{R}^{d}$ and let $\{Q(z,\cdot),z\in\mathsf{Z}\}$
be a family of probability distributions on $(\mathsf{Z},\mathscr{Z})$
from which it is easy to sample. Assume for presentational simplicity
that for any $z\in\mathsf{Z}$, $\pi$ and $Q(z,\cdot)$ have strictly
positive densities with respect to the Lebesgue measure, denoted $\varpi$
and $q(z,\cdot)$. The MH kernel defined by $\pi$ and $Q$ is given
by
\[
P(z,{\rm d}z')=\alpha(z,z')q(z,z'){\rm d}z'+s(z)\delta_{z}({\rm d}z'),
\]
where $\alpha_{{\rm }}(z,z')=1\wedge r(z,z')$, $s(z)=1-\int\alpha(z,z')q(z,z'){\rm d}z'$
and 
\[
r(z,z')=\frac{\varpi(z')q(z',z)}{\varpi(z)q(z,z')}.
\]
Letting $\rho(z,z'):=\varpi(z)q(z,z')$, verifying $\pi-$reversibility
can be reduced to checking that for $f,g\colon\mathsf{Z}\rightarrow[0,1]$
\begin{equation}
\int f(z)g(z')\rho(z,z')\alpha(z,z'){\rm d}z{\rm d}z'=\int g(z)f(z')\rho(z,z')\alpha(z,z'){\rm d}z{\rm d}z',\label{eq:reduced-detailed-balance-MH}
\end{equation}
since 
\[
\int f(z)g(z')\varpi(z)s(z)\delta_{z}({\rm d}z'){\rm d}z=\int f(z)g(z)\varpi(z)s(z){\rm d}z=\int g(z)f(z')\varpi(z)s(z)\delta_{z}({\rm d}z'){\rm d}z.
\]
It is a standard exercise to show that $\rho(z,z')\alpha(z,z')=\rho(z',z)\alpha(z',z)$
and conclude that (\ref{eq:reduced-detailed-balance-MH}) holds. We
outline now a less direct way, which however has the benefit of highlighting
important generic properties required. 

Define $\xi:=(z,z')$, ${\rm d}\xi={\rm d}z{\rm d}z'$, $\phi(z,z')=(z',z)$
and $F_{0}(z,z')=f(z)$ and $G_{0}(z,z')=g(z)$ then (\ref{eq:reduced-detailed-balance-MH})
can be re-expressed as
\begin{equation}
\int F_{0}(\xi)G_{0}\circ\phi(\xi)\rho(\xi)\alpha(\xi){\rm d}\xi=\int F_{0}\circ\phi(\xi)G_{0}(\xi)\rho(\xi)\alpha(\xi){\rm d}\xi.\label{eq:DB-with-F0-G0}
\end{equation}
Further notice that the acceptance ratio is of the form $r(\xi)=\rho\circ\phi/\rho\:(\xi)$
and that, using that $\phi\circ\phi={\rm Id}$,
\[
r\circ\phi(\xi)=\frac{\rho\circ\phi}{\rho}\circ\phi(\xi)=\frac{\rho}{\rho\circ\phi}(\xi)=\frac{1}{r}(\xi),
\]
therefore implying
\begin{equation}
r(\xi)\alpha\circ\phi(\xi)=r(\xi)\left[1\wedge r\circ\phi(\xi)\right]=\alpha(\xi).\label{eq:property-accept-ratio-phi}
\end{equation}

We now show that (\ref{eq:DB-with-F0-G0}) holds for any measurable
$F,G:\mathsf{Z}^{2}\to[0,1]$ 
\begin{align*}
\int F(\xi)G\circ\phi(\xi)\rho(\xi)\alpha(\xi){\rm d}\xi & =\int F(\xi)G\circ\phi(\xi)\rho(\xi)r(\xi)\alpha\circ\phi(\xi){\rm d}\xi\\
 & =\int F(\xi)G\circ\phi(\xi)\rho\circ\phi(\xi)\alpha\circ\phi(\xi){\rm d}\xi\\
 & =\int F\circ\phi(\xi')G(\xi')\rho(\xi')\alpha(\xi'){\rm d}\xi',
\end{align*}
where we have used $\alpha(\xi)=r(\xi)\alpha\circ\phi(\xi)$, $r(\xi)=\rho\circ\phi/\rho(\xi)$
the change of variable $\xi'=\phi(\xi)$ and the fact that $\phi$
is an involution with Jacobian $\left|{\rm det}\phi'(\xi)\right|=1$
(see Theorem~\ref{thm:cov-leb}). This therefore implies (\ref{eq:reduced-detailed-balance-MH})
and in turn that $P$ is $\pi-$reversible. In fact, letting $\mu({\rm d}\xi):=\rho(\xi){\rm d}\xi$,
we notice that this establishes $\mu-$reversibility of an MH kernel
targeting the extended probability distribution $\mu$.

This presentation has the advantage of highlighting a set of generic
properties sufficient to establish $\pi-$reversibility:
\begin{enumerate}
\item the distribution $\pi$ is a marginal of a probability distribution
$\mu$,
\item the proposed state is obtained by applying an involution $\phi$ to
$\xi$,
\item it holds that $\alpha(\xi)\mu({\rm d}\xi)=\alpha\circ\phi(\xi)\mu^{\phi}({\rm d}\xi)$
with $\mu^{\phi}$ the probability distribution of $\xi'=\phi^{-1}(\xi)=\phi(\xi)$,
\end{enumerate}
suggesting that more general choices of $\mu,\phi$ and $\alpha$
can also define $\pi-$reversible Markov kernels. It can be shown
(Theorem~\ref{sec:Beyond-reversibility-and}) that the first two
properties automatically imply the mathematical existence of $\alpha$
such that the third property holds, highlighting the fundamental rôle
played by the involutory nature of $\phi$. Practical implementation
of the algorithm requires two additional properties of $\mu$: the
existence of a tractable probability density to compute $\alpha$
and ease of sampling from the conditional distribution in $\mu({\rm d}\xi)=\pi({\rm d}\xi_{0})\mu_{\xi_{0}}({\rm d}\xi_{-0})$.

The clear benefit of this approach is that establishing correctness
becomes a purely mechanical, or ``algebraic'', exercise, therefore
improving clarity of arguments and facilitating communication.

\section{General scenario}

In order to gain generality and clarify we will appeal to a very small
number of standard measure theoretical notions and results related
to change of variables and Radon--Nykodim derivatives. Although it
is always a good idea to check the proof of classical results, there
is no need to do so in order to understand the content of this manuscript.
\begin{defn}[Pushforward]
\label{def:pushforward}Let $\mu$ be a measure on $(E,\mathscr{E})$
and $\varphi:(E,\mathscr{E})\to(F,\mathscr{F})$ a measurable function.
The pushforward of $\mu$ by $\varphi$ is defined by 
\[
\mu^{\varphi}(A)=\mu(\varphi^{-1}(A)),\qquad A\in\mathscr{F},
\]
where $\varphi^{-1}(A)=\{x\in E:\varphi(x)\in A\}$ is the preimage
of $A$ under $\varphi$.
\end{defn}
For example, if $\mu$ is a probability distribution then $\mu^{\varphi}$
is the probability measure associated with $\varphi(X)$ when $X\sim\mu$.
\begin{defn}[Dominating and equivalent measures]
 For two measures $\mu$ and $\nu$ on the same measurable space
$(E,\mathscr{E})$, 
\begin{enumerate}
\item $\mu$ is said to dominate $\nu$ if for all measurable $A\in\mathscr{E}$,
$\nu(A)>0\Rightarrow\mu(A)>0$ -- this is denoted $\mu\gg\nu$. 
\item $\mu$ and $\nu$ are equivalent, written $\mu\equiv\nu$, if $\mu\gg\nu$
and $\nu\gg\mu$.
\end{enumerate}
\end{defn}
We will need the notion of Radon-Nikodym derivative:
\begin{thm}[Radon--Nikodym]
\label{thm:Radon-Nikodym}Let $\mu$ and $\nu$ be $\sigma$-finite
measures on $(E,\mathscr{E})$. Then $\nu\ll\mu$ if and only if there
exists an essentially unique, measurable, non-negative function $f$
such that
\[
\int_{A}f(\xi)\mu({\rm d}\xi)=\nu(A),\qquad A\in\mathscr{E}.
\]
Therefore we can view ${\rm d}\nu/{\rm d}\mu:=f$ as the density of
$\nu$ w.r.t $\mu$ and in particular if $g$ is integrable w.r.t.
$\nu$ then
\[
\int g(\xi)\frac{{\rm d}\nu}{{\rm d}\mu}(\xi)\mu({\rm d}\xi)=\int g(\xi)\nu({\rm d}\xi).
\]
\end{thm}
This is covered by \citet[Theorems 32.2 \& 16.11]{billingsley1995probability}.

If $\mu$ is a measure and $f$ a non-negative, measurable function
then $\mu\cdot f$ is the measure $(\mu\cdot f)(A)=\int{\bf 1}_{A}(x)f(x)\mu({\rm d}x)$,
i.e. the measure $\nu=\mu\cdot f$ such that the Radon--Nikodym derivative
of ${\rm d}\nu/{\rm d}\mu=f$.
\begin{thm}[Change of variables]
\label{thm:change-of-variables}A function $f:F\to\mathbb{R}$ is
integrable w.r.t. $\mu^{\varphi}$ if and only if $f\circ\varphi$
is integrable w.r.t. $\mu$, in which case
\begin{equation}
\int_{F}f(\xi)\mu^{\varphi}({\rm d}\xi)=\int_{E}f\circ\varphi(\xi)\mu({\rm d}\xi).\label{eq:change-of-variable}
\end{equation}
\end{thm}
This can be found in \citet[Theorem 16.13]{billingsley1995probability}.

\subsection{An abstract result}

The following result is central to the design of MH-based MCMC, formalizes
the observations made in Section~\ref{sec:Motivating-example} and
generalizes parts of \citet[Proposition 1 and Theorem 2]{tierney1998},
concerned with the specific involution $\phi(z,z')=(z',z)$ and a
particular form of distribution $\mu$. We do not pursue necessity
conditions here, to keep the presentation brief and focused on practical
consequences: \citet{tierney1998} discusses such issues, while \citet{thin:2020b}
revisits these issues in a particle nonreversible setup (see Section
\ref{sec:Beyond-reversibility-and}). The proof can be found in Appendix
\ref{sec:appendix-some-proofs}. This result mirrors \citet[Proposition 2]{Andrieu2019}.
\begin{thm}
\label{thm:invo-rev} Let $\mu$ be a finite measure on $(E,\mathscr{E})$,
$\phi:E\to E$ an involution. Then 
\begin{enumerate}
\item there exists a set $S=S(\mu,\mu^{\phi})\in\mathscr{E}$ such that
\begin{enumerate}
\item $\phi(S)=S$,
\item with $\mu_{S}(A):=\mu(A\cap S)$ for any $A\in\mathscr{E}$ we have
$\mu_{S}^{\phi}\equiv\mu_{S}$,
\item $\mu$ and $\mu^{\phi}$ are mutually singular on $S^{\complement}$,
i.e. there exist sets $A,B\in\mathscr{E}$ such that $A\cap B=\emptyset$,
$A\cup B=S^{\complement}$ and $\mu(A)=\mu^{\phi}(B)=0$.
\end{enumerate}
\item \label{enu:invo-rev-part2}defining for $\xi\in\Xi$, 
\begin{equation}
r(\xi):=\begin{cases}
{\rm d}\mu_{S}^{\phi}/{\rm d}\mu_{S}(\xi) & \xi\in S,\\
0 & \text{otherwise},
\end{cases}\label{eq:r-generic-1}
\end{equation}
and letting $a:[0,\infty)\to[0,1]$ such that
\[
a(r)=\begin{cases}
0 & r=0\\
ra(1/r) & r>0
\end{cases},
\]
we have that,
\begin{enumerate}
\item \label{enu:invo-rev-alpha-decomp}for $\xi\in\Xi$,
\[
\alpha(\xi):=a\circ r(\xi)=\begin{cases}
r(\xi)\cdot\alpha\circ\phi(\xi) & \xi\in S,\\
0 & \text{otherwise},
\end{cases}
\]
 
\item for any measurable $F,G:E\to[0,1]$,
\[
\int F(\xi)G\circ\phi(\xi)\alpha(\xi)\mu({\rm d}\xi)=\int F\circ\phi(\xi)G(\xi)\alpha(\xi)\mu({\rm d}\xi),
\]
\item the Markov kernel $\Pi$ defined by
\[
\Pi(\xi,\{\phi(\xi)\})=\alpha(\xi)=1-\Pi(\xi,\{\xi\}),
\]
is $\mu$-reversible.
\end{enumerate}
\end{enumerate}
\end{thm}
\begin{rem}
The condition on $a$ is satisfied by $a(r)=1\wedge r$ (corresponding
to the Metropolis--Hastings acceptance probability), and also $a(r)=r/(1+r)$
(Barker's acceptance probability; see Example~\ref{eg:mh-barker-2}),
therefore ensuring the existence of $\Pi$ and $P$.
\end{rem}
In practice one is interested in the component $\xi_{0}$ of $\mu$,
which is distributed according to $\pi$. In fact, the Markov kernel
$\Pi$ in Theorem~\ref{thm:invo-rev} can be used to define a $\pi$-invariant
Markov kernel $P$. The proof can be found in Appendix~\ref{sec:appendix-some-proofs}.
\begin{prop}
\label{cor:rev-coordinate} Let $\pi$ be a probability distribution
on $(\mathsf{Z},\mathscr{Z})$ and let $\mu$ be a probability distribution
on $(E,\mathscr{E})$ such that 
\[
\mu({\rm d}\xi):=\pi({\rm d}\xi_{0})\mu_{\xi_{0}}({\rm d}\xi_{-0}),
\]
where $\mu_{\xi_{0}}$ denotes the conditional distribution of $\xi_{-0}$
given $\xi_{0}$ under $\mu$. Then the Markov kernel 
\[
P(\xi_{0},A):=\int{\bf 1}_{A}(\xi'_{0})\mu_{\xi_{0}}({\rm d}\xi_{-0})\Pi(\xi;{\rm d}\xi'),\qquad A\in\mathscr{Z},
\]
is $\pi$-reversible. 
\end{prop}
An algorithmic description of $P$ is given in Alg. \ref{alg:RPhi}
highlighting the practical requirement that sampling from $\mu_{\xi_{0}}(\cdot)$
for $\xi_{0}\in\mathsf{Z}$ should be tractable.

\begin{algorithm}
\caption{\label{alg:RPhi}To sample from $P(\xi_{0},\cdot)$}

\begin{enumerate}
\item Given $\xi_{0}$ , sample $\xi_{-0}\sim\mu_{\xi_{0}}$,
\item Compute $\xi'=\phi(\xi)$,
\item With probability $\alpha(\xi)$ return $\xi_{0}'$, otherwise return
$\xi_{0}$.
\end{enumerate}
\end{algorithm}

The implication of these results should be clear. If sampling from
$\pi$ is of interest, any choice of $\mu$ of the form
\begin{equation}
\mu({\rm d}\xi)=\pi({\rm d}\xi_{0})\mu_{\xi_{0}}({\rm d}\xi_{-0}),\label{eq:mu-to-sample-pi}
\end{equation}
together with an involution $\phi$ and an acceptance function $a$
defines a $\pi-$reversible Markov kernel/chain. It turns out that
all MH-type kernels we are aware of, including advanced and complex
implementations, can be described and immediately justified using
this framework. 
\begin{rem}
\label{rem:non-unique}The framework specified is very flexible: to
define a $\pi$-reversible Markov kernel $P$, whose simulation is
described in Algorithm~\ref{alg:RPhi}, it is sufficient to define
a triple $(\mu,\phi,a)$ such that $\pi$ is the $\xi_{0}$-marginal
of $\mu$. This is analogous to the definition of a traditional Metropolis--Hastings
kernel via the choice $(\pi,Q)$ in Section~\ref{sec:Motivating-example}.
Importantly the nature of $\xi_{-0}$ is \emph{a priori} arbitrary
and does not have to coincide with that of $\xi_{0}$, therefore providing
great freedom. In general, the association is not unique: there are
several $(\mu,\phi,a)$ triples corresponding to the same Markov kernel
$P$. In the sequel we will focus primarily on the the measure-involution
pair $(\mu,\phi)$, since the choice of $a$ can be taken independently
of the choice of $(\mu,\phi)$ from a theoretical perspective.
\end{rem}
In the sequel we will consider Markov kernels $\Pi$ as in Theorem~\ref{thm:invo-rev},
or derivatives such as $P$ in Proposition~\ref{cor:rev-coordinate}
as Metropolis--Hastings type kernels. 
\begin{rem}
\label{rmk:use-maire-douc-olsson}In the context of Proposition~\ref{cor:rev-coordinate}
it is natural to ask whether theoretical properties, such as optimality
in terms of optimal variance of $\Pi$ translate into optimality for
$P$. The answer is yes and follows by application of the results
of \citet{maire2014comparison}, later extended in \citet{Andrieu2019}
to the nonreversible scenario treated in Section \ref{sec:Beyond-reversibility-and}.
\end{rem}
We now provide examples of commonly used Markov kernels, which can
be recognized by the particular form of $\mu$ and $\phi$, and for
which expressions of the corresponding acceptance ratios is left to
Section~\ref{subsec:Densities-and-the}. This highlights the fact
that the acceptance ratio is a function depending only on $\mu$ and
$\phi$.
\begin{example}
\label{ex:MH-textbook}The textbook presentation of the MH kernel
considered in the introduction corresponds to the choice of a family
of conditional probability distributions $\{\mu_{z}(\cdot)=Q(z,\cdot),z\in\mathsf{Z}\}$
on $(\mathsf{Z},\mathscr{Z})$, $\xi=(z,z')\in E=\mathsf{Z}\times\mathsf{Z}$,
$\phi(z,z')=(z',z)$ and $a=r\mapsto1\wedge r$.
\end{example}
\begin{example}
\label{eg:rwm-1}The Random Walk Metropolis (RWM) can be thought of
as corresponding to the choice $\{\mu_{z}(\cdot)=\kappa(\cdot),z\in\mathsf{Z}\}$
for some probability distribution $\kappa$ on $(\mathsf{Z},\mathscr{Z})$,
$\xi=(z,v)$ and $\phi(\xi)=(z+v,-v)$. Alternatively, one may express
the RWM as a special case of Example~\ref{ex:MH-textbook} so that
$\{\mu_{z}(\cdot)=Q(\cdot),z\in\mathsf{Z}\}$ , $\xi=(z,z')$ and
$\phi(z,z')=(z',z)$.
\end{example}
\begin{example}[Metropolis--Hastings, Barker, etc.]
\label{eg:mh-barker-2}Let $\mu$ be as in Example~\ref{ex:MH-textbook},
and let $\xi=(z,z')$ with $\xi_{0}=z$. Then Alg~\ref{alg:RPhi}
corresponds to simulating from the Metropolis--Hastings (resp. Barker)
kernel when $\alpha=a\circ r$, with $a(v)=1\wedge v$ (resp. $a(v)=1/(1+v)$).
This corresponds to the presentation adopted by \citet{tierney1998}
and commonly adapted in the literature.
\end{example}
The requirement that $\phi$ be an involution may appear restrictive,
but in fact for a given invertible function one can define a corresponding
involution by extending the space.
\begin{rem}
\label{rem:invertible-to-involution}Let $\mu$ be a measure admitting
$\pi$ as a marginal and $\phi:E\to E$ be invertible, but not an
involution. Then $(\mu_{0},\phi_{0})$ is a corresponding measure-involution
pair, where $\mu_{0}({\rm d}\xi,{\rm d}v):=\mu({\rm d}\xi)\mathbb{I}\{v\in\{-1,1\}\}/2$
and $\phi_{0}(\xi,v):=(\phi^{v}(\xi),-v)$ on $E_{0}:=E\times\{-1,1\}$.
Since $\mu_{0}$ admits $\mu$ as a marginal, it also admits $\pi$
as a marginal.
\end{rem}
\begin{example}[Ordered overrelaxation \citep{neal1998suppressing}]
A Gibbs sampler can be thought of as a MH update where conditional
distributions of the target distribution $\pi$ on $(\mathsf{Z},\mathscr{Z})$
are used in the proposal mechanism. To fix ideas assume $\mathsf{Z}=\mathsf{X}\times\mathsf{Y}$
where $\mathsf{Y}\subset\mathbb{R}$ and let $\eta$ be one such conditional
distribution on $(\mathsf{Y},\mathscr{Y})$ from which sampling is
tractable. The goal of the method is to develop a numerical implementation
the following remark. Let $(x,y)\in\mathsf{Z}$ and let $F_{\eta}$
be the cumulative distribution function (cdf) corresponding to $\eta$,
where $x$ is implicit. Then $y'=F_{\eta}^{-1}\big(1-F_{\eta}(y)\big)$
is antithetic to $y$--in fact for $u\sim{\rm Uniform}(0,1)$ the
pair $\big(F_{\eta}^{-1}(u),F_{\eta}^{-1}\big(1-u)\big)$ is the lower
bound in the Fréchet class of bivariate distributions of marginals
$\eta$. The numerical approximation of this remark exploits the link
between empirical cdf and order statistics. One can sample multiple
times independently from $\eta$, leading to the probability distribution,
for $n\in\mathbb{N_{*}}$, on $(\mathsf{Z}\times\mathsf{Y}^{n},\mathsf{Z}\otimes\mathsf{Y}^{\otimes n})$
\[
\mu=\pi\otimes\eta^{\otimes n}.
\]
Let $z=(x,y_{0})$ and let $\sigma\colon\llbracket0,n\rrbracket\rightarrow\llbracket0,n\rrbracket$
be the $\xi:=(z,y_{1},\ldots,y_{n})$-dependent permutation such that
\[
y_{\sigma(0)}\leq\cdots\leq y_{\sigma(n)}
\]
and let $r\in\llbracket0,n\rrbracket$ be the integer such that $\sigma(r)=0$
i.e. $y_{0}$ is the $r-$th rank order statistic among $y_{0},y_{1},\ldots,y_{n}$.
Now we consider the following involution, for $r\in\llbracket2,n-1\rrbracket$
\[
\phi(z,y_{1}\ldots,y_{n})=(x,y_{\sigma(n-r)},y_{1},\ldots,y_{\sigma(n-r)-1},y_{0},y_{\sigma(n-r)+1},\ldots,y_{n}),
\]
with straightforward adaptation if $r\in\{0,1,n\}$. It should be
clear, from the exchangeability conditional upon $x\in\mathsf{X}$,
that for $\xi\in S\big(\mu,\mu^{\phi}\big)$, $r(\xi)=1$. One can
naturally replace $\eta$ with a proposal distribution of our choosing,
but the acceptance ratio is then not identically equal to $1$. 
\end{example}
Adopting this point of view makes establishing reversibility routine,
even in complex scenarios. However practical implementation of the
update requires an explicit expression for the acceptance ratio $r$
in (\ref{eq:r-generic-1}), not provided by the results above.
\begin{rem}
\label{rem:lazy}Alg. \ref{alg:RPhi} is conceptually simple, but
in practice it may be expedient to avoid a direct implementation.
What is actually required to simulate from $P(\xi_{0},\cdot)$ is
to sample a ${\rm Bernoulli}(\alpha(\xi))$ random variable, where
$\xi_{-0}\sim\mu_{\xi_{0}}$ and to compute $\phi(\xi)_{0}$. In particular,
it may not be necessary to simulate or store $\xi_{-0}$ in its entirety
to perform these task, e.g. when $\xi_{-0}$ is large or even infinite-dimensional.
Some examples are provided in Section \ref{sec:Beyond-reversibility-and}.

We will primarily focus on Alg.~\ref{alg:RPhi} in the sequel. Hence,
for examples and applications of this framework we will identify an
appropriate $(\mu,\phi)$, hence defining $\Pi$ in Theorem~\ref{thm:invo-rev}
up to the choice of $a$. The corresponding $\pi$-reversible Markov
kernel is then defined by $P$ in Proposition~\ref{cor:rev-coordinate}.
There are, of course, other $\mu$-invariant kernels that can be constructed
using $\Pi$. For example, letting $R$ define the refreshment kernel
\[
R(\xi,{\rm d}\xi')=\delta_{\xi_{0}}({\rm d}\xi_{0}')\mu_{\xi_{0}}({\rm d}\xi_{-0}'),
\]
Alg.~\ref{alg:RPhi} corresponds to tracking the $\xi_{0}$-coordinate
of $R\Pi(\xi_{0},\cdot)$. One could instead define a $\mu$-invariant
kernel as $\gamma R+(1-\gamma)\Pi$ for some $\gamma\in(0,1)$. Even
more generally, one could replace $R$ with another Markov kernel
that only leaves the conditional distribution $\mu_{\xi_{0}}$ invariant.
The cycle $R\Pi$ is then $\mu$-invariant and would sometimes be
referred to as a Metropolis-within-Gibbs (MwG) kernel, although we
note that in this case the corresponding $\xi_{0}$-coordinate of
the $\mu$-invariant Markov chain would in general not be Markov.
More generally we will refer to an algorithm involving a mixture (``random-scan'')
or cycle (``deterministic scan'') of kernels targetting the same
distribution as a MwG, a widely accepted misnomer.
\end{rem}

\subsection{Densities and the acceptance ratio\label{subsec:Densities-and-the}}

In order to compute the acceptance ratio $r$ in Theorem~\ref{thm:invo-rev},
one must identify $S$ and have an expression for ${\rm d}\mu_{S}^{\phi}/{\rm d}\mu_{S}$.
We show below how to phrase these objects in terms of a density $\rho={\rm d}\mu/{\rm d}\lambda$,
where $\lambda$ is an appropriate reference measure. Such a density
is often available \emph{a priori} in practice.
\begin{prop}
\label{prop:r-density}Let $\mu$ be a finite measure on $(E,\mathscr{E})$,
$\phi:E\to E$ an involution, let $\lambda\gg\mu$ be a $\sigma$-finite
measure satisfying $\lambda\equiv\lambda^{\phi}$ and let $\rho={\rm d}\mu/{\rm d}\lambda$.
Then we can take $S=S(\mu,\mu^{\phi})$ to be $S=\{\xi:\rho(\xi)\wedge\rho\circ\phi(\xi)>0\}$
and
\begin{equation}
r(\xi)=\begin{cases}
\frac{\rho\circ\phi}{\rho}(\xi)\frac{{\rm d}\lambda^{\phi}}{{\rm d}\lambda}(\xi) & \xi\in S,\\
0 & \text{otherwise},
\end{cases}\label{eq:r-rho}
\end{equation}
in Theorem~\ref{thm:invo-rev}.
\end{prop}
The proof can be found in Appendix \ref{subsec:app-proofs-measure}.
In many situations $\lambda$ will be the Lebesgue or counting measure,
but can also be a product of these, or an infinite-dimensional probability
measure such as a Gaussian measure \citet{hairer2014spectral} or
the law of a Markov chain (this is treated in Subsection~\ref{subsec:doubly-infinite-general}).
Computing (\ref{eq:r-rho}) involves additionally computing the density
${\rm d}\lambda^{\phi}/{\rm d}\lambda$.
\begin{rem}
If in Proposition~\ref{prop:r-density} $\lambda$ is invariant under
$\phi$, i.e. $\lambda=\lambda^{\phi}$ then $r=\rho\circ\phi/\rho$.
In theory, it is always possible to find a reference measure invariant
under $\phi$, e.g. one could instead of $\lambda$ take $\lambda_{0}:=\lambda+\lambda^{\phi}$
or even $\lambda_{0}:=\mu+\mu^{\phi}$, which underpins the proof
of Theorem~\ref{thm:invo-rev}. However, it may not be straightforward
or natural to compute the density ${\rm d}\mu/{\rm d}\lambda_{0}$,
while there is often a natural choice of $\lambda$ for which ${\rm d}\mu/{\rm d}\lambda$
can be computed. 

A standard scenario is when $\lambda$ is the Lebesgue measure on
$E=\mathbb{R}^{d}$ and $\phi$ is a diffeomorphism, in which case
${\rm d}\lambda^{\phi}/{\rm d}\lambda$ corresponds to the absolute
value of the determinant of the Jacobian, since then for any $\lambda$-integrable
$f$ (see Theorem~\ref{subsec:app-proofs-measure} in Appendix~\ref{sec:Measure-theory-tools})
\[
\int f\circ\phi(\xi)\left|{\rm det}\phi'(\xi)\right|\lambda({\rm d}\xi)=\int f(\xi)\lambda({\rm d}\xi)=\int f\circ\phi(\xi)\lambda^{\phi}({\rm d}\xi),
\]
while for an arbitrary, measurable, non-negative $g:E\to\mathbb{R}$
we can take $f=g\circ\phi^{-1}$ to obtain $g=f\circ\phi$ and hence,
\[
\int g(\xi)\left|{\rm det}\phi'(\xi)\right|\lambda({\rm d}\xi)=\int g(\xi)\lambda^{\phi}({\rm d}\xi).
\]
The example of the introduction corresponds to this scenario, but
where in addition $\phi$ is an involution and the reference measure
is invariant under $\phi$. 
\end{rem}
\begin{rem}
\label{rem:d-lambda-phi-common}There are several ways one can determine
${\rm d}\lambda^{\phi}/{\rm d}\lambda$ in common situations. For
example:
\begin{enumerate}
\item Let $E=\mathbb{R}^{d}$ with $\xi=(z_{1},\ldots,z_{d})$, and $\phi$
be an involution that permutes its input, i.e. $\phi(z_{1},\ldots,z_{d})=(z_{\sigma(1)},\ldots,z_{\sigma(d)})$
for some permutation $\sigma$ of $\{1,\ldots,d\}$. Then since $\phi'(\xi)$
is the corresponding permutation matrix and all permutations have
a determinant in $\{-1,1\}$, we obtain $\left|{\rm det}\phi'(\xi)\right|=1$.
So if $\lambda$ is the Lebesgue measure on $\mathbb{R}^{d}$ then
$\lambda^{\phi}=\lambda$.
\item Let $\mu$ be a measure with countable support $\mathsf{X}$, and
let $\lambda$ be the counting measure on $E=\mathsf{X}\cup\phi(\mathsf{X})=\mathsf{X}\cup\{\phi(x):x\in\mathsf{X}\}$.
Then for an arbitrary, measurable $A\subseteq E$ we have $\lambda^{\phi}(A)=\lambda(\phi^{-1}(A))=|A|=\lambda(A)$
since $\phi$ is an involution. Hence $\lambda=\lambda^{\phi}$ so
${\rm d}\lambda^{\phi}/{\rm d}\lambda=1$.
\end{enumerate}
\end{rem}
In some of our applications, $\mu$ has continuous and discrete components,
and a density with respect to a product of a Lebesgue measure and
a counting measure. When the involution for the discrete component
does not depend on the continuous component, we have the following
result.
\begin{lem}
\label{lemnu:leb-count-combi}Let $\lambda_{X}$ be the Lebesgue measure
on $\mathsf{X}$, $\lambda_{Y}$ the counting measure on $\mathsf{Y}$
and $g:\mathsf{Y}\to\mathsf{Y}$ be an involution with $g(\mathsf{Y})\subseteq\mathsf{Y}$.
Let $f:\mathsf{X}\times\mathsf{Y}\to\mathsf{X}$ be a function such
that 
\[
\phi(x,y)=(f(x,y),g(y)),
\]
is an involution. Then ${\rm d}\lambda^{\phi}/{\rm d}\lambda=\left|{\rm det}f_{y}'(x)\right|$.
\end{lem}
We are now in a position to provide expressions for the acceptance
ratios in Examples~\ref{ex:MH-textbook}--\ref{eg:rwm-1}.
\begin{example}[Metropolis--Hastings acceptance ratio]
\label{eg:mh-acceptance-ratio}Let $\pi$ and $\{Q(z,\cdot),z\in\mathsf{Z}\}$
be probability measures on $(\mathsf{Z},\mathscr{Z})$ such that with
$\nu$ the Lebesgue or counting measure we have $\nu\gg\pi$ and $\nu\gg Q(z,\cdot)$
for each $z\in\mathsf{Z}$. Let $\varpi(z)={\rm d}\pi/{\rm d}\nu(z)$
and $q(z,z')={\rm d}Q(z,\cdot)/{\rm d}\nu(z')$ for all $(z,z')\in\mathsf{Z}^{2}$.
With $\xi=(z,z')$ we let $\mu({\rm d}\xi)=\pi({\rm d}z)Q(z,{\rm d}z')$,
and $\phi(z,z')=(z',z)$. Then with $\lambda^{\phi}=\lambda:=\nu\times\nu$
we obtain $\rho(\xi)=\varpi(z)q(z,z')$ and $\rho\circ\phi(\xi)=\varpi(z')q(z',z)$
and the acceptance ratio is, for $\xi\in S(\mu,\mu^{\phi})=\{\xi:\rho(\xi)\wedge\rho\circ\phi(\xi)>0\}$,
\[
r(\xi)=\frac{\rho\circ\phi(\xi)}{\rho(\xi)}=\frac{\varpi(z')q(z',z)}{\varpi(z)q(z,z')}.
\]
\end{example}
\begin{example}[Random walk Metropolis ratio]
The setup is similar to above but we assume that $\mathsf{Z}=\mathbb{R}^{d}$,
$\nu$ is the Lebesgue measure, $q(z,v)=q(v):={\rm d}Q/{\rm d}\nu(v)$
for $(z,v)\in\mathsf{Z}^{2}$ and $q(v)=q(-v)$ for $v\in\mathsf{Z}$.
Here $\lambda=\nu\times\nu$, $\xi=(z,v)\in\mathsf{Z}^{2}$, $\phi(\xi)=(z+v,-v)$
and $\left|{\rm det}\phi'(\xi)\right|=1$, leading to 
\begin{align*}
r(\xi) & =\frac{\rho\circ\phi(\xi)}{\rho(\xi)}=\frac{\varpi(z+v)q(-v)}{\varpi(z)q(v)}=\frac{\varpi(z+v)}{\varpi(z)}.
\end{align*}
\end{example}
It is possible to consider the setting where $\xi_{0}=\xi$ and $\phi$
is an involution with non-unit Jacobian. Such situations are related,
e.g., to the Monte Carlo Markov kernels based on deterministic transformations
proposed by \citet{dutta2014markov}.
\begin{example}
Assume $\rho={\rm d}\mu/{\rm d\lambda}$ with $\{\xi\in E:\rho(\xi)>0\}=(0,1)$
with $\lambda$ the Lebesgue measure on $\mathbb{R}$ and let $\phi(\xi)=1/(2\xi)$.
One can deduce that $\lambda$ and $\lambda^{\phi}$ are equivalent
with ${\rm d}\lambda^{\phi}/{\rm d}\lambda(\xi)=\left|\phi'(\xi)\right|=1/(2\xi^{2})$.
We obtain $S=S(\mu,\mu^{\phi})=\{\xi\in E:\rho(\xi)\wedge\rho\circ\phi(\xi)>0\}=(0,1)\cap\phi^{-1}(0,1)=[1/2,1)$.
Therefore $r(\xi)=\rho\circ\phi(\xi)/(\rho(\xi)2\xi^{2})$ for $\xi\in S$
and $r(\xi)=0$ otherwise.
\end{example}
\begin{example}
Consider $\pi$ a probability measure on $\mathbb{R}$ dominated by
Lebesgue and $\varphi(x)=x^{3}$, which is invertible but not an involution.
Then following Remark~\ref{rem:invertible-to-involution}, we can
extend the space to $E=\mathbb{R}\times\{-1,1\}$ and define $\xi=(x,k)$,
$\mu({\rm d}\xi)=\pi({\rm d}x)\mathbb{I}(k\in\{-1,1\})/2$, $\lambda$
to be the product of the Lebesgue measure and the counting measure,
and $\phi(\xi)=\phi(x,k)=(\varphi^{k}(x),-k)$. Following Lemma~\ref{lemnu:leb-count-combi},
we obtain
\[
\frac{{\rm d}\lambda^{\phi}}{{\rm d}\lambda}(x,k)=\left|{\rm det}(\varphi^{k})'(x)\right|=\begin{cases}
3x^{2} & k=1,\\
\frac{1}{3}\left|x\right|^{-2/3} & k=-1.
\end{cases}
\]
\end{example}
A slightly more general version of Example~\ref{eg:mh-acceptance-ratio}
above can be used when $Q$ is reversible w.r.t. some measure.
\begin{example}
Let $\pi$ be probability measures on $(\mathsf{Z},\mathscr{Z})$,
$\nu$ be a reference measure such that $\pi\ll\nu$ and assume that
$Q$ is $\nu$-reversible. Then with $\xi=(z,z')$ and $\varpi={\rm d}\pi/{\rm d}\nu$,
\[
\mu({\rm d}\xi)=\pi({\rm d}z)Q(z,{\rm d}z')=\varpi(z)\nu({\rm d}z)Q(z,{\rm d}z'),
\]
that is $\rho(\xi)=\varpi(\xi_{0})$ with $\lambda({\rm d}\xi)=\nu({\rm d}\xi_{0})Q(\xi_{0},{\rm d}\xi_{-0})$
and by assumption $\lambda^{\phi}=\lambda$ for $\phi(z,z')=(z',z)$
for $(z,z')\in\mathsf{Z}^{2}$. Therefore
\[
r(\xi)=\frac{\rho\circ\phi}{\rho}(\xi)=\frac{\varpi(z')}{\varpi(z)},\qquad\xi\in S=\{(z,z')\in\mathsf{Z}^{2}:\varpi(z)\wedge\varpi(z')>0\}.
\]

In many common RWM kernels, $\nu$is the Lebesgue (resp. counting)
measure on a continuous (resp. discrete) state space.
\end{example}
\begin{example}
(Simplified Neal tempering) Let $\pi$ be a multimodal distribution
on $(\mathsf{Z},\mathscr{Z})$. A strategy proposed by \citet{neal1996sampling}
to mitigate the effect of multimodality on consists of using an instrumental
distribution $\tilde{\pi}\equiv\pi$ also defined on $(\mathsf{Z},\mathscr{Z})$,
related to $\pi$ but less multimodal, to improve the rate of moves
between modes of $\pi$. More specifically define 
\[
\mu({\rm d}\xi)=\pi({\rm d}\xi_{0})Q(\xi_{0},{\rm d}\xi_{1})\tilde{Q}(\xi_{1},{\rm d}\xi_{2})Q(\xi_{2},{\rm d}\xi_{3}),
\]
where $\tilde{Q}$ (resp. $Q$) is $\tilde{\pi}-$reversible (resp.
$\pi-$reversible) and consider the involution on $E=\mathsf{Z}^{4}$
such that $\phi(\xi_{0},\xi_{1},\xi_{2},\xi_{3})=(\xi_{3},\xi_{2},\xi_{1},\xi_{0})$.
Using these properties, we obtain
\begin{align*}
\mu^{\phi}({\rm d}\xi) & =\pi({\rm d}\xi_{3})Q(\xi_{3},{\rm d}\xi_{2})\tilde{Q}(\xi_{2},{\rm d}\xi_{1})Q(\xi_{1},{\rm d}\xi_{0})\\
 & =\pi({\rm d}\xi_{2})Q(\xi_{2},{\rm d}\xi_{3})\tilde{Q}(\xi_{2},{\rm d}\xi_{1})Q(\xi_{1},{\rm d}\xi_{0})\\
 & =\frac{{\rm d}\pi}{{\rm d}\tilde{\pi}}(\xi_{2})Q(\xi_{2},{\rm d}\xi_{3})\tilde{\pi}({\rm d}\xi_{2})\tilde{Q}(\xi_{2},{\rm d}\xi_{1})Q(\xi_{1},{\rm d}\xi_{0})\\
 & =\frac{{\rm d}\pi}{{\rm d}\tilde{\pi}}(\xi_{2})Q(\xi_{2},{\rm d}\xi_{3})\tilde{Q}(\xi_{1},{\rm d}\xi_{2})\tilde{\pi}({\rm d}\xi_{1})Q(\xi_{1},{\rm d}\xi_{0})\\
 & =\frac{{\rm d}\pi}{{\rm d}\tilde{\pi}}(\xi_{2})\frac{{\rm d}\tilde{\pi}}{{\rm d}\pi}(\xi_{1})Q(\xi_{2},{\rm d}\xi_{3})\tilde{Q}(\xi_{1},{\rm d}\xi_{2})\pi({\rm d}\xi_{1})Q(\xi_{1},{\rm d}\xi_{0})\\
 & =\frac{{\rm d}\pi}{{\rm d}\tilde{\pi}}(\xi_{2})\frac{{\rm d}\tilde{\pi}}{{\rm d}\pi}(\xi_{1})Q(\xi_{2},{\rm d}\xi_{3})\tilde{Q}(\xi_{1},{\rm d}\xi_{2})Q(\xi_{0},{\rm d}\xi_{1})\pi({\rm d}\xi_{0})\\
 & =\frac{{\rm d}\pi}{{\rm d}\tilde{\pi}}(\xi_{2})\frac{{\rm d}\tilde{\pi}}{{\rm d}\pi}(\xi_{1})\mu({\rm d}\xi).
\end{align*}

It follows that we can take $S=\{\xi\in E\colon{\rm d}\tilde{\pi}/{\rm d}\pi(\xi_{1})\:{\rm d}\pi/{\rm d}\tilde{\pi}({\rm d}\xi_{2})>0\}$,
and we obtain,
\[
r(\xi)=\frac{{\rm d}\tilde{\pi}}{{\rm d}\pi}(\xi_{1})\frac{{\rm d}\pi}{{\rm d}\tilde{\pi}}(\xi_{2}),\qquad\xi\in S.
\]
In this case, we can think of $\lambda=\mu\equiv\mu^{\phi}=\lambda^{\phi}$,
$\rho\equiv1$ and ${\rm d}\lambda^{\phi}/{\rm d}\lambda=r(\xi)$
on $S$.  In practice, computation of the acceptance ratio may be
facilitated by convenient densities for $\tilde{\pi}$ and $\pi$
with respect to a common dominating measure. The above can be viewed
as the justification for the tempered transitions kernel introduced
by \citet{neal1996sampling}, where several instrumental distributions
are used; these ideas are also related to the methodology in \citet{neal2005taking}.
\end{example}
\begin{example}[Penalty method \citet{ceperley1999penalty}]
In this scenario $\mu\big({\rm d}(z,w)\big)=\mu_{0}({\rm d}z)Q_{z}({\rm d}w)$,
$\xi=(z,w)\in\mathsf{Z}\times\mathsf{W}$ with $\mathsf{W}\subset\mathbb{R}_{+}^{*}$,
$\phi(z,w)=\big(\phi_{0}(z),1/w\big)$ for $\phi_{0}\colon\mathsf{Z}\rightarrow\mathsf{Z}$
an involution and $w\cdot Q_{z}({\rm d}w)=Q_{\phi_{0}(z)}^{1/\cdot}({\rm d}w)$
for $w>0$, where for any $f\colon\mathsf{Z}\rightarrow[0,1]$
\[
\int f(w)Q_{z}^{1/\cdot}({\rm d}w):=\int f\big(w^{-1}\big)Q_{z}({\rm d}w),
\]
therefore implying for $\xi\in S$
\[
r(\xi)=\frac{{\rm d}\pi^{\phi_{0}}}{{\rm d}\pi}(z)\frac{Q_{\phi_{0}(z)}^{1/\cdot}({\rm d}w)}{Q_{z}({\rm d}w)}=\frac{{\rm d}\pi^{\phi_{0}}}{{\rm d}\pi}(z)w.
\]
The motivation for this setup is concerned with the situation where
a noisy version of the acceptance ratio ${\rm d}\pi^{\phi_{0}}/{\rm d}\pi(z)$
is available, where the noise is additive in the log-domain, corresponding
to noisy energies in Physics. The condition on $Q_{z}$ is satisfied
by the random variable $W:=\exp(-\sigma_{z}^{2}/2+\sigma_{z}Z)$ with
$Z\sim\mathcal{N}(0,1)$ for $z\mapsto\sigma_{z}=\sigma_{\phi_{0}(z)}$
because
\begin{align*}
\int f(w)Q_{z}({\rm d}w) & =\int f\big[\exp(-\sigma_{z}^{2}/2+\sigma_{z}x)\big]\mathcal{N}({\rm d}x;0,1)\\
 & =\int f\big[\exp\big(\sigma_{z}^{2}/2-\sigma_{z}(x+\sigma_{z})\big)\big]\mathcal{N}({\rm d}x;0,1)\\
 & =\int f\big[\exp(\sigma_{z}^{2}/2-\sigma_{z}x)\big]\mathcal{N}({\rm d}x;\sigma_{z},1)\\
 & =\int f\big[\exp(\sigma_{\phi_{0}(z)}^{2}/2-\sigma_{\phi_{0}(z)}x)\big]\cdot\exp(-\sigma_{\phi_{0}(z)}^{2}/2+\sigma_{\phi_{0}(z)}x)\mathcal{N}({\rm d}x;0,1)\\
 & =\int f\big(w^{-1}\big)w\cdot Q_{\phi_{0}(z)}({\rm d}w)\\
 & =\int f\big(w\big)w^{-1}\cdot Q_{\phi_{0}(z)}^{1/\cdot}({\rm d}w).
\end{align*}
One can also consider, with $z\mapsto\omega_{z}=\omega_{\phi_{0}(z)}>0$,
\[
Q_{z}({\rm d}w)=\frac{\omega_{z}}{1+\omega_{z}}\delta_{\omega_{z}}({\rm d}w)+\frac{1}{1+\omega_{z}}\delta_{1/\omega_{z}}({\rm d}w),
\]
 because
\begin{align*}
\int f(w)Q_{z}({\rm d}w) & =\frac{\omega_{z}}{1+\omega_{z}}f(\omega_{z})+\frac{1}{1+\omega_{z}}f(\omega_{z}^{-1})\\
 & =\omega_{z}\frac{1}{1+\omega_{z}}f(\omega_{z})+\omega_{z}^{-1}\frac{\omega_{z}}{1+\omega_{z}}f(\omega_{z}^{-1})\\
 & =\int f(w)wQ_{\phi_{0}(z)}^{1/\cdot}({\rm d}w).
\end{align*}
\end{example}
\begin{example}[Reversible jump MCMC \citet{green1995reversible}]
Here we are concerned with the situation where $\mathsf{X}$ is a
disjoint union, for example $\mathsf{X}=\bigsqcup_{i\in\mathbb{N}}\{i\}\times\mathsf{X}_{i}$
with, for $i\in\mathbb{N}$, $(\mathsf{X}_{i},\mathscr{X}_{i})$ a
measurable space and $\mathscr{X}$ a sigma algebra associated to
$\mathsf{X}$; see \citet[214K]{fremlin2000measure} for a construction.
Here the probability distribution of interest is $\pi(i,{\rm d}x_{i})$,
that is for $i\in\mathbb{N}$, $\pi(i,\cdot)\colon\mathscr{X}_{i}\mapsto\mathbb{R}_{+}$
is a finite measure and $\sum_{i=1}^{\infty}\pi(i,\mathsf{X}_{i})=1$.
The idea of \citet{green1995reversible} to circumvent the possibly
differing nature of the $\mathsf{X}_{i}$'s is to introduce the following
space and probability embeddings:
\begin{enumerate}
\item $E:=\bigsqcup_{i,j\in\mathbb{N}}\{(i,j)\}\times\mathsf{X}_{i}\times\mathsf{U}_{ij}$
such that for $(i,j)\in\mathbb{N}^{2}$ there exist measurable bijections
$\mathsf{X}_{i}\times\mathsf{U}_{ij}\rightarrow\mathsf{X}_{j}\times\mathsf{U}_{ji}$
for the measurable sapces $\big(\mathsf{X}_{i}\times\mathsf{U}_{ij},\mathscr{X}_{i}\otimes\mathscr{U}_{ij}\big)$
and $\big(\mathsf{X}_{j}\times\mathsf{U}_{ji},\mathscr{X}_{j}\otimes\mathscr{U}_{ji}\big)$; 
\item for $(i,j)\in\mathbb{N}^{2}$ one chooses mappings $\phi_{ij}=\phi_{ji}^{-1}$
and define $\phi\colon E\rightarrow E$ the $\phi(i,j,x_{i},u_{ij}):=\big(j,i,\phi_{ij}(x_{i},u_{ij})\big)$.
\item the probability distribution $\pi$ is embedded in $\mu(i,j,d(x_{i},u_{ij}))=\pi(i,{\rm d}x_{i})\mu_{i}(j,{\rm d}u_{ij}\mid x_{i})$.
\end{enumerate}
This can be viewed as a natural generalization of Remark~\ref{rem:invertible-to-involution}.
\end{example}
\begin{rem}
In light of Example~\ref{eg:mh-barker-2} and its relation to the
framework in \citet{tierney1998}, it is natural to ask whether the
framework considered here is more powerful in terms of its ability
to express and validate Markov kernels. In fact it is not, but is
perhaps more natural to use since one does not introduce additional
auxiliary variables in $\mu$. In particular, for a given choice of
$\mu$ and $\phi$, one can always embed $(\xi,\phi(\xi))$ in the
extended space $E\times E$ with distribution $\tilde{\mu}({\rm d}x,{\rm d}y)=\mu({\rm d}x)\delta_{\phi(x)}({\rm d}y)$,
and use the involution $\tilde{\phi}(x,y)=(y,x)$. The $\tilde{\mu}$-reversibility
then follows from Theorem~\ref{thm:invo-rev}. For an expression
for the acceptance ratio, it is then convenient to consider the $\tilde{\phi}$-invariant
reference measure $\upsilon=\lambda+\lambda^{\tilde{\phi}}$. We obtain
that ${\rm d}\tilde{\mu}/{\rm d}\upsilon(x)=\rho(x)\left\{ 1+\frac{{\rm d}\lambda^{\phi}}{{\rm d}\lambda}(x)\right\} ^{-1}$,
where $\rho={\rm d}\mu/{\rm d}\lambda$. We obtain that for $x$ in
the same $S=S(\mu,\mu^{\phi})$,
\[
r((x,\phi(x)),(\phi(x),x))=\frac{\frac{{\rm d}\tilde{\mu}^{\tilde{\phi}}}{{\rm d}\upsilon}(x)}{\frac{{\rm d}\tilde{\mu}}{{\rm d}\upsilon}(x)}=\frac{\rho\circ\phi}{\rho}(x)\frac{1+\frac{{\rm d}\lambda^{\phi}}{{\rm d}\lambda}(x)}{1+\frac{{\rm d}\lambda}{{\rm d}\lambda^{\phi}}(x)}=\frac{\rho\circ\phi}{\rho}(x)\frac{{\rm d}\lambda^{\phi}}{{\rm d}\lambda}(x),
\]
as in Proposition~\ref{prop:r-density}.
\end{rem}

\section{Beyond reversibility and standard deterministic proposals\label{sec:Beyond-reversibility-and}}

Reversibility plays a central role in the design of MCMC algorithms
but is not necessarily a desirable property. In fact, it has been
shown that nonreversible Markov chains can converge more quickly in
some cases \citep{diaconis2000analysis}, and their ergodic averages
can have smaller asymptotic variance in comparison to a suitable reversible
counterpart \citep{neal2004improving,sun2010improving,chen2013accelerating,andrieu2016random}.
This can be intuitively attributed to the fact that reversible processes
tend to backtrack and/or move in a diffusive way, suggesting slower
exploration of the target distribution in comparison to nonreversible
processes that move in a more systematic way through the state space.

We discuss here a popular class of nonreversible MH type updates which
can be understood as being the cycle of two $\mu-$reversible Markov
kernels. This type of non reversibility is referred to as $(\mu,\mathfrak{S})-$reversibility
in the literature \citep{Andrieu2019} and was first discussed in
\citet{yaglom1949statistical} as a generalisation of deterministic
time-reversible systems. The necessity for some of the conditions
below is discussed in \citet{thin:2020b}.
\begin{prop}
\label{prop:psi-is-sigma-phi}Let $\mu$ be a probability distribution
on $(E,\mathscr{E})$, $\phi,\sigma\colon E\rightarrow E$ be involutions
with $\sigma$ such that $\mu^{\sigma}=\mu$. Let
\begin{enumerate}
\item $\Pi$ be the $\mu-$reversible Markov kernel using $\phi$ and acceptance
function $a(r)=1\wedge r$,
\item $\mathfrak{S}$ be such that for $\xi\in E$, $\mathfrak{S}(\xi,\{\sigma(\xi)\})=1$
(or for $\xi'\in E$, $\mathfrak{S}(\xi,{\rm d}\xi')=\delta_{\sigma(\xi)}({\rm d}\xi')$),
\item $\lambda\gg\mu$ be such that $\lambda\equiv\lambda^{\phi}$ and $\lambda^{\sigma}=\lambda$.
\end{enumerate}
Let $\psi:=\sigma\circ\phi$ and $\Psi$ such that for $\xi\in E$,
$\Psi(\xi,\{\psi(\xi)\})=1$ (or for $\xi'\in E$, $\Psi(\xi,{\rm d}\xi')=\delta_{\psi(\xi)}({\rm d}\xi')$)
then
\begin{enumerate}
\item \label{enu:norev-prop-psi}$\psi^{-1}=\sigma\circ\psi\circ\sigma$,
$\lambda^{\phi}=\lambda^{\psi^{-1}}$, and $\rho\circ\sigma=\rho$,
\item \label{enu:non-rev-def-kernel}the $\mu-$invariant cycle $\varPi:=\Pi\mathfrak{S}$
is given by
\[
\varPi(\xi,{\rm d}\xi')=a\circ r(\xi)\cdot\Psi\big(\xi,{\rm d}\xi'\big)+[1-a\circ r(\xi)]\mathfrak{S}(\xi,{\rm d}\xi'),
\]
where with $S=\{\xi\in E\colon\rho(\xi)\wedge[\rho\circ\psi(\xi){\rm d}\lambda^{\psi^{-1}}/{\rm d}\lambda(\xi)]>0\}$,
\[
r(\xi)=\begin{cases}
\frac{\rho\circ\psi}{\rho}(\xi)\frac{{\rm d}\lambda^{\psi^{-1}}}{{\rm d}\lambda}(\xi) & \xi\in S,\\
0 & \text{otherwise}.
\end{cases}
\]
\item \label{enu:nonrev-def-muS-reversible}In fact $\varPi$ is $(\mu,\mathfrak{S})-$reversible
(or satisfied the modified or skew detailed balance), that is for
$\xi,\xi'\in E$
\[
\mu({\rm d}\xi)\varPi(\xi,{\rm d}\xi')=\mu({\rm d}\xi')\mathfrak{S}\varPi\mathfrak{S}(\xi',{\rm d}\xi).
\]
\item \label{enu:nonrev-marginal-kernel-piS-rev}Let $\mu({\rm d}\xi):=\pi({\rm d}\xi_{0})\mu_{\xi_{0}}({\rm d}\xi_{-0}),$
where $\mu_{\xi_{0}}$ denotes the conditional distribution of $\xi_{-0}$
given $\xi_{0}$ under $\mu$. Assume $\varPi$ to be $(\mu,\mathfrak{S})-$reversible,
where $\mathfrak{S}(\xi_{0},\xi_{-0})=\big(\mathfrak{S}_{0}(\xi_{0}),\xi_{-0}\big)$
for $\mathfrak{S}_{0}$ and involution. Then the Markov kernel 
\[
\varPi_{0}(\xi_{0},A):=\int{\bf 1}_{A}(\xi'_{0})\mu_{\xi_{0}}({\rm d}\xi_{-0})\varPi(\xi;{\rm d}\xi'),\qquad A\in\mathscr{Z},
\]
is $(\pi,\mathfrak{S}_{0})-$reversible.
\item \label{enu:nonrev-SvarPi}Let $\varPi':=\mathfrak{S}\Pi$, then with
$\psi':=\phi\circ\sigma$ and $\Psi'(\xi,{\rm d}\xi')=\delta_{\psi'(\xi)}({\rm d}\xi')$
then Properties \ref{enu:norev-prop-psi}-\ref{enu:nonrev-marginal-kernel-piS-rev}
hold with $\varPi$, $\psi$, $\Psi$ and $r$ replaced with $\varPi'$,
$\psi'$, $\Psi'$ and $r':=r\circ\sigma$.
\end{enumerate}
\end{prop}
\begin{cor}
\label{cor:time-reversible-vol-preserving}If $\psi$ in Proposition~\ref{prop:psi-is-sigma-phi}
preserves $\lambda$ then one has $r(\xi)=\rho\circ\psi/\rho\:(\xi)$
on $S=\{\xi\in E\colon\rho(\xi)\wedge\rho\circ\psi(\xi)>0\}$. Indeed,
if $\lambda^{\psi}=\lambda$ then $\lambda^{\psi\circ\psi^{-1}}=\lambda^{\psi^{-1}}$
so $\lambda^{\psi^{-1}}=\lambda$ also.
\end{cor}
\begin{cor}
\label{cor:psi-sigma-insteadof-phi}In many situations, nonreversible
kernels are given in the form of $\varPi$ or $\varPi'$, where $\psi,\psi'\colon E\rightarrow E$
are invertible mappings with the property that $\psi^{-1}=\sigma\circ\psi\circ\sigma$
for $\sigma$ an involution leaving $\mu$ and $\lambda$ invariant,
and similarly for $\psi'$. This time-reversal feature ensures that
we are in the setup of Proposition~\ref{prop:psi-is-sigma-phi},
since indeed in this setup $\phi:=\sigma\circ\psi$ (or $\tilde{\phi}:=\psi\circ\sigma$)
is an involution, therefore defining $\Pi$ satisfying the right property.
In particular we always have the decomposition $\varPi=\mathfrak{S}\tilde{\Pi}=\Pi\mathfrak{S}$
where $\Pi$ and $\tilde{\Pi}$ satisfy detailed balance \citet[Theorem 4]{Andrieu2019}.
\end{cor}
\begin{rem}
Proposition~\ref{prop:psi-is-sigma-phi} highlights the fundamental
difference between reversible and this type of nonreversible kernels.
Without refreshment of $\xi_{-0}$, the reversible Markov chain started
at $\xi$ oscillates between $\xi$ and $\phi(\xi)$ due to the involutive
property, while the nonreversible chain can in principle explore a
large subset of states $\psi^{k}(\text{\ensuremath{\xi}})$, $k\in\mathbb{N}$,
although rejection leads to backtracking. This fundamental qualitative
behaviour is exploited in more general and realistic setups, even
when $\xi_{-0}$ is refreshed.
\end{rem}
\begin{rem}
In the same way the results of \citet{maire2014comparison} can be
used in the context of Proposition~\ref{cor:rev-coordinate} (see
Remark~\ref{rmk:use-maire-douc-olsson}) one can, for example, deduce
optimality properties of $\varPi_{0}$ from those of $\varPi$ by
using \citet{Andrieu2019}.
\end{rem}
In practice, a number of deterministic transformations $\psi$ are
used to define $\pi$-invariant Markov kernels. The validity of such
kernels often rests primarily on showing that the transformation is
measure-preserving, typically with the measure being the Lebesgue
measure. We give here some examples where $\pi$ is a probability
measure associated with a position variable $x\in\mathbb{R}^{d}$
and a velocity variable $v\in\mathbb{R}^{d}$.

A general class of nonreversible MH kernels relies on the choices
$\xi=(x,v)\in E=\mathsf{X}\times\mathsf{V}$, $\sigma(x,v)=(x,-v)$
and $\mu({\rm d}x,{\rm d}v)=\pi({\rm d}x)\kappa({\rm d}v)$ where
$\kappa$ is such that $\mu^{\sigma}=\mu$. In order to keep presentation
simple we will assume that $\mathsf{X}=\mathsf{V}=\mathbb{R}^{d}$
and that $\mu$ has a density $\rho(x,v)=\varpi(x)\kappa(v)$ with
respect to the Lebesgue measure on $\mathbb{R}^{2d}$. Note that the
Lebesgue measure is invariant by $\sigma$ since its Jacobian is $1$.
\begin{lem}
\label{lem:update-one-preserve}Let $x\in\mathsf{X}\subseteq\mathbb{R}^{d}$
and $y\in\mathsf{Y}\subseteq\mathbb{R}^{d'}$, and $\psi:\mathsf{X}\times\mathsf{Y}\to\mathsf{X}\times\mathsf{Y}$
be defined as $\psi(x,y)=(x,y+f(x))$ for some function $f:\mathsf{X}\to\mathsf{Y}$.
Then $\psi$ preserves the Lebesgue measure $\lambda$ on $\mathsf{X}\times\mathsf{Y}$.
\end{lem}
\begin{example}[Guided Random walk (GRW), \citealt{gustafson1998guided}]
Let $\psi(x,v)=(x+v,v)$ then $\phi=\sigma\circ\psi=(x+v,-v)$ is
an involution, and is in fact the involution used to define the random
walk Metropolis. Then $\psi$ preserves $\lambda$ by Lemma~\ref{lem:update-one-preserve}.
Hence, using that $\kappa(-v)=\kappa(v)$ for $\xi\in S$,
\[
r(x,v)=\frac{\rho\circ\psi}{\rho}(x,v)=\frac{\varpi(x+v)}{\varpi(x)},
\]
which coincides with the acceptance ratio of the RWM Metropolis. In
fact $P(x,{\rm d}x'):=\int\kappa({\rm d}v)\varPi(x,v;{\rm d}x')$
is the $\pi-$reversible RWM Markov kernel. The GRW, of transition
$\varPi$, differs in that it is $\mu$-invariant but not reversible
and has the property that it introduces memory on the velocity component
of the process. On its own $\varPi$ does not lead to an ergodic chain
and must be combined with other updates, e.g. occasionally sampling
$v$ afresh from $\kappa$. 
\end{example}

Before covering Hamiltonian Monte Carlo, and in particular the common
variant using the velocity Verlet, or leapfrog, integrator we note
that transformations $\psi$ satisfying $\psi^{-1}=\sigma\circ\psi\circ\sigma$
are particularly intuitive in that the iterated maps $\psi\circ\cdots\circ\psi$
can be ``reversed''.
\begin{rem}
\label{rem:time-reversible-map}Let $\psi^{0}={\rm Id}$ and $\psi^{k}=\psi\circ\psi^{k-1}$
for $k\in\mathbb{N}$. If $\psi$ satisfies $\psi^{-1}=\sigma\circ\psi\circ\sigma$,
then $\psi$ is time-reversible in the sense that $\phi_{k}=\sigma\circ\psi^{k}$
is an involution for any $k\in\mathbb{N}$. Indeed, we have
\[
{\rm Id}=\psi^{-k}\circ\psi^{k}=(\sigma\circ\psi\circ\sigma)^{k}\circ\psi^{k}=\sigma\circ\psi^{k}\circ\sigma\circ\psi^{k}.
\]
\end{rem}
\begin{lem}
\label{lem:alternating-maps-preserve-reversible}Let $x,v\in\mathbb{R}^{d}$
and $\psi:\mathbb{R}^{2d}\to\mathbb{R}^{2d}$ be
\[
\psi=\psi_{B}\circ\psi_{A}\circ\psi_{B},
\]
where $\psi_{B}=(x,v)\mapsto(x,v+\imath(x))$ and $\psi_{A}=(x,v)\mapsto(x+\jmath(v),v)$
for some functions $\imath\colon\mathsf{X}\rightarrow\mathsf{V}$
and $\jmath\colon\mathsf{V}\rightarrow\mathsf{X}$, where $\jmath(-v)=-\jmath(v)$
for $v\in\mathbb{R}^{d}$. Then $\psi_{A}$, $\psi_{B}$ and $\psi$
preserve the Lebesgue measure on $\mathbb{R}^{2d}$ and $\psi^{-1}=\sigma\circ\psi\circ\sigma$
so that $\psi$ is time-reversible in the sense of Remark~\ref{rem:time-reversible-map}.
\end{lem}
\begin{example}[HMC - leapfrog integrator]
\label{eg:HMC-setup}Let $\pi$ have density $\rho=\varpi\otimes\kappa$
w.r.t. $\lambda$, the Lebesgue measure on $E=\mathbb{R}^{2d}$. Consider
the function $h=\psi_{B}\circ\psi_{A}\circ\psi_{B}$ as in Lemma~\ref{lem:alternating-maps-preserve-reversible}
with $\imath(x)=\frac{\epsilon}{2}\nabla\log\varpi(x)$ and $\jmath(v)=\epsilon\nabla\log\kappa(v)$.
Let $\varPi=\Pi\mathfrak{S}$ be the nonreversible kernel in Proposition~\ref{prop:psi-is-sigma-phi}
with $\psi=h^{k}$ for some $k\in\mathbb{N}$, and acceptance ratio
\[
r(\xi)=\frac{\rho\circ\psi}{\rho}(\xi),\qquad\xi\in S.
\]
This kernel is a version of the HMC kernel with leapfrog integrator
(see Remark~\ref{rem:nonrev-to-rev} below). It has desirable properties,
but it is also clear that the $\mu$-invariance of $\varPi$ applies
for a much broader class of $\imath$ and $\jmath$, as implied by
the appeal to Lemma~\ref{lem:alternating-maps-preserve-reversible}.
For example, it is well known that one could replace $\varpi$ in
$\imath$ with some approximate density \citep[see, e.g., ][Section 5.5]{neal2011mcmc},
i.e. run the ``leapfrog integrator'' for a different density but
accept or reject using $\rho=\varpi\otimes\kappa$. In order to preserve
persistence of motion (and nonreversibility) this update is typically
combined with partial refreshment of the velocity. As discussed below,
full refreshment leads to a reversible algorithm.
\end{example}
\begin{rem}
\label{rem:nonrev-to-rev}It is often the case, as was the case in
part of the seminal paper of \citet{horowitz1991generalized}, that
the kernel considered is reversible. Indeed in those works the kernel
considered is, for $(x,A)\in\mathsf{X}\times\mathscr{X}$
\[
P(x,A):=\int\kappa({\rm d}v)\varPi(x,v;{\rm d}(y,w))\mathbb{I}\{(y,w)\in A\times\mathsf{V}\},
\]
that is the velocity is refreshed at each iteration and with $\xi=(x,v)$
and $f,g\colon\mathsf{X}\rightarrow[0,1]$,
\begin{align*}
\int f(x)g(x')\pi({\rm d}x)P(x,{\rm d}x') & =\int f(x)g(x')\mu({\rm d}\xi)\varPi(\xi,{\rm d}\xi')\\
 & =\int f(x)g(x')\mu({\rm d}\xi')\mathfrak{S}\varPi\mathfrak{S}(\xi,{\rm d}\xi')\\
 & =\int f(x)g(x')\mu\mathfrak{S}({\rm d}\xi')\mathfrak{S}\varPi(\xi,{\rm d}\xi')\\
 & =\int f(x)g(x')\mu({\rm d}\xi')\varPi(\xi,{\rm d}\xi')\\
 & =\int f(x)g(x')\pi({\rm d}x')P(x',{\rm d}x).
\end{align*}
where we have used the $(\mu,\mathfrak{S})-$reversibility of $\varPi$,
$\mu=\mu\mathfrak{S}$ and the fact that $\mathfrak{S}g=g$ for this
choice of function.
\end{rem}
\begin{example}[MALA and generalized MALA ]
Standard, reversible normal (i.e. $\kappa$ is the standard normal
distribution) MALA \citep{besag-discussion-1994} corresponds to one
iteration of HMC - leapfrog integrator with full refreshment of the
velocity at each iteration, and indeed here $\psi(x,v)=\big(x+\imath(x)+\frac{\epsilon}{2}v,-v-\tfrac{\epsilon}{2}\big[\imath(x)+\imath\big(x+\tfrac{\epsilon}{2}\imath(x)+\frac{\epsilon}{2}v\big)\big]\big)$
for $\imath(x)=\nabla_{x}\log\varpi(x)$ and $\epsilon>0$. In \citet{poncet2017generalized}
it is proposed to consider $\imath\colon\mathsf{X}\times\{-1,1\}\rightarrow\mathsf{X}$
with $\imath(x,s)=\nabla_{x}\log\varpi(x)+s\gamma(x)$ for $\gamma\colon\mathsf{X}\rightarrow\mathsf{X}$.
A naïve idea would be to take $\psi_{B}=(x,v,s)\mapsto(x,v+\imath(x,s),s)$
and $\psi_{A}=(x,v,s)\mapsto(x+\jmath(v),v,s)$ with $\sigma(x,v,s)=(x,v,-s)$
which is shown to have poor properties; this leads to the development
of a scheme relying on an implicit integration scheme.
\end{example}
\begin{example}[Hyperplane reflection]
If $\lambda_{V}$ is the Lebesgue measure on $\mathbb{R}^{d}$, the
involution ${\rm b}(x,v)=(x,v-2\{n(x)^{\top}v\}n(x))$ preserves $\lambda=\lambda_{X}\times\lambda_{V}$,
where $n:\mathbb{R}^{d}\to\mathbb{R}^{d}$ satisfies $\left\Vert n(x)\right\Vert ^{2}=n(x)^{\top}n(x)=1$
for all $x\in\mathbb{R}^{d}$. Indeed, we can write the $v$-component
of $\phi(x,v)$ as
\[
v-2(n(x)^{\top}v)n(x)=({\rm Id}-2n(x)n(x)^{\top})v,
\]
and we see that ${\rm B}_{x}=({\rm Id}-2n(x)n(x)^{\top})$ is a matrix
with ${\rm B}_{x}^{2}={\rm Id}$ and so $\left|\det{\rm B}_{x}\right|=1$.
Since ${\rm b}$ does not move the $x$-component, it follows that
${\rm b}$ is $\lambda$-preserving.

If $\lambda=\lambda_{X}\times\lambda_{V}$ and $\lambda_{V}$ is instead
the uniform measure on the sphere $\mathbb{S}^{d-1}=\{v\in\mathbb{R}^{d}:v^{\top}v=1\}$
then ${\rm B}_{x}^{\top}={\rm B}_{x}$ so $\left\Vert {\rm B}_{x}v\right\Vert ^{2}=({\rm B}_{x}v)^{\top}{\rm B}_{x}v=v^{\top}v=\left\Vert v\right\Vert ^{2}$,
so ${\rm B}_{x}$ preserves the norm $\left\Vert \cdot\right\Vert $.
Letting ${\rm Leb}$ denote the Lebesgue measure on $\mathbb{R}^{d}$,
and noting from the argument above that ${\rm b}$ preserves ${\rm Leb}$,
we can then conclude that ${\rm b}$ is $\lambda$-preserving as above
because for any measurable $A\subseteq\mathbb{S}^{d-1}$, $\lambda_{V}(A)\propto{\rm Leb}(\{tv:v\in A,t\in[0,1]\})$.
\end{example}
A natural question is whether the requirement that $\sigma\colon\mathsf{E}\rightarrow\mathsf{E}$
be an involution can be relaxed to invertibility only. More precisely
let $\psi_{0},\sigma_{0}\colon\mathsf{Z}\rightarrow\mathsf{Z}$ be
invertible with $\psi_{0}^{-1}=\sigma_{0}^{-1}\circ\psi_{0}\circ\sigma_{0}$--such
a structure is known as time-reversible symmetry when $\psi_{0}$
is the flow of a dynamical system with this property \citep{lamb1998time}.
Let $\Psi_{0}(z,{\rm d}z')=\delta_{\psi_{0}(z)}({\rm d}z')$, $\mathfrak{S}_{0}^{\pm1}(z,{\rm d}z')=\delta_{\sigma_{0}^{\pm}(z)}({\rm d}z')$
and $\mu_{0}$ a probability distribution on $\big(\mathsf{Z},\mathscr{Z}\big)$
such that $\mu_{0}\mathfrak{S}_{0}^{\pm1}=\mu_{0}$. Can one define
a deterministic MH type kernel leaving $\mu_{0}$ invariant -- \citet{fang2014compressible}
provide us with an answer, see below. Our answer consists of embedding
this problem in the $(\mu,\mathfrak{S})-$reversible framework. Let
$\mathsf{E}=\mathsf{Z}\times\mathsf{U}$ where $\mathsf{U}=\{-1,1\}$
and define $\mu\big({\rm d}(z,u)\big):=\frac{1}{2}\mu_{0}({\rm d}z)\mathbb{I}\big\{ u\in\mathsf{U}\big\}$
and consider the mappings $\sigma,\psi\colon\mathsf{Z}\times\mathsf{U}\rightarrow\mathsf{Z}\times\mathsf{U}$
such that for $f\colon\mathsf{Z}\times\mathsf{U}\rightarrow\mathbb{R}$,
$f\circ\sigma(z,u)=f\big(\sigma_{0}^{u}(z),-u\big)$ and $\psi(z,u):=\big(\sigma_{0}^{-u}\circ\psi_{0}^{u}(z),u\big)$.
For any $(z,u)\in\mathsf{Z}\times\mathsf{U}$ we have that $\sigma^{2}(z,u)=\big(\sigma_{0}^{-u}\circ\sigma_{0}^{u}(z),u\big)=(z,u)$,
that is $\sigma$ is an involution and one can check that $\psi^{-1}(z,u)=\big(\psi_{0}^{-u}\circ\sigma_{0}^{u}(z),u\big)$.
Noting that $\psi_{0}^{-1}=\sigma_{0}^{-1}\circ\psi_{0}\circ\sigma_{0}$
is equivalent to $\psi_{0}^{-u}=\sigma_{0}^{-u}\circ\psi_{0}^{u}\circ\sigma_{0}^{u}$
for $u\in\mathsf{U}$ we have for $(z,u)\in\mathsf{Z}\times\mathsf{U}$
\begin{align*}
\sigma\circ\psi\circ\sigma(z,u) & =\sigma\circ\psi\big(\sigma_{0}^{u}(z),-u\big)\\
 & =\sigma\big(\sigma_{0}^{u}\circ\psi_{0}^{-u}\circ\sigma_{0}^{u}(z),-u\big)\\
 & =\big(\psi_{0}^{-u}\circ\sigma_{0}^{u}(z),u\big)\\
 & =\psi^{-1}(z,u).
\end{align*}
Finally, for any $f\colon\mathsf{Z}\times\mathsf{U}\rightarrow\mathbb{R}$
we have $\mu\mathfrak{S}=\mu$, since
\begin{align*}
\int f(z,u)\mu^{\sigma}\big({\rm d}(z,u)\big) & =\int f\circ\sigma(z,u)\mu\big({\rm d}(z,u)\big)\\
 & =\int f\big(\sigma_{0}^{u}(z),-u\big)\frac{1}{2}\mu_{0}({\rm d}z)\mathbb{I}\big\{ u\in\mathsf{U}\big\}\\
 & =\int f\big(z,-u\big)\frac{1}{2}\mu_{0}^{\sigma_{0}^{u}}({\rm d}z)\mathbb{I}\big\{ u\in\mathsf{U}\big\}\\
 & =\int f\big(z,u\big)\frac{1}{2}\mu_{0}({\rm d}z)\mathbb{I}\big\{ u\in\mathsf{U}\big\}\\
 & =\int f(z,u)\mu\big({\rm d}(z,u)\big),
\end{align*}
We are therefore back in the $(\mu,\mathfrak{S})-$reversible setup
and with
\begin{align*}
\alpha(z,u) & :=a\left(\frac{\rho\circ\psi(z,u)}{\rho(z,u)}\frac{{\rm d}\lambda^{\psi^{-1}}}{{\rm d}\lambda}(z,u)\right)\\
 & =a\left(\frac{\rho_{0}\circ\sigma_{0}^{-u}\circ\psi_{0}^{u}(z)}{\rho_{0}(z)}\frac{{\rm d}\lambda^{\psi_{0}^{-u}\circ\sigma_{0}^{u}}}{{\rm d}\lambda_{0}}(z)\right)\\
 & =a\left(\frac{\rho_{0}\circ\psi_{0}^{u}(z)}{\rho_{0}(z)}\frac{{\rm d}\lambda^{\psi_{0}^{-u}}}{{\rm d}\lambda_{0}}(z)\right)
\end{align*}
we can define the kernel
\begin{align*}
\varPi f(z,u) & =\alpha(z,u)\cdot f\circ\psi(z,u)+\bar{\alpha}(z,u)\cdot f\circ\sigma(z,u)\\
 & =\alpha(z,u)\cdot f\big(\sigma_{0}^{-u}\circ\psi_{0}^{u}(z),u\big)+\bar{\alpha}(z,u)\cdot f\big(\sigma_{0}^{u}(z),-u\big).
\end{align*}
The kernel $\widetilde{\varPi}\colon\mathsf{Z}\times\mathscr{Z}\rightarrow[0,1]$
proposed by \citet{fang2014compressible} is, for $g\colon\mathsf{Z}\rightarrow\mathbb{R}$,
\begin{align*}
\widetilde{\varPi}g(z) & =\alpha(z,1)g\circ\psi_{0}(z)+\bar{\alpha}(z,1)g\circ\sigma_{0}(z)\\
 & =\alpha(z,1)\tilde{g}\circ\sigma\circ\psi(z,1)+\bar{\alpha}(z,1)\tilde{g}\circ\sigma(z,1)
\end{align*}
where $\tilde{g}\colon\mathsf{Z}\times\mathsf{U}\rightarrow\mathbb{R}$
is such that $\tilde{g}(z,1)=\tilde{g}(z,-1):=g(z)$, which can therefore
be thought of as being $\varPi$ but used for the value $u=1$ only--one
could equally have chosen $u=-1$, naturally. One can check that this
kernel satisfies global balance for $\mu_{0}$ directly \citep{fang2014compressible}.
 The kernel $\widetilde{\varPi}$ does not satisfy detailed or skew
detailed balance, but noting that $\phi:=\sigma\circ\psi$ is an involution
and letting $\Phi(z,{\rm d}z'):=\delta_{\phi(z)}({\rm d}z')$, that
is $\Phi f(z,u)=(\psi_{0}^{u},-u)$ we use ``reversibility'' (see
the proof Proposition~\ref{prop:psi-is-sigma-phi}) to show
\begin{align*}
\frac{1}{2}\int\alpha(z,1)g\circ\psi_{0}(z)\mu_{0}({\rm d}z) & =\int\mathbf{1}_{\mathsf{Z}\times\{1\}}(z,u)\tilde{g}\circ\phi(z,u)\alpha(z,u)\mu({\rm d}z,u)\\
 & =\int\mathbf{1}_{\mathsf{Z}\times\{1\}}(z,u)\Phi\tilde{g}(z,u)\alpha(z,u)\mu({\rm d}z,u)\\
 & =\int\tilde{g}(z,u)\Phi\mathbf{1}_{\mathsf{Z}\times\{1\}}(z,u)\alpha(z,u)\mu({\rm d}z,u)\\
 & =\int g(z)\mathbf{1}_{\mathsf{Z}\times\{-1\}}(z,u)\alpha(z,u)\mu({\rm d}z,u)\\
 & =\frac{1}{2}\int g(z)\alpha(z,-1)\mu_{0}({\rm d}z)
\end{align*}
and similarly with $\Phi$ replaced with $\mathfrak{S}$ (see the
proof Proposition~\ref{prop:psi-is-sigma-phi})
\begin{align*}
\frac{1}{2}\int g\circ\sigma_{0}(z)\alpha(z,1)\mu_{0}({\rm d}z) & =\int\mathbf{1}_{\mathsf{Z}\times\{1\}}(z,u)\mathfrak{S}\tilde{g}(z,u)\alpha(z,u)\mu({\rm d}z,u)\\
 & =\int\tilde{g}(z,u)\mathfrak{S}\mathbf{1}_{\mathsf{Z}\times\{1\}}(z,u)\alpha(z,u)\mu({\rm d}z,u)\\
 & =\int g(z)\mathbf{1}_{\mathsf{Z}\times\{-1\}}(z,u)\alpha(z,u)\mu({\rm d}z,u)\\
 & =\frac{1}{2}\int g(z)\alpha(z,-1)\mu_{0}({\rm d}z).
\end{align*}
Therefore for any $g\colon\mathsf{Z}\rightarrow[0,1]$
\[
\int\widetilde{\varPi}g(z)\mu_{0}({\rm d}z)=\int g\circ\sigma_{0}(z)\mu_{0}({\rm d}z)=\int g(z)\mu_{0}({\rm d}z)
\]
and we conclude.

%% file: stop.tex
\section{Markov chain proposals, stopping times and processes \& NUTS\label{sec:NUTS}}

In some scenarios it is desirable for $\mu$ to involve simulation
of a stopped process. In particular, this allows the amount of simulation
required to produce a suitable proposal to be random and ideally be
appropriately adapted to features of the target distribution and the
current point. As mentioned in Remark~\ref{rem:non-unique}, the
specification of $(\mu,\phi)$ is not unique for a given Markov kernel,
so there is some flexibility in precisely how stopping times and stopped
processes are captured in $\xi$ and described by $\mu$. In particular,
one often has flexibility in allowing $\xi$ to be infinite-dimensional
and to contain a realization of the original process as well as the
stopping time, or for $\xi$ to be finite-dimensional and to contain
only the stopped process. In the former case, one will need to adopt
an indirect implementation as per Remark~\ref{rem:lazy}.

\subsection{A toy example\label{subsec:A-toy-example}}

We illustrate the former approach on a simple example with i.i.d.
proposals. Let $\xi=(n,\mathtt{Z},k)\in\mathbb{N}\times\mathsf{Z}^{\mathbb{N}}\times\mathbb{N}$.
Assume that $\nu\gg\pi$ and let $\varpi(z)={\rm d}\pi/{\rm d}\nu(z)$
-- a common situation is when $\pi$ and $\nu$ have densities w.r.t.
some common dominating measure and $\varpi(z)=\pi(z)/\nu(z)$ if we
keep the same notation for these densities. Assume that the distribution
of $\mathtt{Z}$ under $\mu$ is that $z_{0}\sim\pi$ and for $i\in\mathbb{N}_{*}$,
$z_{i}\overset{{\rm iid}}{\sim}\nu$. Let $(s_{n})_{n\in\mathbb{N}_{*}}$
be a sequence of functions such that $s_{n}:\mathsf{Z}^{\mathbb{N}}\to\{0,1\}$
depends only on the first $n+1$ members of its argument; i.e. $s_{n}(z)$
depends only on $z_{0},\ldots,z_{n}$. Define the stopping time for
$\mathtt{Z}\in\mathsf{Z}^{\mathbb{N}}$
\begin{equation}
\tau=\tau(\mathtt{Z}):=\inf\{n\geq1:s_{n}(\mathtt{Z})=1\}.\label{eq:generic-def-tau}
\end{equation}
For example, one could choose $s_{n}(\mathtt{Z})=\mathbb{I}\{\sum_{i=0}^{n}\varpi(z_{i})>c\}$
for some constant $c>0$, or $s_{n}(\mathtt{Z})=\mathbb{I}\{{\rm ESS}(z_{0},\ldots,z_{n})>c\}$
with 
\[
{\rm ESS}(z_{0},\ldots,z_{n}):=\frac{\left\{ \sum_{i=0}^{n}\varpi(z_{i})\right\} ^{2}}{\sum_{i=0}^{n}\varpi(z_{i})^{2}},
\]
heuristically to ensure that sufficiently many samples have been drawn
and that one can be chosen to produce a sample approximately drawn
from $\pi$. For $\mathtt{Z}\in\mathsf{Z}^{\mathbb{N}}$, let $n:=\tau(\mathtt{Z})$
and $k\sim\varsigma(\cdot;n,\mathtt{Z})$ where $\varsigma(\cdot;n,\mathtt{Z})$
is an arbitrary categorical distribution taking values in $\llbracket0,n\rrbracket$
and with probabilities depending on $z_{0},\ldots,z_{n}$ only. For
$k\in\mathbb{N}$, let $\sigma_{k}:\mathsf{Z}^{\mathbb{N}}\to\mathsf{Z}^{\mathbb{N}}$
be the swapping function such that, with $\mathtt{Z}'=\sigma_{k}(\mathtt{Z})$,
$z_{0}'=z_{k}$, $z_{k}'=z_{0}$ and $z'_{j}=z_{j}$ for $j\notin\{0,k\}$.
Clearly, $\sigma_{k}$ is an involution and we consider $\phi(n,\mathtt{Z},k):=(\tau\circ\sigma_{k}(\mathtt{Z}),\sigma_{k}(\mathtt{Z}),k)$,
which is an involution since
\[
\phi\circ\phi(n,\mathtt{Z},k)=(\tau\circ\sigma_{k}\circ\sigma_{k}(\mathtt{Z}),\sigma_{k}\circ\sigma_{k}(\mathtt{Z}),k)=(n,\mathtt{Z},k).
\]
Letting $\mathcal{\nu^{\otimes\infty}}$ denote the probability measure
associated with an infinite sequence of independent $\nu$-distributed
random variables, we have for $n\in\mathbb{N}$ and $k\in\llbracket0,n\rrbracket$,
\begin{align*}
\mu(n,{\rm d}\mathtt{Z},k) & =\varpi(z_{0})\nu^{\otimes\infty}({\rm d}\mathtt{Z})\varsigma(k;n,\mathtt{Z})s_{n}(\mathtt{Z})\prod_{i=1}^{n-1}\left\{ 1-s_{i}(\mathtt{Z})\right\} ,
\end{align*}
where we note that $z_{0}$ is marginally distributed according to
$\pi$, which together with $\phi$ above defines the kernel outlined
in Alg.~\ref{alg:adaptive-IMH-1}. One can check that the acceptance
ratio is, with $(n',\mathtt{Z}',k')=\phi(n,\mathtt{Z},k)=(\tau\circ\sigma_{k}(\mathtt{Z}),\sigma_{k}(\mathtt{Z}),k)$
and for $\xi\in S=\{(n,\mathtt{Z},k):n=\tau(\mathtt{Z}),0\leq k\leq\tau(\mathtt{Z})\wedge\tau\circ\sigma_{k}(\mathtt{Z})\}$,
\[
r(\xi)=\frac{\varpi(z'_{0})\varsigma(k';n',\mathtt{Z}')}{\varpi(z_{0})\varsigma(k;n,\mathtt{Z})}=\frac{\varpi(z_{k})\varsigma(k;n',\mathtt{Z}')}{\varpi(z_{0})\varsigma(k;n,\mathtt{Z})}.
\]
Although theoretically convenient, the algorithm described in Alg.
\ref{alg:adaptive-IMH-1} is not very practical due to the requirement
to sample the infinite-dimensional $\mathtt{Z}_{-0}$. However the
definitions of $(s_{n})_{n\in\mathbb{N_{*}}}$, $\tau(\mathtt{Z})$,
$k$ and $\sigma_{k}(\mathtt{Z})$ are such that $\mathtt{Z}$ is
not required in its entirety to simulate from the kernel, which can
be achieved with finite computation provided $\tau(\mathtt{Z})<\infty$.
This is described in Alg. \ref{alg:adaptive-IMH-2}, with a slight
abuse of notation since $s_{k}$ and $\sigma_{k}$ are defined on
$\mathsf{Z}^{\mathbb{N}}$. We will refer to this as a ``lazy''
implementation or simulation and adopt the presentation in Alg.~\ref{alg:adaptive-IMH-1}
for brevity. In particular, the explicit lazy implementation in Alg.~\ref{alg:adaptive-IMH-2}
involves simulating only those components of $\mathtt{Z}$ and $\mathtt{Z}'$
that are required to implement Alg.~\ref{alg:adaptive-IMH-1}, the
details of which are fairly straightforward and tend to obscure the
simplicity of the approach. Note that throughout we give the expression
for the acceptance ratio on $S$ only in order to alleviate presentation.

\begin{algorithm}
\caption{\label{alg:adaptive-IMH-1}Impractical algorithm}

\begin{enumerate}
\item Simulate (lazily) $z_{i}\overset{{\rm iid}}{\sim}\nu$, for $i\in\mathbb{N}_{*}$.
\item Set $n\leftarrow\tau(\mathtt{Z})$.
\item Simulate $k\sim\varsigma(\cdot\mid n,\mathtt{Z})$.
\item Set $\mathtt{Z}'\leftarrow\sigma_{k}(\mathtt{Z})$ and $n'\leftarrow\tau\circ\sigma_{k}(\mathtt{Z})$.
\item With probability
\[
a\left(\frac{\varpi(z_{k})\varsigma(k;n',\mathtt{Z}')}{\varpi(z_{0})\varsigma(k;n,\mathtt{Z})}\right),
\]
output $z_{k}$, otherwise output $z_{0}$.
\end{enumerate}
\end{algorithm}

\begin{algorithm}
\caption{\label{alg:adaptive-IMH-2}Practical, lazy implementation}

\begin{enumerate}
\item Set $i\leftarrow1$, simulate $z_{1}\sim\nu$.
\item While $s_{i}(z_{0:i})=0$
\begin{enumerate}
\item Set $i\leftarrow i+1$.
\item Simulate $z_{n}\sim\nu$.
\end{enumerate}
\item Set $n\leftarrow i=\tau(\mathtt{Z})$.
\item Simulate $k\sim\varsigma(\cdot;n,\mathtt{Z})$ and set $z'_{0:n}=\sigma_{k}(z_{0:n})$.
\item Set $i\leftarrow1$.
\item While $s_{i}(z_{0:i}')=0$
\begin{enumerate}
\item Set $i\leftarrow i+1$.
\item If $i>n$, simulate $z'_{i}=z_{i}\sim\nu$.
\end{enumerate}
\item Set $n'\leftarrow i=\tau\circ\sigma_{k}(\mathtt{Z})=\tau(\mathtt{Z}')$.
\item With probability
\[
a\left(\frac{\varpi(z_{k})\varsigma(k;n',\mathtt{Z}')}{\varpi(z_{0})\varsigma(k;n,\mathtt{Z})}\right),
\]
output $z_{k}$, otherwise output $z_{0}$.
\end{enumerate}
\end{algorithm}

If we choose $(s_{n})_{n\in\mathbb{N}_{*}}$ such that $\tau(\mathtt{Z})=1$
and $\varsigma(1;1,\mathtt{Z})=1$ for all $\mathtt{Z}\in\mathsf{Z}^{\mathbb{N}}$
then this reduces to the independent MH (IMH), but of course in general
it allows more than one candidate sample from $\nu$ to be simulated.
We refer to the kernel in Alg.~\ref{alg:adaptive-IMH-1} as an adaptive
IMH kernel for this reason. If we let $\varsigma(k;n,\mathtt{Z})\propto\varpi(z_{k}){\bf 1}_{\llbracket0,n\rrbracket}(k)$
then we obtain
\[
r(\xi)=\frac{\sum_{i=0}^{n}\varpi(z_{i})}{\sum_{i=0}^{n'}\varpi(z_{i}')}{\bf 1}_{\llbracket0,n\wedge n'\rrbracket}(k)
\]
An important point is that $n'$ may not equal $n$ in general, requiring
in particular additional simulations when $n'>n$. By choosing $(s_{n})_{n\in\mathbb{N}}$
and $\varsigma(\cdot;n,\mathtt{Z})$ appropriately one can ensure
that $n=n'$ for all $\mathtt{Z}\in\mathsf{Z}^{\mathbb{N}}$. This
has the appeal that there is no need to perform additional simulations
once $z_{1},\ldots,z_{\tau}$ are realized, and can also mean that
the acceptance ratio is one The following lemma provides sufficient
conditions for this equality to hold.
\begin{lem}
\label{lem:stop-IMH}Let $\mathtt{Z}\in\mathsf{Z}^{\mathbb{N}}$ be
such that, with the definition in (\ref{eq:generic-def-tau}), $\tau(\mathtt{Z})<\infty$
and assume further that $(s_{k})_{k\in\mathbb{N}_{*}}$ satisfies
\begin{enumerate}
\item \label{enu:sk-nondecreasing}$k\mapsto s_{k}(\mathtt{Z})$ is non-decreasing;
\item \label{enu:IMH-lem-tau-constant}$s_{k}(\mathtt{Z})=s_{k}\circ\sigma_{l}(\mathtt{Z})$
for all $l\in\llbracket k\rrbracket$;
\end{enumerate}
Then $\tau\circ\sigma_{l}(\mathtt{Z})=\tau(\mathtt{Z})$ for $l\in\llbracket\tau(\mathtt{Z})-1\rrbracket$.
\end{lem}
\begin{example}
\label{eg:aimh-stop-nprime-equals-n}For $k\in\mathbb{N}$, let $s_{k}(\mathtt{Z}):=\mathbb{I}\{\sum_{i=0}^{k}\varpi(z_{i})>c\}$
for some constant $c>0$ and let $\mathtt{Z}\in\mathsf{Z}^{\mathbb{N}}$
satisfy $\tau(\mathtt{Z})<\infty$. Clearly $k\mapsto s_{k}(\mathtt{Z})$
is non-decreasing and condition \ref{enu:sk-nondecreasing} in Lemma~\ref{lem:stop-IMH}
holds. Assume that for $n\in\mathbb{N}$, $s_{1}(\mathtt{Z})=\cdots=s_{n-1}(\mathtt{Z})=0$
while $s_{n}(\mathtt{Z})=1$. It follows that for $l\in\llbracket0,n-1\rrbracket$,
if $k\in\llbracket n-1\rrbracket$ then $s_{k}\circ\sigma_{l}(\mathtt{Z})\leq\mathbb{I}\{\sum_{i=0}^{n-1}\varpi(z_{i})>c\}=s_{n-1}(\mathtt{Z})=s_{k}(\mathtt{Z})=0$
while if $k\geq n$ then $s_{k}\circ\sigma_{l}(\mathtt{Z})\geq\mathbb{I}\{\sum_{i=0}^{n}\varpi(z_{i})>c\}=s_{n}(\mathtt{Z})=1=s_{k}(\mathtt{Z})$
so condition \ref{enu:IMH-lem-tau-constant} in Lemma~\ref{lem:stop-IMH}
holds. If we take $\varsigma(k;n,\mathtt{Z})\propto\varpi(z_{k}){\bf 1}_{\llbracket0,n-1\rrbracket}(k)$
then the conclusion of Lemma~\ref{lem:stop-IMH} holds and we obtain
for $k\in\llbracket0,n-1\rrbracket$, with $\xi=(n,\mathtt{Z},k)$
and $\phi(n,\mathtt{Z},k):=(\tau\circ\sigma_{k}(\mathtt{Z}),\sigma_{k}(\mathtt{Z}),k)$,
\[
r(\xi)=\frac{\sum_{i=0}^{n-1}\varpi(z_{i})}{\sum_{i=0}^{n-1}\varpi(z_{i}')}=\frac{\sum_{i=0}^{n-1}\varpi(z_{i})}{\sum_{i=0}^{n-1}\varpi(z{}_{i})}=1.
\]
\end{example}
In contrast, for the choice $s_{n}(\mathtt{Z})=\mathbb{I}\{{\rm ESS}(z_{0},\ldots,z_{n})>c\}$,
condition \ref{enu:sk-nondecreasing} in Lemma~\ref{lem:stop-IMH}
is not satisfied and there is no reason for $n=n'$ to hold for all
$\mathtt{Z}\in\mathsf{Z}^{\mathbb{N}}$. We expand on the idea of
Lemma~\ref{lem:stop-IMH} to validate the NUTS and stopping-time
MTM in Sections~\ref{sec:NUTS} and~\ref{sec:MTM}.

\subsection{\label{subsec:doubly-infinite-general}Doubly-infinite Markov chain
proposal and change of measure}

Here we demonstrate how one can verify that Markov chain proposals
can be used to define $\pi$-invariant Markov kernels. First we show
how to deal with a distribution $\mu$ involving a doubly-infinite
Markov chain as well as a proposal index. Then we consider the more
involved but practical scenario where the proposal index is selected
from a window of random size, adapted according to user-defined constraint
functions. A special case of this framework, and indeed the inspiration
for the generalization here, is when the Markov chain is a deterministic
dynamical system is the No U-Turn Sampler (NUTS) of \citet{hoffman2014no}.

Let $\pi$ and $\nu$ be measures on $(\mathsf{Z},\mathscr{Z})$ where
$\pi$ is a probability, $\nu\gg\pi$ and let $\varpi:={\rm d}\pi/{\rm d}\nu$.
In order to present our algorithm we require the definition of a two
sided Markov chain, from which a proposal state is chosen within a
MH kernel update.
\begin{defn}[Two-sided $(k,\pi,Q,Q^{*})$-Markov chain probability measure $\Lambda^{k}$]
Let $\pi$ be a probability measure on $\big(\mathsf{Z},\mathscr{Z}\big)$
and $Q$, $Q^{*}\colon\mathsf{Z}\times\mathscr{Z}\rightarrow[0,1]$
be transition kernels. Then for $k\in\mathbb{Z}$, denote by $\Lambda^{k}$
the probability measure on $\big(\mathsf{Z}^{\mathbb{Z}},\mathscr{Z}^{\otimes\mathbb{Z}}\big)$
associated with the Markov chain $Z$ such that $Z_{k}\sim\pi$, and
for $i\in\mathbb{N}_{*}$, $Z_{i+k}\mid\{Z_{i+k-1}=z\}\sim Q(z,\cdot)$
and $Z_{k-i}\mid\{Z_{k-i+1}=z\}\sim Q^{*}(z,\cdot)$.
\end{defn}
We define $\mathcal{Q}(z_{0},\cdot)$ to be the probability measure
for $\mathtt{Z}\sim\Lambda^{0}$ conditional on a fixed $z_{0}$.

For any $\mathtt{Z}\in\mathsf{Z}^{\mathbb{Z}}$ we let $\varsigma(\cdot;\mathtt{Z})$
be a probability distribution on $\big(\mathbb{Z},\mathscr{P}(\mathbb{Z})\big)$
and we are interested in the update outlined in Alg.~\ref{alg:MC-proposal-general},
where for $i\in\mathbb{Z}$, $\theta^{i}:\mathsf{Z}^{\mathbb{Z}}\to\mathsf{Z}^{\mathbb{Z}}$
is the shift function defined via $\theta^{i}(\mathtt{Z})_{j}=z_{i+j}$
and to ease the presentation of the algorithms, we write that one
should ``lazily'' simulate a realization of a double-infinite Markov
chain, by which we mean that only a finite number of states of the
Markov chain should be required to perform the rest of the algorithm.
The simulation of $\mathtt{Z}$ is naturally not practical and a stopping
criterion is required, while making sense of the acceptance ratio
and its expression require an additional assumption on $(\nu,Q,Q^{*})$.
These are the topics of the remainder of the subsection. 

\begin{algorithm}[H]
\caption{\label{alg:MC-proposal-general}To sample from $P_{{\rm MC}}^{{\rm (general)}}(z_{0},\cdot)$}

\begin{enumerate}
\item Lazily simulate $\mathtt{Z}\sim\mathcal{Q}(z_{0},\cdot)$.
\item Simulate $k\sim\varsigma(\cdot;\mathtt{Z})$.
\item With probability
\[
a\left(\frac{\varpi(z_{k})\varsigma\big(-k;\theta^{k}(\mathtt{Z})\big)}{\varpi(z_{0})\varsigma(k;\mathtt{Z})}\right),
\]
output $z_{k}$. Otherwise output $z_{0}$.
\end{enumerate}
\end{algorithm}

We first introduce an assumption on $(\nu,Q,Q^{*})$ justifying the
form of the acceptance ratio in full generality.
\begin{defn}[Reversible triplet $(\nu,Q,Q^{*})$]
 Let $\nu$ be a measure on $(\mathsf{Z},\mathscr{Z})$ and $Q,Q^{*}\colon\mathsf{Z}\times\mathscr{Z}\rightarrow[0,1]$
be two Markov kernels. We say that $(\nu,Q,Q^{*})$ is a reversible
triplet if for $z,z'\in\mathsf{Z}$, 
\begin{equation}
\nu({\rm d}z)Q(z,{\rm d}z')=\nu({\rm d}z')Q^{*}(z',{\rm d}z),\label{eq:Q-Qstar-time-reversal}
\end{equation}
\end{defn}
This implies in particular that $Q$ is $\nu$-invariant, whether
$\nu$ is a probability measure or not, and $Q^{*}$ is the time-reversal
of $Q$. In operator theoretic language $Q^{*}$ is the $\nu-$adjoint
of $Q$ for the inner product $\bigl\langle f,g\bigr\rangle=\int fg{\rm d}\nu$
on $L^{2}(\nu)$. Importantly for practical purposes, we observe that
(\ref{eq:Q-Qstar-time-reversal}) accommodates invertible mappings
that leave $\nu$ invariant: 
\begin{prop}
\label{prop:measure-preserving-reversal}Let $\nu$ be a measure on
$(\mathsf{Z},\mathscr{Z})$ and $\psi\colon\mathsf{Z}\rightarrow\mathsf{Z}$
be invertible and such that $\nu^{\psi}=\nu$. Then (\ref{eq:Q-Qstar-time-reversal})
holds with $Q(z,{\rm d}z')=\Psi(z,{\rm d}z'):=\delta_{\psi(z)}({\rm d}z')$
and $Q^{*}(z,{\rm d}z')=\Psi^{*}(z,{\rm d}z'):=\delta_{\psi^{-1}(z)}({\rm d}z')$.
\end{prop}
The assumption that $(\nu,Q,Q^{*})$ is a reversible triplet implies
that for any $k\in\mathbb{Z}$, $\Lambda^{0}$ and $\Lambda^{k}$
are equivalent on a suitable restriction of $\mathsf{Z}^{\mathbb{Z}}$,
with a simple Radon--Nikodym derivative involving $\varpi$ only.
This is the property used in Alg.~\ref{alg:MC-proposal-general}
to propose that a chain $\mathtt{Z}$ distributed according to $\Lambda^{0}$
is mapped to a chain distributed according to $\Lambda^{k}$.
\begin{lem}
\label{lemma:markov-chain-change-of-measure}Let $\pi$ and $\nu$
be measures on $(\mathsf{Z},\mathscr{Z})$ where $\pi$ is a probability
and $\nu\gg\pi$ and let $\varpi:={\rm d}\pi/{\rm d}\nu$. Assume
that $(\nu,Q,Q^{*})$ is a reversible triplet. For any $k\in\mathbb{Z}$
let $\Lambda^{k}$ be the two-sided $(k,\pi,Q,Q^{*})$-Markov chain
probability measure and $S_{k}:=\{\mathtt{Z}\in\mathsf{Z}^{\mathbb{Z}}:\varpi(z_{0})\wedge\varpi(z_{k})>0\}$
. Then for any $k\in\mathbb{Z}$ and $\mathtt{Z}\in S_{k}$,
\[
\frac{{\rm d}\Lambda_{S_{k}}^{k}}{{\rm d}\Lambda_{S_{k}}^{0}}(\mathtt{Z})=\frac{\varpi(z_{k})}{\varpi(z_{0})}.
\]
\end{lem}
We can now establish correctness of Alg.~\ref{alg:MC-proposal-general}.
\begin{cor}
For any $\mathtt{Z}\in\mathsf{Z}^{\mathbb{Z}}$ let $\varsigma(\cdot;\mathtt{Z})$
be a probability distribution on $\big(\mathbb{Z},\mathscr{P}(\mathbb{Z})\big)$,
$\xi:=(\mathtt{Z},k)\in\mathsf{Z}^{\mathbb{Z}}\times\mathbb{Z}$,
$\mu({\rm d}\mathtt{Z},k)=\Lambda^{0}({\rm d}\mathtt{Z})\varsigma(k;\mathtt{Z})$
and define the involution $\phi(\mathtt{Z},k):=(\theta^{k}(\mathtt{Z}),-k)$.
Then for $\xi\in S=\left\{ (\mathtt{Z},k):\varpi(z_{0})\wedge\varpi(z_{k})\wedge\varsigma(k;\mathtt{Z})\wedge\varsigma\big(-k;\theta^{k}(\mathtt{Z})\big)>0\right\} $,
we have
\[
\frac{{\rm d}\mu_{S}^{\phi}}{{\rm d}\mu_{S}}(\xi)=\frac{\varpi(z_{k})}{\varpi(z_{0})}\frac{\varsigma\big(-k;\theta^{k}(\mathtt{Z})\big)}{\varsigma(k;\mathtt{Z})},
\]
and apply Theorem~\ref{thm:invo-rev}.
\end{cor}
Alg.~\ref{alg:MC-proposal-general} is in general not practical due
to the requirement of simulation from $\mathcal{Q}(z_{0},\cdot)$,
a prerequisite to sample from $\varsigma(\cdot;\mathtt{Z})$. Key
to this is to make the dependence of $\varsigma(\cdot;\mathtt{Z})$
on $\mathtt{Z}\in\mathsf{Z}^{\mathbb{Z}}$ ``finite'', that is dependent
on a finite number of coordinates of $\mathtt{Z}$ in order to ensure
a finite amount of computation. Numerous options are possible and
we outline two here. The first one is purely deterministic.
\begin{example}
\label{eg:two-deterministic-windows}Let $\tau\in\mathbb{N}$ and
assume that for any $\mathtt{Z}\in\mathsf{Z}^{\mathbb{Z}}$ the probability
$\varsigma(\cdot;\mathtt{Z})$ is entirely determined by the $2\tau+1$
states $z_{-\tau},\ldots,z_{\tau}$ and of support $\llbracket-\tau,\tau\rrbracket$.
In this case simulating $k\sim\varsigma(\cdot;\mathtt{Z})$ only requires
simulation of this subsequence. However in order to compute the acceptance
ratio it is required to simulate what is unrealized in the subsequence
$z_{k-\tau},\ldots,z_{k+\tau}$, that is $z_{i}$ for $i\in\llbracket k-\tau,k+\tau\rrbracket\cap\llbracket-\tau,\tau\rrbracket^{\complement}$.
An example for $\varsigma(\cdot;\mathtt{Z})$ is $\varsigma(k;\mathtt{Z})\propto{\bf 1}_{\llbracket-\tau,\tau\rrbracket}(k)\varpi(z_{k})$,
in which case the acceptance ratio is $\sum_{i=k-\tau}^{k+\tau}\varpi(z_{i})/\sum_{i=-\tau}^{\tau}\varpi(z_{i})$
on $S$.
\end{example}
The above example, in the context of HMC, gives a simple version of
what is described in \citet{neal1994improved}, which can of course
be embellished in various ways. It is also possible to adapt $\tau$
to the realization $\mathtt{Z}$.
\begin{example}
\label{eg:two-deterministic-windows-random}It is possible to make
$\tau$ a function $\mathtt{Z}\mapsto\tau(\mathtt{Z})\in\mathbb{N}$
in Example~\ref{eg:two-deterministic-windows}, more precisely a
stopping time adapted to sequences of the form $z_{-i},z_{-i+1},\ldots,z_{0},\ldots,z_{i-1},z_{i}$
such that with $n=\tau(\mathtt{Z})$, sampling $k\sim\varsigma(\cdot;\mathtt{Z})$
is entirely determined by $z_{-n},\ldots,z_{0},\ldots,z_{n}$. This
leads to the same need for additional simulation i.e. $z_{i}$ for
$i\in\llbracket k-\tau\circ\theta^{k}(\mathtt{Z}),k+\tau\circ\theta^{k}(\mathtt{Z})\rrbracket\cap\llbracket-\tau(\mathtt{Z}),\tau(\mathtt{Z})\rrbracket^{\complement}$
where we notice the need to determine the value of the stopping time
value for the sequence $\theta^{k}(\mathtt{Z})$, also required for
the computation of the acceptance ratio $\sum_{i=k-\tau\circ\theta^{k}(\mathtt{Z})}^{k+\tau\circ\theta^{k}(\mathtt{Z})}\varpi(z_{i})/\sum_{i=-\tau(\mathtt{Z})}^{\tau(\mathtt{Z})}\varpi(z_{i})$
on $S$, for the choice $\varsigma(k;\mathtt{Z})\propto{\bf 1}_{\llbracket-\tau(\mathtt{Z}),\tau(\mathtt{Z})\rrbracket}(k)\varpi(z_{k})$.
\end{example}
In the next section we explore a general technique of ensuring that
both windows of states coincide, therefore leading to simplified algorithms.

\begin{rem}
One could considerably weaken the condition (\ref{eq:Q-Qstar-time-reversal})
in Lemma~\ref{lemma:markov-chain-change-of-measure} to 
\begin{equation}
\nu({\rm d}z)Q(z,{\rm d}z')=\nu^{*}({\rm d}z')Q^{*}(z',{\rm d}z),\label{eq:Q-star-gen}
\end{equation}
where $\nu$ and $\nu^{*}$ are equivalent but not necessarily equal,
at the expense of simplicity. In this case, we obtain for $\mathtt{Z}\in S=\{\mathtt{Z}:\varpi(z_{0})\wedge\varpi(z_{k})>0\}$,
\[
\frac{{\rm d}\Lambda_{S}^{k}}{{\rm d}\Lambda_{S}^{0}}(\mathtt{Z})=\frac{\varpi(z_{k})}{\varpi(z_{0})}F_{k}(\mathtt{Z}),
\]
where 
\[
F_{k}(\mathtt{Z})=\begin{cases}
\prod_{i=1}^{k}\frac{{\rm d}\nu}{{\rm d}\nu^{*}}(z_{i}) & k>0,\\
\prod_{i=k+1}^{0}\frac{{\rm d}\nu^{*}}{{\rm d}\nu}(z_{i}) & k<0,\\
1 & k=0.
\end{cases}
\]
When $\nu=\nu^{\psi}$, that is $\psi$ is $\nu-$preserving, then
${\rm d}\nu^{\psi}/{\rm d}\nu=1$. The generality of (\ref{eq:Q-star-gen})
is natural in the context of deterministic, invertible maps $\psi$
that are not measure preserving but such that $\nu^{\psi}\equiv\nu$.
In particular, the analogue of Proposition~\ref{prop:measure-preserving-reversal}
holds with $\nu^{*}=\nu^{\psi}$. 
\end{rem}
\begin{lem}
\label{lem:non-measure-preserve-map}Let $\psi\colon\mathsf{Z}\rightarrow\mathsf{Z}$
be invertible and such that $\nu^{\psi}\equiv\nu$. Then with $\Psi(z,{\rm d}z')=\delta_{\psi(z)}({\rm d}z')$
and $\Psi^{*}(z,{\rm d}z')=\delta_{\psi^{-1}(z)}({\rm d}z')$ then
with $\nu^{*}=\nu^{\psi}$,
\[
\nu({\rm d}z)\Psi(z,{\rm d}z')=\nu^{*}({\rm d}z')\Psi^{*}(z',{\rm d}z).
\]
\end{lem}
In the specific case that $\nu$ is the Lebesgue measure and $\psi^{-1}$
is a diffeomorphism, then
\[
\frac{{\rm d}\nu^{\psi}}{{\rm d}\nu}(z)=\left|{\rm det}(\psi^{-1})'(z)\right|
\]
is the Jacobian.

\subsection{Doubly-infinite Markov chain proposal and coinciding windows}

In Examples~\ref{eg:two-deterministic-windows} and \ref{eg:two-deterministic-windows-random},
the two windows around $z_{0}$ and $z_{k}$ are typically different
when $k\neq0$. We now explain how to devise an instance of the framework
where the windows around $z_{0}$ and $z_{k}$ are identical by construction.
The main idea consists of introducing an auxiliary variable $\ell$
that can be thought of as determining the left index of the realized
window. To be precise, let $m\in\mathbb{N}$ be the fixed size of
the window to be realized and $\xi:=(\mathtt{Z},\ell,k)\in\mathsf{Z}^{\mathbb{Z}}\times\mathbb{N}\times\mathbb{N}$
where $\ell\sim{\rm Uniform}(\llbracket0,m-1\rrbracket)$ and $k\sim\varsigma(k;\ell,\mathtt{Z})=\varsigma(k;\ell,z_{-\ell},\ldots,z_{r})$
with $r:=m-1-\ell$, that is here $\mu({\rm d}\xi):=\Lambda_{0}({\rm d}\mathtt{Z})\varsigma(k;\ell,\mathtt{Z})\mathbb{I}\big\{\ell\in\llbracket0,m-1\rrbracket\big\}/m$.
Now define the involution
\[
\phi(\xi)=(\theta^{k}(\mathtt{Z}),\ell+k,-k),
\]
then, observing that by construction $\theta^{k}(\mathtt{Z})_{-(\ell+k):(r-k)}=(z_{-\ell},\ldots,z_{r})$,
we obtain the acceptance ratio 
\[
r(\mathtt{Z},\ell,k)=\frac{\varpi(z_{k})\varsigma(-k;\ell+k,z_{-\ell},\ldots,z_{r})}{\varpi(z_{0})\varsigma(k;\ell,z_{-\ell},\ldots,z_{r})},
\]
on $S=\{\xi\colon\varpi(z_{k})\wedge\varpi(z_{0})\wedge\varsigma(-k;\ell+k,z_{-\ell},\ldots,z_{r})\wedge\varsigma(k;\ell,z_{-\ell},\ldots,z_{r})>0\}$.
For example, if $\varsigma(k;\ell,z_{-\ell},\ldots,z_{r})\propto{\bf 1}_{\llbracket-\ell,r\rrbracket}(k)\varpi(z_{k})$
then the acceptance ratio is $1$ for all $k\in\llbracket-\ell,r\rrbracket$
such that $\varsigma(k;\ell,z_{-\ell},\ldots,z_{r})>0$. The resulting
algorithm is presented in Alg.~\ref{alg:MC-proposal-general-coinciding window}.

\begin{algorithm}[H]
\caption{\label{alg:MC-proposal-general-coinciding window}To sample from $P_{{\rm MC}}^{{\rm (general)}}(z_{0},\cdot)$}

\begin{enumerate}
\item Lazily simulate $\mathtt{Z}\sim\mathcal{Q}(z_{0},\cdot)$.
\item Simulate $\ell\sim{\rm Uniform}(\llbracket0,m-1\rrbracket)$ and $k\sim\varsigma(k;\ell,z_{-\ell},\ldots,z_{r})$.
\item With probability
\[
a\left(\frac{\varpi(z_{k})\varsigma(-k;\ell+k,z_{-\ell},\ldots,z_{r})}{\varpi(z_{0})\varsigma(k;\ell,z_{-\ell},\ldots,z_{r})}\right),
\]
output $z_{k}$. Otherwise output $z_{0}$.
\end{enumerate}
\end{algorithm}

In order to introduce NUTS-like kernels, it is helpful at this point
to consider the case where $m=2^{n}$ for some $n\in\mathbb{N}$ and
we shall reparameterize $\ell$ as a sequence of $n$ bits. That is,
we define $b\in\{0,1\}^{n}$ to be a sequence of independent ${\rm Bernoulli}(1/2)$
random variates and write $\ell=\sum_{j=1}^{n}b_{j}2^{j-1}$, so that
the distribution of $\ell$ is indeed ${\rm Uniform}(\llbracket0,m-1\rrbracket)$.
In order to specify the involution in this reparameterization we define
$\beta:\llbracket0,2^{n-1}\rrbracket\to\{0,1\}^{n}$ to be the function
that computes the ``reversed'' binary representation of its input
with $n$ bits, e.g. $\beta(13)=(1,0,1,1,0)$, which has the property
that $\beta\circ\ell(b)=b$. Finally, we specify $\phi(\mathtt{Z},b,k)=(\theta^{k}(\mathtt{Z}),\beta(\ell(b)+k),-k)$.
The intuition is that given $\mathtt{Z}$, the binary string $b$
defines a particular window $z_{-\ell},\ldots,z_{r}$ around $z_{0}$
and for a given $k\in\llbracket-\ell,r\rrbracket$, there is a corresponding
binary string $\beta\big(\ell(b)+k\big)$ that defines the same window,
but around $z_{k}$.

\subsection{NUTS-like kernels}

Let $\pi$ and $\nu$ be measures on $(\mathsf{Z},\mathscr{Z})$ where
$\pi$ is a probability and $\nu\gg\pi$ and let $\varpi:={\rm d}\pi/{\rm d}\nu$. 

\begin{algorithm}
\caption{\label{alg:NUTS}To sample from $P_{{\rm MC}}^{({\rm NUTS})}(z_{0},\cdot)$}

NUTS-like algorithm
\begin{enumerate}
\item Lazily simulate $\mathtt{Z}\sim\mathcal{Q}(z_{0},\cdot)$ and $b$.
\item \label{enu:NUTS-construct-window}Sample $(\ell,r)$:
\begin{enumerate}
\item Set $n\leftarrow0$ and $\ell_{0}\leftarrow r_{0}\leftarrow0$.
\item Set $n\leftarrow n+1$.
\item Set $\ell_{n}\leftarrow\ell_{n-1}+b_{n}2^{n-1}$ and $r_{n}\leftarrow r_{n-1}+(1-b_{n})2^{n-1}$.
\item If $s_{n}(\mathtt{Z},b)=0$, go to 2(b)
\item Set $\ell\leftarrow\ell_{n-1}$, $r\leftarrow r_{n-1}$.
\end{enumerate}
\item Sample $k\sim\varsigma(\cdot\mid\mathtt{Z},\ell,r)$.
\item With probability
\[
a\left(\frac{\varpi(z_{k})\varsigma(-k\mid\theta^{k}(\mathtt{Z}),\ell+k,r-k)}{\varpi(z_{0})\varsigma(k\mid\mathtt{Z},\ell,r)}\right),
\]
output $z_{k}$. Otherwise output $z_{0}$.
\end{enumerate}
\end{algorithm}

Define for $n\in\mathbb{N}$,
\[
\ell_{n}(b):=\sum_{j=1}^{n}b_{j}2^{j-1},\qquad r_{n}(b):=2^{n}-1-\ell_{n}(b).
\]
Let $(s_{n})_{n\in\mathbb{N}}$ be a sequence of functions where $s_{n}\colon\mathsf{Z}^{\mathbb{Z}}\times\{0,1\}^{\mathbb{N}}\rightarrow\{0,1\}$
depends only on windows of states in a way that is made clear below.
For $(\mathtt{Z},b)\in\mathsf{Z}^{\mathbb{N}}\times\{0,1\}^{\mathbb{N}}$
define the stopping time 
\[
\tau(b,\mathtt{Z}):=\inf\{n\geq1:s_{n}(\mathtt{Z},b)=1\}
\]
Specifically, we require that $s_{n}(\mathtt{Z},b)$ is a function
of the vector $(z_{-\ell_{n}(b)},\ldots,z_{r_{n}(b)})$.

Note that $\ell_{n}(b)$ and $r_{n}(b)$ can be computed recursively,
which suggests step \ref{enu:NUTS-construct-window} in Alg.~\ref{alg:NUTS}
where for a sequence random variables $b=(b_{1},b_{2},\ldots)\in\{0,1\}^{\mathbb{N}}$
, $b_{i}\overset{{\rm iid}}{\sim}{\rm Bernoulli}(1/2)$ one finds
$\tau:=\tau(b,\mathtt{Z})$ and the final window is defined by $\ell=\ell_{\tau-1}(b)$
and $r=r_{\tau-1}(b)$, i.e. the most recently added states are ignored.
The reason for this will become clearer below, but is essentially
analogous to the argument in Example~\ref{eg:aimh-stop-nprime-equals-n}.
We now turn to the specification of $(s_{n})_{n\in\mathbb{N}}$. For
$(n,b)\in\mathbb{N}\times\{0,1\}^{\mathbb{N}}$ define $m_{n}(b):=2^{n-1}-1-\ell_{n}(b)$,
so that $\llbracket-\ell_{n}(b),m_{n}(b)\rrbracket$ and $\llbracket m_{n}(b)+1,r_{n}(b)\rrbracket$
are integer sequences of length $2^{n-1}$, one of which is $\llbracket\ell_{n-1}(b),r_{n-1}(b)\rrbracket$:
\begin{itemize}
\item if $b_{n}=0$ then $\ell_{n-1}(b)=\ell_{n}(b)$ and $r_{n-1}(b)=r_{n}(b)-2^{n-1}=2^{n}-1-\ell_{n}(b)-2^{n-1}=m_{n}(b)$,
\item if $b_{n}=1$ then $\ell_{n-1}(b)=\ell_{n}(b)-2^{n-1}=-(m_{n}(b)+1)$
and $r_{n-1}(b)=r_{n}(b)$
\end{itemize}
Let for $n\in\mathbb{N}_{*}$
\[
s_{n}(\mathtt{Z},b)=f_{n-1}(z_{-\ell_{n}(b)},\ldots,z_{m_{n}(b)})\vee f_{n-1}(z_{m_{n}(b)+1},\ldots,z_{r_{n}(b)}),
\]
where $\big\{ f_{k}\colon\mathsf{Z}^{2^{k}}\rightarrow\{0,1\},k\in\mathbb{N}\big\}$
is defined recursively via
\[
f_{k}(\mathfrak{z}_{1:2^{k}})=\begin{cases}
g_{k}\big(\mathfrak{z}_{1:2^{k}}\big)\vee f_{k-1}\big(\mathfrak{z}_{1:2^{k-1}}\big)\vee f_{k-1}\big(\mathfrak{z}_{(2^{k-1}+1):2^{k}}\big) & k\in\mathbb{N_{*}}\\
g_{0}\big({\rm \mathfrak{z}}_{1}\big) & k=0
\end{cases},
\]
for some functions $\big\{ g_{k}:\mathsf{Z}^{2^{k}}\to\{0,1\},k\in\mathbb{N}\big\}$,
which encode the condition for stopping while the functions $f_{k-1}(\mathfrak{z}_{1:2^{k-1}})$
and $f_{k-1}(\mathfrak{z}_{(2^{k-1}+1):2^{k}})$ report whether stopping
was triggered in either of the two main subtrees.
\begin{example}[HMC-NUTS]
Assume the setup of Example~\ref{eg:HMC-setup} that is with $z=(x,v)\in\mathsf{X}\times\mathsf{V}$
we target $\pi(x,v)\propto\gamma(x)\kappa(v)$ and let $Q(z,{\rm d}z')=\Psi(z,{\rm d}z')=\delta_{\psi(z)}({\rm d}z')$
for $\psi$ the leapfrog mapping and assume that the dominating measure
satisfies $\nu({\rm d}z)\Psi(z,{\rm d}z')=\nu({\rm d}z')\Psi^{*}(z',{\rm d}z)$
with $\Psi^{*}(z,{\rm d}z')=\delta_{\psi^{-1}(z)}({\rm d}z')$. A
possible choice is for $k\in\mathbb{N}_{*}$ and $\ell,r\in\mathbb{Z}^{2}$
such that $\ell-r+1=2^{k}$, 
\begin{equation}
g_{k}(z_{\ell}\ldots,z_{r})=\mathbb{I}\left\{ (x_{r}-x_{\ell})^{\top}v_{\ell}>0\right\} \vee\mathbb{I}\left\{ (x_{r}-x_{\ell})^{\top}v_{r}<0\right\} \vee\mathbb{I}\left\{ \max\nolimits _{i,j\in\llbracket\ell,r\rrbracket}\pi(z_{i})/\pi(z_{j})>\Delta_{{\rm max}}\right\} ,\label{eq:gk-for-NUTS}
\end{equation}
for some $\Delta_{{\rm max}}>0$, and $g_{0}\equiv0$. The first two
indicators correspond to the choice made in \citet{hoffman2014no},
the motivation for NUTS--see Appendix~\ref{sec:app-NUTS-motivation}
for some details. The last indicator is our own suggestion to address
numerical errors, since the setup considered by \citet{hoffman2014no}
corresponds to Example~(\ref{eg:sliced-NUTS}) below and numerical
errors are addressed in a slightly different way.
\end{example}
\begin{defn}[Slice sampler \citet{besag1995bayesian,neal2003slice}]
\label{def:slice-sampler}  Given a target distribution $\pi$, with
density $\varpi$ w.r.t. some dominating measure $\nu$, one can define
an extended target distribution $\tilde{\pi}$ via the decomposition
\[
\varpi(z,u)=\varpi(z)\mathbb{I}\{u\leq\varpi(z)\}/\varpi(z)=\mathbb{I}\{u\leq\varpi(z)\},
\]
where $\varpi$ is the density of $\tilde{\pi}$ w.r.t. the product
of $\nu$ and the Lebesgue measure, and in which conditional on $z$,
$u$ is uniformly distributed on $[0,\varpi(z)]$. A MwG Markov kernel
leaving $\tilde{\pi}$ invariant consists of sampling $u$ uniformly
on $[0,\varpi(z)]$ and then applying any Markov kernel leaving the
conditional distribution $\pi_{u}$ of $z$ given $u$ invariant:
the uniform distribution on the ``slice'' $\{z\in\mathsf{Z}:\varpi(z)\geq u\}$.
For the purposes of this work, we may seek to define a sophisticated
$\pi_{u}$-invariant Markov kernel.
\end{defn}
\begin{example}[sliced-HMC-NUTS]
\label{eg:sliced-NUTS} This is what \citet{hoffman2014no} refer
to as the ``simplified'' NUTS algorithm. Here the overall target
distribution has density $\eta(x,v,u)\propto\mathbb{I}\{u\leq\gamma(x)\kappa(v)\}$
and the algorithm consists of a MwG alternating between updating $u$
given $z=(x,v)$ and vice versa. Given $u\in[0,\gamma\otimes\kappa(z)]$,
we focus on sampling from $\pi_{u}(x,v)\propto\mathbb{I}\{u\leq\gamma(x)\kappa(v)\}$.
In this scenario, in addition to (\ref{eq:gk-for-NUTS}) it is suggested
to use, for $k\in\mathbb{N}_{*}$ and $\ell,r\in\mathbb{Z}^{2}$ such
that $\ell-r+1=2^{k}$, 
\[
g_{k}(z_{\ell}\ldots,z_{r})=\mathbb{I}\left\{ (x_{r}-x_{\ell})^{\top}v_{\ell}>0\right\} \vee\mathbb{I}\left\{ (x_{r}-x_{\ell})^{\top}v_{r}<0\right\} ,
\]
and
\[
g_{0}(z)=\mathbb{I}\left\{ \log\gamma\otimes\kappa(z)<\log u-\Delta_{{\rm max}}\right\} ,
\]
for $\Delta_{{\rm max}}\geq0$ in order to stop computation when the
error arising from the numerical integration of Hamilton's dynamic
leads to an ``astronomically'' large error.
\end{example}
\begin{lem}
\label{lem:sn-properties-1}The functions $\{s_{n},n\in\mathbb{N}\}$
satisfy for any $(n,\mathtt{Z},b)\in\mathbb{N}\times\mathsf{Z}^{\mathbb{Z}}\times\{0,1\}^{\mathbb{N}_{*}}$
\begin{enumerate}
\item \label{enu:lem:sn-enu1}$s_{n}(\mathtt{Z},b)$ depends only on the
order and the values of $z_{-\ell_{n}(b)},\ldots,z_{r_{n}(b)}$, and
not on how they are indexed;
\item \label{enu:lem:sn-enu2}$s_{n}(\mathtt{Z},b)\geq s_{n-1}(\mathtt{Z},b)$
for $n\geq2$,
\item \label{enu:lem:sn-enu3}$s_{n}(\mathtt{Z},b)\geq g_{n-1}(z_{-\ell_{n-1}(b)},\ldots,z_{r_{n-1}(b)})$.
\end{enumerate}
\end{lem}
\begin{rem}
Part~\ref{enu:lem:sn-enu3} of Lemma~\ref{lem:sn-properties-1}
is useful for computational reasons: if one observes that $s_{n}(\mathtt{Z},b)=0$
but $g_{n}(z_{-\ell_{n}(b)},\ldots,z_{r_{n}(b)})=1$ then one can
stop and take $\tau(\mathtt{Z},b)=n+1$ without further simulation
being needed.
\end{rem}
Let $\beta_{k}\colon\llbracket0,2^{k}-1\rrbracket\rightarrow\{0,1\}^{k}$
be the reversed binary representation of $i\in\llbracket0,2^{k}-1\rrbracket$
with $k$ bits, so that e.g., $\beta_{5}(13)=(1,0,1,1,0)$. This function
has the property that $\beta_{k}\circ\ell_{k}(b)=b_{1:k}$. For $b\in\{0,1\}^{\mathbb{N}}$,
define $\bar{b}_{i}:=(b_{i+1},b_{i+2},\ldots)$ for $i\in\mathbb{N}$.
We can now present the result that allows one to relate the stopped
Markov processes, which is analogous to Lemma~\ref{lem:stop-IMH}
in this setting.
\begin{lem}
\label{lem:nuts-tau-equal-1}Let $n\in\mathbb{N}$, $\ell,r\in\mathbb{N}_{*}$
be such that $\ell+r+1=2^{n-1}$, $k\in\llbracket-\ell,r\rrbracket$,
$\mathtt{Z}\in\mathsf{Z}^{\mathbb{Z}}$ and $\mathtt{Z}':=\theta^{k}(\mathtt{Z})$.
Then for any $b=(b_{1:n-1},\bar{b}_{n-1})\in\{0,1\}^{\mathbb{N}_{*}}$
such that $\tau(\mathtt{Z},b)=n$, $\ell_{n-1}(b)=\ell$ and $r_{n-1}(b)=r$,
there exists a unique $b'_{1:n-1}=\beta_{n-1}(\ell+k)$ such that
with $b'=(b_{1:n-1}',\bar{b}_{n-1})$, 
\begin{enumerate}
\item $\ell_{n-1}(b')=\ell+k$, 
\item $r_{n-1}(b')=r-k$, 
\item $(z'_{-(\ell+k)},\ldots,z'_{r-k})=(z{}_{-\ell},\ldots,z{}_{r})$,
\item and $\tau(\mathtt{Z}',b')=n$.
\end{enumerate}
\end{lem}
\begin{cor}
For $n\in\mathbb{N}_{*}$ and $k\in\llbracket-\ell_{n-1}(b),r_{n-1}(b)\rrbracket$
let $\chi_{n-1,k}\colon\{0,1\}^{\mathbb{N}_{*}}\rightarrow\{0,1\}^{\mathbb{N}_{*}}$
such that $\chi_{n-1,k}(b):=\big(\beta_{n-1}(\ell_{n-1}(b)+k),\bar{b}\big)$.
Then $(b,k)\mapsto(\chi_{n-1,k}(b),-k)$ is an involution since $b'_{1:n-1}=\beta_{n-1}\big(\ell_{n-1}(b)+k\big)$,
$\ell_{n-1}(b')=\ell_{n-1}(b)+k$ and $b_{1:n-1}=\beta_{n-1}\big(\ell_{n-1}(b)+k-k\big)$. 
\end{cor}
For $n\in\mathbb{N}_{*}$ let $\mathfrak{s}(n,b)\colon\mathsf{Z}^{\mathbb{Z}}\times\mathbb{N}\times\{0,1\}^{\mathbb{N}_{*}}\rightarrow\{0,1\}$
\[
\mathfrak{s}(\mathtt{Z},n,b):=s_{n}(\mathtt{Z},b)\prod_{i=1}^{n-1}\left[1-s_{i}(\mathtt{Z},b)\right],
\]
 $\xi:=(\mathtt{Z},n,b,k)\in\mathsf{Z}^{\mathbb{Z}}\times\mathbb{N}\times\{0,1\}^{\mathbb{N}}\times\mathbb{Z}$
and
\[
\mu({\rm d}\mathtt{Z},n,b,k):=\Lambda^{0}({\rm d}\mathtt{Z})\Upsilon(b)\mathfrak{s}(\mathtt{Z},n,b)\varsigma\big(k;\ell_{n-1}(b),r_{n-1}(b),\mathtt{Z}\big).
\]

From Lemma~\ref{lem:nuts-tau-equal-1} for $(n,b)\in\mathbb{N}\times\{0,1\}^{\mathbb{N}}$
$\mathfrak{s}(n,b)=1$ implies that $\mathfrak{s}(n,\chi_{n-1,k}(b))=1$
for any $k\in\llbracket-\ell_{n-1}(b),r_{n-1}(b)\rrbracket$. Then
with the involution 
\[
\phi(\mathtt{Z},n,b,k)=(\theta^{k}(\mathtt{Z}),n,\chi_{n-1,k}(b),-k)
\]
we define
\begin{align*}
S:=\{ & (\mathtt{Z},n,b,k)\in\Xi:\varpi(z_{k})\wedge\varpi(z_{0})\wedge\mathfrak{s}(\mathtt{Z},n,b)\wedge\mathfrak{s}\big(\mathtt{Z},n,\chi_{n-1,k}(b)\big)\\
 & \wedge\varsigma(k;\ell_{n-1}(b),r_{n-1}(b),\mathtt{Z})\wedge\varsigma(-k;\ell_{n-1}(b)+k,r_{n-1}(b)-k,\theta^{k}(\mathtt{Z}))>0\},
\end{align*}
Now for any $\xi\in S$ we obtain an expression for the acceptance
ratio
\begin{align*}
r(\xi) & =\frac{{\rm d}\Lambda^{k}}{{\rm d}\Lambda^{0}}(\mathtt{Z})\frac{\varsigma(-k;\theta^{k}(\mathtt{Z}),\ell+k,r-k)}{\varsigma(k;\mathtt{Z},\ell,r)}\\
 & =\frac{\varpi(z_{k})}{\varpi(z_{0})}\frac{\varsigma(-k;\theta^{k}(\mathtt{Z}),\ell+k,r-k)}{\varsigma(k;\mathtt{Z},\ell,r)}.
\end{align*}
A natural choice of $\varsigma(k;\mathtt{Z},\ell,r)$ is $\varsigma(k;\mathtt{Z},\ell,r)\propto\varpi(z_{k})\mathbb{I}\{k\in\llbracket-\ell,r\rrbracket\}$,
in which case $r(\mathtt{Z},\ell,r,k)=1$ for any $k\in\llbracket-\ell,r\rrbracket$
(This is the same in the slice setting, where it corresponds to choosing
uniformly from points in the slice). One can always improve this slightly
(Peskun) by excluding $k=0$, and having an acceptance ratio that
is not $1$ in general. That is, taking 
\[
\varsigma(k;\mathtt{Z},\ell,r)\propto\begin{cases}
\varpi(z_{k})\mathbb{I}\{k\in\llbracket-\ell,r\rrbracket\setminus\{0\}) & \exists i\in\llbracket-\ell,r\rrbracket\setminus\{0\}:\varpi(z_{i})>0,\\
\mathbb{I}\{k=0\} & \text{otherwise}.
\end{cases}
\]
in which case $r(\xi)=1$ for $\xi\in S$. Our understanding from
\citet{betancourt2017conceptual}, is that Stan uses this (unsliced)
``multinomial'' (i.e. categorical) sampling, but the exact expression
for $\varsigma(k;\mathtt{Z},\ell,r)$ is not clear.

\section{Multiple-try Metropolis and related schemes\label{sec:MTM}}

\subsection{\label{subsec:Standard-MTM}Standard MTM}

A simple multiple-try Metropolis (MTM) kernel \citep{liu2000multiple}
involves $n\in\mathbb{N}$ proposals conditional upon the input, from
which one is chosen as a candidate to move to. The acceptance probability
then involves simulating $n-1$ proposals from this candidate. The
kernel is presented in Alg.~\ref{alg:Standard-MTM-kernel}. We can
write $\xi=(\mathtt{Z}_{1},\mathtt{Z}_{2},k,\ell)$ where $\mathtt{Z}_{1},\mathtt{Z}_{2}\in\mathsf{Z}^{n}$,
$k,\ell\in\llbracket1,n\rrbracket$, and here $\xi_{0}=z_{1,k}$.
For simplicity we will assume that the target $\pi$ and proposals
$Q(z,\cdot)$, $z\in\mathsf{Z}$, have densities $\varpi$ and $q(z,\cdot)$
w.r.t. a common reference measure. We can write, 
\[
\rho(\mathtt{Z}_{1},\mathtt{Z}_{2},k,\ell)=\frac{\mathbb{I}\{k\in\llbracket1,n\rrbracket\}}{n}\varpi(z_{1,k})\left\{ \prod_{i=1}^{n}q(z_{1,k},z_{2,i})\right\} \varsigma(\ell;z_{1,k},\mathtt{Z}_{2})\left\{ \prod_{i=1,i\neq k}^{n}q(z_{2,\ell},z_{1,i})\right\} .
\]
The associated involution is $\phi(\mathtt{Z}_{1},\mathtt{Z}_{2},k,\ell):=(\mathtt{Z}_{2},\mathtt{Z}_{1},\ell,k)$,
so that the acceptance ratio is, for $\xi\in S$
\[
r(\xi)=\frac{\rho\circ\phi}{\rho}(\xi)=\frac{\varpi(z_{2,\ell})q(z_{2,\ell},z_{1,k})\varsigma(k;z_{2,\ell},\mathtt{Z}_{1})}{\varpi(z_{1,k})q(z_{1,k},z_{2,\ell})\varsigma(\ell;z_{1,k},\mathtt{Z}_{2})}.
\]
In practice, one often chooses $\varsigma(\ell;z_{1,k},\mathtt{Z}_{2})\propto w(z_{2,\ell},z_{1,k})$
where $w$ is a weight function. In particular, \citet{liu2000multiple}
suggest to use
\[
w(z,z')=\varpi(z)q(z,z')\lambda(z,z'),
\]
where $\lambda(z,z')=\lambda(z',z)$ for all $z,z'\in\mathsf{Z}$.
Then the acceptance ratio can be expressed as
\[
r(\xi)=\frac{\sum_{i=1}^{n}w(z_{2,i},z_{1,k})}{\sum_{i=1}^{n}w(z_{1,i},z_{2,\ell})},
\]
for $\xi\in S$. To illustrate, a possible choice is to take $\lambda(z,z')=\left\{ q(z,z')q(z',z)\right\} ^{-1}$,
in which case $w(z,z')=\varpi(z)/q(z',z)$.

\begin{algorithm}
\caption{\label{alg:Standard-MTM-kernel}Standard MTM kernel}

\begin{enumerate}
\item Given $z$, sample $k\sim{\rm Uniform}\big(\llbracket1,n\rrbracket\big)$
and set $z_{1,k}=z$
\item Sample $z_{2,i}\overset{{\rm iid}}{\sim}Q(z_{1,k},\cdot)$ for $i\in\llbracket n\rrbracket$
\item Sample $\ell\sim\varsigma(\cdot;z_{1,k},\mathtt{Z}_{2})$
\item Sample $z_{1,i}\overset{{\rm iid}}{\sim}Q(z_{2,\ell},\cdot)$ for
$i\in\llbracket n\rrbracket\setminus\{k\}$
\item With probability
\[
a\left(\frac{\varpi(z_{2,\ell})q(z_{2,\ell},z_{1,k})\varsigma(k;z_{2,\ell},\mathtt{Z}_{1})}{\varpi(z_{1,k})q(z_{1,k},z_{2,\ell})\varsigma(\ell;z_{1,k},\mathtt{Z}_{2})}\right)
\]
return $z_{2,\ell}$, otherwise $z_{1,k}$.
\end{enumerate}
\end{algorithm}

\subsection{Stopping time MTM\label{subsec:Stopping-time-MTM}}

We consider now locally adaptive selection of the number of samples
$n$ in MTM. In particular, the approach taken in Section~\ref{subsec:Standard-MTM}
needs to be slightly adapted and then the stopping time random variables
introduced, one for each of the $\mathtt{Z}_{1}:=(z_{1,1},z_{1,2},\ldots)\in\mathsf{Z}^{\mathbb{N}_{*}}$
samples and the $\mathtt{Z}_{2}:=(z_{2,1},z_{2,2},\ldots)\in\mathsf{Z}^{\mathbb{N}_{*}}$
samples and for $(m,n)\in\{1,2\}\times\mathbb{N_{*}}$ let $\mathtt{Z}_{m,n}:=(z_{m,1},\ldots,z_{m,n})$.
The kernel is presented in Alg.~\ref{alg:Locally-adaptive-MTM}.
We define $\xi:=(\mathtt{Z}_{1},\mathtt{Z}_{2},m,n,k,\ell)$ where
$m,n\in\mathbb{N}$, $(\mathtt{Z}_{1},\mathtt{Z}_{2})\in\mathsf{Z}^{\mathbb{N_{*}}}\times\mathsf{Z}^{\mathbb{N}_{*}}$
and $(k,\ell)\in\llbracket m\rrbracket\times\llbracket n\rrbracket$.
We let $\xi_{0}=z_{1,1}$. Let $\sigma_{k}:\mathsf{Z}^{\mathbb{N}_{*}}\to\mathsf{Z}^{\mathbb{N}_{*}}$
be the swapping function such that, with $\mathtt{Z}'=\sigma_{k}(\mathtt{Z})$,
$z_{1}'=z_{k}$, $z_{k}'=z_{1}$ and $z'_{j}=z_{j}$ for $j\notin\{1,k\}$.
For any $i\in\mathbb{N}_{*}$ let $s_{i}:\mathsf{Z}\times\mathsf{Z}^{\mathbb{N}_{*}}\to\{0,1\}$
be such that $i\mapsto s_{i}(z,\mathtt{Z}')$ is non-decreasing and
$s_{i}(z,\mathtt{Z}')=s_{i}\big(z,\sigma_{\ell}(\mathtt{Z}')\big)$
for any $\ell\in\llbracket i-1\rrbracket$. For example, one could
choose $s_{i}(z,\mathtt{Z}')=\mathbb{I}\left\{ \sum_{j=1}^{i}w(z,z'_{i})\geq c\right\} $
for all $i\in\mathbb{N}_{*}$ and for some $c>0$, where $w$ is a
weight function as described in the previous subsection. 

The ``forward'' stopping times of interest are, for $z\in\mathsf{Z}$
and $\mathtt{Z}_{1},\mathtt{Z}_{2}\in\mathsf{Z}^{\mathbb{N}}$
\[
\tau_{1}(z,\mathtt{Z}_{1})=\inf\{n\geq1:s_{n}(z,\mathtt{Z}_{1})=1\}\text{ and }\tau_{2}(z,\mathtt{Z}_{2})=\inf\{n\geq1:s_{n}(z,\mathtt{Z}_{2})=1\},
\]
and we define the $\{0,1\}$-valued functions, for $n\in\mathbb{N}_{*}$,
\[
\mathfrak{s}(n,z,\mathtt{Z})=s_{n}(z,\mathtt{Z})\prod_{i=1}^{n-1}\big[1-s_{i}(z,\mathtt{Z})\big].
\]
For $i\in\{1,2\}$ the quantity $\mathfrak{s}(n,z_{3-i,1},\mathtt{Z}_{i})$
can be thought of as the probability that $\tau_{i}=n$ given the
values $z_{3-i,1}$ and $\mathtt{Z}_{i,n}$. We define $\mathcal{Q}_{2}(z,{\rm d}\mathtt{Z})$
to correspond to the distribution of $z_{i}\overset{{\rm iid}}{\sim}Q(z,\cdot)$
for $i\in\mathbb{N}_{*}$ and $\mathcal{Q}_{1}(z,{\rm d}\mathtt{Z})$
such that for $i\geq2$ $z_{i}\overset{{\rm iid}}{\sim}Q(z,\cdot)$
and $\mathcal{Q}_{1}(z,{\rm d}z_{1})=\pi({\rm d}z_{1})$, that is
under $\mu$ below $z_{1,1}\sim\pi$, 
\begin{align*}
\mu\big({\rm d}(\mathtt{Z}_{1},\mathtt{Z}_{2}),m,n,k,\ell\big) & :=\mathcal{Q}_{2}(z_{1,1},{\rm d}\mathtt{Z}_{2})\mathfrak{s}(n,z_{1,1},\mathtt{Z}_{2})\\
 & \hspace{2cm}\varsigma(\ell;z_{1,1},\mathtt{Z}_{2},n-1)\mathcal{Q}_{1}(z_{2,\ell},{\rm d}\mathtt{Z}_{1})\mathfrak{s}(m,z_{2,\ell},\mathtt{Z}_{1})\frac{\mathbb{I}\big\{ k\in\llbracket m-1\rrbracket\big\}}{m-1},
\end{align*}
which is the distribution of a process that simulates the stopped
processes described above, and chooses $\ell$ and $k$, respectively,
from a categorical distribution on $\llbracket n-1\rrbracket$ and
a uniform distribution on $\llbracket m-1\rrbracket$. We define $\phi(\mathtt{Z}_{1},\mathtt{Z}_{2},m,n,k,\ell)=(\sigma_{\ell}(\mathtt{Z}_{2}),\sigma_{k}(\mathtt{Z}_{1}),n,m,\ell,k)$.
It is straightforward to verify that $\phi$ is an involution. What
is more interesting is, assuming densities as in the previous subsection,
that for $\xi\in S$
\[
r(\xi)=\frac{\varpi(z_{2,\ell})q(z_{2,\ell},z_{1,k})\varsigma(k;z_{2,\ell},\sigma_{k}(\mathtt{Z}_{1}),m-1)}{\varpi(z_{1,1})q(z_{1,k},z_{2,\ell})\varsigma(\ell;z_{1,1},\mathtt{Z}_{2},n-1)},
\]
one can apply Lemma~\ref{lem:stop-IMH} (reindexing to take into
account 1-indexing as opposed to 0-indexing) twice to determine that
$\mathfrak{s}(n,z_{1,1},\mathtt{Z}_{2})=1$ implies $\mathfrak{s}\big(n,z_{1,k},\sigma_{\ell}(\mathtt{Z}_{2})\big)=1$
and $\mathfrak{s}(m,z_{2,\ell},\mathtt{Z}_{1})=1$ implies $\mathfrak{s}\big(m,z_{2,\ell},\sigma_{k}(\mathtt{Z}_{1})\big)=1$.

\begin{algorithm}
\caption{\label{alg:Locally-adaptive-MTM}Locally adaptive MTM kernel}

\begin{enumerate}
\item Given $z$ set $z_{1,1}=z$
\item Sample $z_{2,i}\overset{{\rm iid}}{\sim}Q(z_{1,1},\cdot)$ for $i\geq1$
lazily and obtain $n=\tau_{2}(z_{1,1},\mathtt{Z}_{2})$
\item Sample $\ell\sim\varsigma(\cdot;z_{1,1},\mathtt{Z}_{2},n)$
\item Sample $z_{1,i}\overset{{\rm iid}}{\sim}Q(z_{2,\ell},\cdot)$ for
$i\geq2$ lazily and obtain $m=\tau_{1}(z_{2,\ell},\mathtt{Z}_{1})$
\item Sample $k\sim{\rm Uniform}\big(\llbracket m-1\rrbracket\big)$.
\item With probability
\[
a\left(\frac{\varpi(z_{2,\ell})q(z_{2,\ell},z_{1,k})\varsigma(k;z_{2,\ell},\sigma_{k}(\mathtt{Z}_{1}),m-1)}{\varpi(z_{1,1})q(z_{1,k},z_{2,\ell})\varsigma(\ell;z_{1,1},\mathtt{Z}_{2},n-1)}\right)
\]
return $z_{2,\ell}$, otherwise $z_{1,1}$.
\end{enumerate}
\end{algorithm}

\begin{rem}
It is possible, of course, to specify functions $s_{i}$ that do not
satisfy the conditions above. In this case, the reverse and the forward
stopping time probabilities are not necessarily equal, and their ratios
will appear in the acceptance ratio.
\end{rem}

\subsection{Pseudo-marginal algorithms}

It is relatively straightforward to adapt the MTM kernels above to
the pseudo-marginal setting \citep{Beaumont1139,andrieu2009pseudo}.
It is also possible to extend the example below to the situation where
one uses stopping times to determine the number of simulations, and
also to the particle MCMC \citep{andrieu2010particle} setting, as
is done in \citet{Lee2011} which also contains an earlier version
of the stopping time framework detailed in Section~\ref{subsec:Stopping-time-MTM}.
\citet{lee2012choice} and \citet{moral2015alive} provide some examples
of each in simple scenarios.
\begin{example}[Pseudo-marginal MTM]
\label{eg:pm-mtm}In particular, in this setting one targets a distribution
with density $\varpi(z)$ w.r.t. some measure $\nu$ where $\varpi:\mathsf{Z}\to\mathbb{R}_{+}$
cannot be calculated but for any $z\in\mathsf{Z}$ one can simulate
a random variable $w\sim Q_{z}$ with expectation $\varpi(z)$. We
introduce the auxiliary distribution with density
\[
\pi({\rm d}z,{\rm d}w)=\nu({\rm d}z)wQ_{z}({\rm d}w),
\]
such that $\pi({\rm d}z)=\nu({\rm d}z)\int wQ_{z}({\rm d}w)=\nu({\rm d}z)\varpi(z)$.
 Letting $Q$ be a Markov kernel evolving on $\mathsf{Z}$, $\mathtt{W}=(w_{1},\ldots,w_{n})\in\mathbb{R}^{n}$
and $\mathtt{W}'=(w_{1}',\ldots,w_{n}')\in\mathbb{R}^{n}$ we consider
the choice $\xi=(z,z',\mathtt{W},\mathtt{W}',k,\ell)\in\mathsf{Z}\times\mathsf{Z}\times\mathbb{R}^{n}\times\mathbb{R}^{n}$,
$\xi_{0}=(z,w_{k})$ and
\[
\mu({\rm d}z,{\rm d}z',{\rm d}\mathtt{W},{\rm d}\mathtt{W}',k,\ell)=\frac{\mathbb{I}\{k\in\llbracket n\rrbracket\}}{n}\nu({\rm d}z)w_{k}Q_{z}^{\otimes n}({\rm d}\mathtt{W})Q(z,{\rm d}z')Q_{z'}^{\otimes n}({\rm d}\mathtt{W}')\varsigma(\ell;\mathtt{W}'),
\]
where $\varsigma(k;\mathtt{W})\propto w_{k}\mathbf{1}_{\left\llbracket n\right\rrbracket }(k)$.
The involution can be chosen to be $\phi(z,z',\mathtt{W},\mathtt{W}',k,\ell)=(z',z,\mathtt{W}',\mathtt{W},\ell,k)$,
giving the acceptance ratio
\[
r(\xi)=\frac{p(z')q(z',z)}{p(z)q(z,z')}\cdot\frac{\sum_{i=1}^{n}w'_{i}}{\sum_{i=1}^{n}w{}_{i}},
\]
where for simplicity we assume that $\{Q(z,\cdot);z\in\mathsf{Z}\}$
and $\nu$ have densities $\{q(z,\cdot);z\in\mathsf{Z}\}$ and $p$
w.r.t. some dominating reference measure. We can view the averages
of the $w_{i}$ (resp. $w'_{i}$) as approximations of $\varpi(z)$
(resp. $\varpi(z')$) and the value of $n$ controls the variability
of the approximation. The corresponding Markov kernel is subtly different
from the standard pseudo-marginal approach, in that here one simulates
$k\sim{\rm Uniform}(\left\llbracket n\right\rrbracket )$ and $\mathtt{W}_{-k}\sim Q_{z}^{\otimes n-1}$
rather than having these variables fixed. In some sense, one can view
the standard pseudo-marginal kernel as a MwG approach where one fixes
$(k,\mathtt{W})$, rather than only fixing $\xi_{0}$.
\end{example}
The next example is an interesting variant in which a shared stopping
time is defined, and which has been shown to inherit desirable properties
from the limiting MH kernel associated with the pair $(\pi,Q)$ as
$n\to\infty$ but which would naturally require computation of $f$
under conditions where the kernel of Example~\ref{eg:pm-mtm} with
any fixed $n$ would not \citep{lee2014variance}.
\begin{example}[One-hit kernel of \citealp{Lee2012}]
Consider the setting of Example~\ref{eg:pm-mtm} but where $Q_{z}$
is a ${\rm Bernoulli}(\varpi(z))$ distribution with $\varpi:\mathsf{Z}\to[0,1]$.
In this case, $w=1$ under $\pi$. We want here to adapt the number
of simulations so that the acceptance ratio is a reasonable approximation
of the limiting acceptance ratio as $n\to\infty$, but does not require
an excessive number of simulations. In particular, using a fixed number
of simulations may lead to acceptance ratios with a large variance
and hence a Markov chain that can get ``stuck'' for long periods
when in regions of the state space with very small $\varpi(z)$. Let
$\xi=(z,z',\mathtt{W},\mathtt{W}',n)\in\mathsf{Z}\times\mathsf{Z}\times\{0,1\}^{\mathbb{N}_{*}}\times\{0,1\}^{\mathbb{N}_{*}}\times\mathbb{N}_{*}$,
$\xi_{0}=(z,w_{1})$ and let for $z\in\mathsf{Z}$, $F_{z}$ be the
probability measure associated with an infinite sequence of independent
$Q_{z}$-distributed random variables. The idea is given $z,z'\in\mathsf{Z}$
and $w_{1}=1$ we wish to simulate $w_{i}\overset{\mathrm{iid}}{\sim}Q_{z}$,
$i\in\mathbb{N},i\geq2$ and $w_{i}'\overset{\mathrm{iid}}{\sim}Q_{z'}$,
$i\in\mathbb{N}_{*}$ independently until there is one further ``hit'',
i.e. $w_{i}$ and/or $w_{i}'$ is equal to $1$. So we define
\[
\tau(\xi):=\inf\{n\geq1:s_{n}(\mathtt{W},\mathtt{W}')=1\},
\]
where $s_{n}(\mathtt{W},\mathtt{W}')=\mathbb{I}\left\{ \sum_{i=1}^{n}w_{i}+\sum_{i=1}^{n}w_{i}'\geq2\right\} $.
We then define
\[
\mu({\rm d}z,{\rm d}z',{\rm d}\mathtt{W},{\rm d}\mathtt{W}',n):=\nu({\rm d}z)w_{1}Q(z,{\rm d}z')F_{z}({\rm d}\mathtt{W})F_{z'}({\rm d}\mathtt{W}')\mathfrak{s}(n,\mathtt{W},\mathtt{W}'),
\]
where for $n\in\mathbb{N}_{*}$, $\mathfrak{s}(n,\mathtt{W},\mathtt{W}')=s_{n}(\mathtt{W},\mathtt{W}')\prod_{i=1}^{n-1}[1-s_{i}(\mathtt{W},\mathtt{W}')]$
so that if $\xi=(z,z',\mathtt{W},\mathtt{W}',n)\sim\mu$ then $n=\tau(\xi)$.
We define the involution 
\[
\phi(z,z',\mathtt{W},\mathtt{W}',n)=(z',z,\sigma_{n}(\mathtt{W}'),\sigma_{n}(\mathtt{W}),n),
\]
where for $i\in\mathbb{N}_{*}$, $\sigma_{i}\colon\{0,1\}^{\mathbb{N}_{*}}\rightarrow\{0,1\}^{\mathbb{N}_{*}}$
is the permutation that swaps its $1-$st and $i-$th inputs. We can
obtain, for $\xi$ such that $p(z)q(z,z')>0$, 
\[
r(\xi)=\frac{p(z')q(z',z)}{p(z)q(z,z')}\mathbb{I}\left\{ w_{1}=w'_{n}=1,\tau(\xi)=\tau\circ\phi(\xi)=n\right\} ,
\]
i.e. it is essential that $w'_{n}=1$ and that the stopping time is
preserved by the involution. We observe that if $\tau(\xi)=1$, then
necessarily $w'_{1}=1$ and $\tau\circ\phi(\xi)=1=\tau(\xi)$, so
the indicator above is $1$. Now consider $\tau(\xi)>1$. If $w_{\tau(\xi)}=0$
then $(w_{1},w_{2},\ldots,w_{\tau(\xi)})=(1,0,\ldots,0)$, while necessarily
$(w'_{1},\ldots,w_{\tau(\xi)}')=(0,\ldots,0,1)$ so $\tau\circ\phi(\xi)=\tau(\xi)$
and the indicator above is $1$. If on the other hand $w_{\tau(\xi)}=1$
then $(w_{1},w_{2},\ldots,w_{\tau(\xi)})=(1,0,\ldots,1)$ so even
if $w'_{\tau(\xi)}=1$ we have $\tau\circ\phi(\xi)=1\neq\tau(\xi)$
and so $r(\xi)=0$. 
\end{example}

%% file: dr.tex
\section{Delayed rejection\label{sec:Delayed-rejection}}

In delayed rejection, several sources of randomness and involutions
are considered in turn until one is accepted.

\subsection{Stochastic delayed rejection\label{subsec:Stochastic-delayed-rejection}}

Let $\pi$ be a probability distribution on $(\mathsf{Z}_{0},\mathscr{Z}_{0})$.
For $k\in\mathbb{N}_{*}$ let $(\mathsf{Z}_{k},\mathscr{Z}_{k})$
be measurable spaces, for $\mathtt{Z}^{k-1}\in\mathsf{Z}^{k-1}:=\mathsf{Z}_{0}\times\mathsf{Z}_{1}\times\cdots\mathsf{Z}_{k-1}$
let $Q_{k}(\mathtt{Z}^{k-1},\cdot)$ be probability distributions
on $(\mathsf{Z}_{k},\mathscr{Z}_{k})$. Define $\eta_{k}({\rm d}\mathtt{Z}^{k}):=\pi({\rm d}z_{0})\prod_{i=1}^{k}Q_{i}(\mathtt{Z}^{i-1};{\rm d}z_{i})$,
for $k\in\mathbb{N}_{*}$ let $\phi_{k}\colon\mathsf{Z}^{k+1}\rightarrow\mathsf{Z}^{k+1}$
be involutions and let $\alpha_{k}=a_{k}\circ r_{k}$ for an acceptance
function $a_{k}$ and $\beta_{k}(\mathtt{Z}^{k})=\beta_{k-1}(\mathtt{Z}^{k-1})[1-\alpha_{k-1}(\mathtt{Z}^{k-1})]$
with $\beta_{0}\equiv1$ and $\alpha_{0}\equiv0$ where 
\[
r_{k}(\mathtt{Z}^{k}):=\begin{cases}
\frac{\beta_{k}\circ\phi_{k}}{\beta_{k}}(\mathtt{Z}^{k})\frac{{\rm d}\eta_{k,S_{k}}^{\phi_{k}}}{{\rm d}\eta_{k,S_{k}}}(\mathtt{Z}^{k}) & \mathtt{Z}^{k}\in S_{k}\\
0 & \text{otherwise}
\end{cases},
\]
with $S_{k}:=S(\eta_{k},\eta_{k}^{\phi_{k}})\cap\{\mathtt{Z}^{k}\in\mathsf{Z}^{k}\colon\beta_{k}(\mathtt{Z}^{k})\wedge\beta_{k}\circ\phi_{k}(\mathtt{Z}^{k})>0\}$,
where $S(\eta_{k},\eta_{k}^{\phi_{k}})$ is as in Theorem~\ref{thm:invo-rev}.
The delayed rejection algorithm is described in Alg. \ref{alg:Stochastic-delayed-rejection}
and its justification follows from the following:
\begin{prop}
\label{prop:DR-random}With the notation above, define the probability
distribution of marginal $\pi$,
\begin{align*}
\mu(k,{\rm d}\mathtt{Z}^{k}): & =\alpha_{k}(\mathtt{Z}^{k})\beta_{k}(\mathtt{Z}^{k})\eta_{k}({\rm d}\mathtt{Z}^{k}),
\end{align*}
on $E=\{(k,\mathtt{Z}^{k})\in\{k\}\times\mathsf{Z}^{k}:k\in\mathbb{N}\}$
and for any $\xi=(k,\mathtt{Z}^{k})\in E$ the involution $\phi(\xi)=\phi(k,\mathtt{Z}^{k})=\big(k,\phi_{k}(\mathtt{Z}^{k})\big)$.
Then,
\[
r(\xi)=\begin{cases}
1 & \xi\in S(\mu,\mu^{\phi})\\
0 & \text{otherwise }
\end{cases}.
\]
 
\end{prop}
\begin{algorithm}
\caption{Stochastic delayed rejection \label{alg:Stochastic-delayed-rejection}}

\begin{enumerate}
\item Given $z\in\mathsf{Z}$, set $k\leftarrow0$ and $z_{0}\leftarrow z$.
\item Set $k\leftarrow k+1$ and simulate $z_{k}\sim Q_{k}(\mathtt{Z}^{k-1},\cdot)$.
\item With probability $\alpha_{k}(\mathtt{Z}^{k})$ output $\phi_{k}(\mathtt{Z}^{k})_{0}$,
otherwise go to 2.
\end{enumerate}
\end{algorithm}

Although $\beta_{k}(\mathtt{Z}^{k})$ can be updated as the algorithm
progresses, the computation of $\beta_{k}\circ\phi_{k}(\mathtt{Z}^{k})$
can be expensive. Indeed, letting $\zeta^{k}:=\phi_{k}(\mathtt{Z}^{k})$
for $k\in\mathbb{N}_{*}$, we have
\[
\beta_{k}\circ\phi_{k}(\mathtt{Z}^{k})=\beta_{k}(\zeta^{k})=\prod_{i=1}^{k-1}[1-\alpha_{i}(\zeta^{i})],
\]
where for $i\in\llbracket k-1\rrbracket$ $\alpha_{i}(\zeta^{i})=a_{i}\circ r_{i}(\zeta^{i})$
which may need to be computed afresh for each value of $k$ in general.
We will see in Subsection~\ref{subsec:Deterministic-delayed-rejection}
an interesting scenario where this is not the case.
\begin{example}[Delayed-rejection of \citealp{tierney1999some}]
\label{eg:delayed-rejection} Assume $Q_{i}$ has density $q_{i}$
with respect to the Lebesgue or counting measure for each $i\in\llbracket n\rrbracket$
for some $n\in\mathbb{N_{*}}$ and $\phi_{i}(z_{1},z_{2},\ldots,z_{i-1},z_{i})=(z_{i},z_{i-1},\ldots,z_{2},z_{1})$
for $i\in\llbracket n-1\rrbracket$ ``reverse time'' and $\phi_{n}={\rm Id}$
(which ensures finite computations).
\end{example}
\begin{example}[Generalized delayed-rejection of \citealp{green2001delayed}]
 In this scenario involutions other than those of Example \ref{eg:delayed-rejection}
can be used. As an example, one can choose $Q_{1}(x,{\rm d}y)$ on
$\mathsf{X}\times\mathscr{Y}$, $Q_{2}(x,y;{\rm d}z,{\rm d}w)$ on
$\mathsf{X}\times\mathsf{Y}\times\mathscr{Z}\otimes\mathscr{W}$ and
$Q_{3}$ arbitrary. Assume $Q_{i}$ has density $q_{i}$ for each
$i\in\{1,2\}$, we may choose $\phi_{1}(x,y)=(y,x)$, $\phi_{2}(x,y,z,w)=(z,w,x,y)$
and $\phi_{3}(x,y,z,w,\ldots)=(x,y,z,w,\ldots)$.
\end{example}

\subsection{Deterministic delayed rejection\label{subsec:Deterministic-delayed-rejection}}

Delayed rejection can be usefully applied to sample from $\pi$ defined
on $(\mathsf{Z},\mathscr{Z})$ using purely deterministic proposals.
It is possible to use the framework above, but it is more convenient
notationally and conceptually to instead consider $E:=\{(k,z)\colon k\in\mathbb{N},z\in\mathsf{Z}\}$,
the embedding distribution, for $k\in\mathbb{N^{*}}$,
\[
\mu(k,{\rm d}z)=\alpha_{k}(z)\beta_{k}(z)\pi({\rm d}z),
\]
involutions $\phi_{k}\colon\mathsf{Z}\rightarrow\mathsf{Z}$ and as
before, for each $k\in\mathbb{N}$, let $\alpha_{k}=a_{k}\circ r_{k}$,
$\beta_{k}(z)=\beta_{k-1}(z)[1-\alpha_{k-1}(z)]$ with $\beta_{0}\equiv1$
and $\alpha_{0}\equiv0$, where 
\[
r_{k}(z)=\begin{cases}
\frac{{\rm d}\pi_{S_{k}}^{\phi_{k}}}{{\rm d}\pi_{S_{k}}}(z)\frac{\beta_{k}\circ\phi_{k}}{\beta_{k}}(z) & z\in S_{k}\\
0 & \text{otherwise}
\end{cases},
\]
with $S_{k}=S(\pi,\pi^{\phi_{k}})\cap\{z\in\mathsf{Z}\colon\beta_{k}(z)\wedge\beta_{k}\circ\phi_{k}(z)>0\}$.
An algorithmic presentation of delayed rejection with deterministic
proposals is given in Alg.~\ref{alg:Deterministic-delayed-rejection}.
Its justification follows along the same lines as above, and as before
one can choose $\phi_{n}={\rm Id}$ for some $n\in\mathbb{N}_{*}$
to ensure that the stage $n$ ``proposal'' is accepted.

As in the stochastic scenario, the computation of $\beta_{k}\circ\phi_{k}(z)$
can be expensive since this requires in particular the computation
of $r_{1}\circ\phi_{k}(z),\ldots,r_{k-1}\circ\phi_{k}(z)$. Assume
for simplicity that $\pi$ has a density $\varpi$ with respect to
some measure $\nu$ invariant under $\phi_{i}$ for $i\in\llbracket k-1\rrbracket$,
then we remark that on $S_{k}$
\[
r_{i}\circ\phi_{k}(z)=\frac{\varpi\circ\phi_{i}\circ\phi_{k}}{\varpi\circ\phi_{k}}(z)\frac{\beta_{i}\circ\phi_{i}\circ\phi_{k}}{\beta_{i}\circ\phi_{k}}(z).
\]
Then, if for $i\leq k$ the identity $\varpi\circ\phi_{i}\circ\phi_{k}=\varpi\circ\phi_{k-i}$
holds, we see that no new evaluation of the probability density is
required, which is to be contrasted with the general setup in Subsection
\ref{subsec:Stochastic-delayed-rejection}. This identity holds when
for $\psi\colon\mathsf{Z}\rightarrow\mathsf{Z}$, assumed invertible,
is such that for an involution $\sigma\colon\mathsf{Z}\rightarrow\mathsf{Z}$,
$\sigma\circ\psi\circ\sigma=\psi^{-1}$ one considers the involutions
$\phi_{i}=\sigma\circ\psi^{i}$ and has the property $\varpi\circ\sigma=\varpi$,
since
\[
\phi_{i}\circ\phi_{k}=\sigma\circ\psi^{i}\circ\sigma\circ\psi^{k}=\psi^{-i}\circ\psi^{k}=\psi^{k-i}=\sigma\circ\phi_{k-i}.
\]
 This is the setup considered in \citet{sohl2014hamiltonian,campos2015extra}
where an additional twist, detailed in the next subsection, is used.

\begin{algorithm}
\caption{Deterministic delayed rejection \label{alg:Deterministic-delayed-rejection}}

\begin{enumerate}
\item Given $z\in\mathsf{Z}$, set $k\leftarrow0$.
\item Set $k\leftarrow k+1$.
\item With probability $\alpha_{k}(z)$ output $\phi_{k}(z)$ otherwise
go to 2.
\end{enumerate}
\end{algorithm}

\begin{example}[DR deterministic]
 Consider $\pi$ defined on $\big(\mathsf{X}\times\mathsf{V},\mathscr{X}\otimes\mathscr{V}\big)$
of density $\varpi(x,v)=\gamma(x)\kappa(v)$ and let $\phi_{i}:=\sigma\circ\psi^{i}$
for $i\in\llbracket n-1\rrbracket$ and $\phi_{n}={\rm Id}$ with
$\psi(x,v):=(x+v,v)$ and $\sigma(x,v)=(x,-v)$. This can be useful
when trying to traverse a region of low probability. As an example,
let $\mathsf{X}\times\mathsf{V}=\llbracket3\rrbracket\times\{-1,1\}$
and $\varpi(1,1)=\varpi(1,-1)>0$, $\varpi(3,1)=\varpi(3,-1)>0$ but
$\varpi(2,1)=\varpi(2,-1)=0$. In this scenario it is a good idea
to choose $n=3$ rather than the standard $n=2$ choice.
\end{example}

\subsection{Sliced delayed rejection}

The introduction of an auxiliary slice variable can mitigate the computational
cost of the delayed rejection approach, and some recently proposed
algorithms \citet{sohl2014hamiltonian,campos2015extra} can be viewed
as following this principle. In particular, we can define
\[
\varpi(z,u)=\varpi(z)\mathbb{I}\{u\leq\varpi(z)\}/\varpi(z)=\mathbb{I}\{u\leq\varpi(z)\},
\]
and use a slice sampler \citep{neal2003slice} (see Definition~\ref{def:slice-sampler}),
that is a MwG alternating between updating $u$ given $z$ and vice-versa,
that is sampling uniformly from the ``slice'' $\{z\in\mathsf{Z}:\varpi(z)\geq u\}$,
for a fixed $u\in\mathbb{R}_{+}$. One may use any Markov kernel that
leaves this distribution invariant and we naturally focus on MH type
updates.
\begin{example}[Extra chance slice]
\label{eg:xtra-chance} For some fixed $u$, let $\varpi_{u}(z)\propto\mathbb{I}\{u\leq\varpi(z)\}$.
Let $\phi_{i}(z)=\sigma\circ\psi_{i}(z)$ for $i\in\llbracket n-1\rrbracket$
and $\phi_{n}={\rm Id}$. If $a_{i}(r)=a(r)=1\wedge r$, we find that
(see Appendix \ref{sec:X-tra-chance-proof} for a proof) for $z\in\mathsf{Z}$,
for $k\in\llbracket n\rrbracket$ and the convention $\vee_{i=1}^{0}=0$,
\[
r_{k}(z)=\mathbb{I}\{\vee_{i=1}^{k-1}\varpi\circ\phi_{i}(z)<u\leq\varpi(z)\wedge\varpi\circ\phi_{k}(z)\}
\]
while $\beta_{k}(z)=\mathbb{I}\{\varpi(z)\wedge\vee_{i=1}^{k-1}\varpi\circ\phi_{i}(z)<u\}$.
Hence, one accepts as soon as $\varpi\circ\phi_{k}(z)\geq u$ or one
reaches the identity involution $\phi_{n}={\rm Id}$. One notices
that for $\varpi(z)>0$ and $u\sim{\rm Uniform}(0,\varpi(z))$ then
$u_{0}:=u/\varpi(z)\sim{\rm Uniform}(0,1)$ and one can rewrite 
\[
r_{k}(z)=\mathbb{I}\Bigl\{\vee_{i=1}^{k-1}\varpi\circ\phi_{i}(z)/\varpi(z)<u_{0}\leq1\wedge\big(\varpi\circ\phi_{k}(z)/\varpi(z)\big)\Bigr\}.
\]
The overall slice sampler therefore looks like a standard MH algorithm
targetting $\varpi$, where given $u_{0}\sim{\rm Uniform}(0,1)$ one
scans the states $\phi_{i}(z)$ for $i\in\llbracket n\rrbracket$
until the right hand side inequality is satisfied or $n$ is reached.
When $n=2$ we recover the standard MH algorithm targetting $\varpi$
and with deterministic proposal, corresponding to a remark going as
far back as \citet{doi:10.1080/01621459.1998.10473712}.

In the context of HMC samplers this can be a way of taking into account
the oscillatory nature of the energy $i\mapsto H\circ\psi^{i}(x,v)$
under the leapfrog dynamics. More specifically we may have $\varpi\circ\phi_{k}(z)\geq u$
even though $\varpi\circ\phi_{i}(z)<u$ for $i\in\llbracket k-1\rrbracket$.
Note that $\psi$ may involve several steps of the numerical integrator
(which preserves Lebesgue measure and is time-reversible).
\end{example}
\begin{example}
The ``sequential-proposal Metropolis(--Hastings) algorithm'' of
\citet{park2020markov} shares the precise structure of \citet{campos2015extra}
albeit in the scenario where the states are proposed randomly, but
this connection was not made by the authors.
\end{example}

\subsection{Discrete time bouncy particle samplers}

Let $\mathrm{b}$ be a volume preserving ``bounce'' involution,
e.g. with $\mathrm{b}_{v}(x,v):=v-2\left\langle v,n(x)\right\rangle n(x)$
for some function $n:\mathbb{R}^{d}\to\mathbb{R}^{d}$ such that for
all $x\in\mathsf{X},$ $\|n(x)\|=1$ we let $\mathrm{b}(x,v):=\big(x,\mathrm{b}_{v}(x,v)\big)$
for $(x,v)\in\mathsf{X}\times\mathsf{V}$. To fix ideas, for the two
following examples the scenario where $\psi(x,v)=(x+v,v)$ corresponds
to the algorithms of \citet{sherlock2017discrete} and \citet{vanetti2017piecewise}.
Similar ideas are briefly alluded to in \citet{neal2003slice}. 
\begin{example}[Bouncy I - \citet{sherlock2017discrete}]
 Let $\phi_{1}=\sigma\circ\psi$ and $\phi_{2}=\phi_{1}\circ\mathrm{b}\circ\phi_{1}$.
Note that $\phi_{2}$ is an involution since $\phi_{1}$ and $\mathrm{b}$
are involutions. We have the convenient property that $\phi_{1}\circ\phi_{2}=\mathrm{b}\circ\phi_{1}$
(since $\phi_{1}$ is an involution), and $\mathrm{b}\circ\phi_{1}(z)$
will already have been computed to produce $\phi_{2}(z)$. We have
$\mathrm{b}\circ\phi_{1}(x,v)=\big(x+v,\mathrm{b}_{v}(x+v,-v)\big)$
and $\phi_{2}(x,v)=\big(x+v+\mathrm{b}_{v}(x+v,-v),-\mathrm{b}_{v}(x+v,-v)\big)$
and therefore for $\xi\in S$ the acceptance ratio is 
\[
r_{2}(\xi)=\frac{\bar{\alpha}_{1}\circ\phi_{2}(x,v)(\gamma\otimes\kappa)\circ\phi_{2}(x,v)}{\bar{\alpha}_{1}(x,v)\gamma\otimes\kappa(x,v)}=\frac{\gamma\big(x+v+\mathrm{b}_{v}(x+v,-v)\big)}{\gamma(x)}
\]
since, 
\[
r_{1}\circ\phi_{2}(x,v)=\frac{\gamma(x+v)\kappa\circ\mathrm{b}_{v}(x+v,-v)}{\gamma(x)\kappa(v)}=r_{1}(x,v).
\]
\end{example}
\begin{example}[Bouncy II - \citet{vanetti2017piecewise}]
 Let $\phi_{1}=\sigma\circ\psi$ and $\phi_{2}=\mathbf{\mathrm{b}}$,
where $\mathrm{b}$ is an involution. Here more computations are required
since $\phi_{1}\circ\phi_{2}=\phi_{1}\circ\mathrm{b}$ and $\phi_{1}\circ\mathrm{b}(z)$
is not typically computed as a by-product of computing $\phi_{1}(z)$
or $\phi_{2}(z)$. Here for $\xi=(x,v)\in S$
\[
r_{2}(\xi)=\frac{\bar{\alpha}_{1}\circ\phi_{2}(x,v)(\gamma\otimes\kappa)\circ\phi_{2}(x,v)}{\bar{\alpha}_{1}(x,v)\gamma\otimes\kappa(x,v)}=\frac{\bar{\alpha}_{1}\circ\phi_{2}(x,v)}{\bar{\alpha}_{1}(x,v)}
\]
where 
\[
r_{1}\circ\phi_{2}(x,v)=\frac{\gamma\big(x+{\rm b}_{v}(x,v)\big)\kappa\big(-\mathrm{b}_{v}(x+v,-v)\big)}{\gamma(x)\kappa(v)}=\frac{\gamma\big(x+{\rm b}_{v}(x,v)\big)}{\gamma(x)}.
\]
\end{example}

\subsection{Discrete-time exact event chain algorithms}

In a lineage of contributions \citet{jaster1999improved,bernard2009event,michel2014generalized,michel2015event,michel:2016}
efficient continuous time nonreversible Markov process Monte Carlo
(MPMC) algorithms have been developed to sample from models arising
in statistical physics. We show here that it is possible to develop
discrete time and exact counterparts of those, that is algorithms
of finite run time without any approximation but the machine numerical
precision limit and are ensured to leave the desired distribution
invariant. More specifically let $x=(x_{1},\ldots,x_{m})$ for some
$m\in\mathbb{N}$ where $x_{1},\ldots,x_{m}\in\mathsf{X}$. For example
$\mathsf{X}$ might be a bounded subset of $\mathbb{R}^{d}$ for some
$d$ or as is a common in the physics literature a torus. This can
be thought of as the positions of $m$ particles, modelled as spheres.
We also define a velocity variable $v\in\mathsf{\mathsf{V}\subset\mathbb{R}}^{d}$.
The target distribution has density with respect to some measure $\nu$,
which can be the product of the Lebesgue or Hausdorff or counting
measure depending on the scenario considered,
\begin{equation}
\varpi(x,v)=\kappa(v)\gamma(x)\propto\kappa(v)\prod_{1\leq i<j\leq m}\mathbb{I}\left\{ \Vert x_{i}-x_{j}\Vert>\delta_{ij}\right\} ,\label{eq:hard-sphere-model}
\end{equation}
where $\delta_{ij}\in\mathbb{R}_{+}$ for all $i,j\in\left\llbracket m\right\rrbracket $,
$i<j$ and it is assumed that $\kappa(v)=\kappa(-v)$ for all $v\in\mathsf{V}$.
No simplify notation we introduce the feasible set $F\subset\mathsf{X}$
such that $\mathbb{I}\{x\in F\}=\prod_{1\leq i<j\leq m}\mathbb{I}\left\{ \Vert x_{i}-x_{j}\Vert>\delta_{ij}\right\} $.
In the absence of mean field we see the necessity to constrain $\mathsf{X}$
to be ``bounded'' for this to define a probability distribution.
This clearly accommodates hard constraints on the distance between
the particles. We now briefly describe the aforementioned MPMC in
the hard sphere scenario given by (\ref{eq:hard-sphere-model}). This
MPMC is a so-called piecewise deterministic Markov process where a
sphere, labelled $i\in\llbracket m\rrbracket$, evolves continuously
along a straight line of direction the velocity $v$ until a collision
with another sphere occurs, say $j\in\llbracket m\rrbracket\setminus\{i\}$
or until an exponential clock of fixed parameter rings. When a collision
occurs the velocity is transferred to sphere $j$, while when the
clock rings a new velocity is drawn afresh from $\kappa$. For soft
potentials implementation of the algorithm will typically require
time discretisation in order to determine the time to a ``soft''
collision. For completeness we provide a description of the generator
of the MPMC above for soft potentials in Appendix \ref{sec:Event-chain-algorithms}.
In Alg. \ref{alg:BKM} we introduce a novel exact discretization of
aforementioned MPMC which circumvents the need for a time discretization
approximation thanks to a MH kernel involving delayed rejection. Note
that in practice this kernel is composed with $\mathfrak{S}$ such
that for any $f\in\mathsf{X}^{m}\times\mathsf{V}\times\llbracket m\rrbracket^{2}$,
$\mathfrak{S}f(x,v,i,j)=f\circ\sigma(x,v,i,j)$ where for any $(x,v,i,j)\in\mathsf{X}^{m}\times\mathsf{V}\times\llbracket m\rrbracket^{2}$,
$\sigma(x,v,i,j)=(x,-v,i,j)$ denotes the function that flips the
sign of $v$. 

\begin{algorithm}
\caption{\label{alg:BKM} $\mathfrak{S}-$ symmetrisation of the discrete time
event chain kernel}

Input: $(x,v,i)$
\begin{enumerate}
\item Set $I\leftarrow\emptyset$. For $k\in\left\llbracket m\right\rrbracket \setminus\{i\}$,
if $\Vert x_{i}+v-x_{k}\Vert\leq\delta_{ik}$ set $I\leftarrow I\cup\{k\}$.
\item If $I=\emptyset$, set $j=i$, $x_{i}\leftarrow x_{i}+v$ and output
$(x,-v,i,j)$.
\item Otherwise, sample $j\sim{\rm Uniform}(I)$.
\item Set $I'\leftarrow\emptyset$. For $k\in\left\llbracket m\right\rrbracket \setminus\{j\}$,
if $\Vert x_{k}+v-x_{j}\Vert\leq\delta_{jk}$ set $I'\leftarrow I'\cup\{k\}$.
\item With probability $a(|I|/|I'|)$ output $(x,-v,j)$, otherwise output
$(x,v,i)$.
\end{enumerate}
\end{algorithm}

To justify the algorithm we let $\xi:=(x,v,i,j)\in\mathsf{X}^{m}\times\mathsf{V}\times\llbracket m\rrbracket^{2}$:
here $i$ is the index of the particle that is ``moving'' and $j$
the index of a candidate particle that will be ``given'' the velocity
of the $i$th particle, in a way that will become clear. The algorithm
is a two-stage delayed rejection MH kernel. The first involution is
$\phi_{1}=\sigma\circ\psi_{1}$ where $\psi_{1}(x,v,i,j)=(x',v,i,j)$
with $x'_{i}=x_{i}+v$ and $x'_{j}=x_{j}$ for $j\in\left\llbracket m\right\rrbracket \setminus\{i\}$,
which we may denote $x'=x+\mathbf{e}_{i}\varoast v$ with $\varoast$
the Kronecker product and $\{\mathbf{e}_{i}\in\mathsf{X}^{m},i\in\llbracket m\rrbracket\}$
such that $(\mathbf{e}_{i})_{j}=\mathbb{I}\{i=j\}$. An interpretation
of $\psi_{1}$ is that the $i-$th particle is translated by $v$
and all other particles remain fixed. Let $I(x,v,i):=\{j\in\left\llbracket m\right\rrbracket \setminus\{i\}:\Vert x_{i}+v-x_{j}\Vert\leq\delta_{ij}\}$,
i.e. $I(x,v,i)$ is the set of particle indices $j$ such that $x_{i}+v$
``collides'' with $x_{j}$ and let $|I|(x,v,i):=|I(x,v,i)|$, so
that $|I|(x,v,i)$ is the number of such collisions. The second involution
is simply $\phi_{2}(x,v,i,j)=(x,-v,j,i)$, that is particle $j$ becomes
active and has velocity $-v$. We define $\mu({\rm d}\xi)\propto\gamma\otimes\kappa({\rm d}x,{\rm d}v)\mathbb{I}\{i\in\left\llbracket m\right\rrbracket \}q(j;x,v,i)$,
where we are free to choose the following proposal distribution for
the next active particle, among those in $I(x,v,i)$,
\[
q(j;x,v,i)=\begin{cases}
\mathbb{I}\{i=j\} & I(x,v,i)=\emptyset,\\
\frac{\mathbb{I}\{j\in I(x,v,i)\}}{|I|(x,v,i)} & I(x,v,i)\neq\emptyset.
\end{cases}
\]
We have (with $0/0=0$ here)
\[
\rho(x,v,i,j)=\frac{\kappa(v)}{m}\mathbb{I}\{x\in F,i\in\llbracket m\rrbracket\}\begin{cases}
\mathbb{I}\{I(x,v,i)=\emptyset,i=j\} & I(x,v,i)=\emptyset,\\
\frac{\mathbb{I}\{j\in I(x,v,i)\neq\emptyset\}}{|I|(x,v,i)} & I(x,v,i)\neq\emptyset.
\end{cases}
\]
Notice that $\mathbb{I}\{x+\mathbf{e}_{i}\varoast v\in F\}=\mathbb{I}\{I(x,v,i)=\emptyset\}$
and equivalently $\mathbb{I}\{x\in F\}=\mathbb{I}\{I\circ\tilde{\phi}_{1}(x,v,i)=\emptyset\}$
with $\tilde{\phi}_{1}(x,v,i)=(x+\mathbf{e}_{i}\varoast v,-v,i)$
and $I\circ\tilde{\phi}_{1}(x,v,i):=I\big(\tilde{\phi}_{1}(x,v,i)\big)$
\begin{align*}
\rho\circ\phi_{1}(x,v,i,j) & =\frac{\kappa(-v)}{m}\mathbb{I}\{x+\mathbf{e}_{i}\varoast v\in F,i\in\llbracket m\rrbracket\}\begin{cases}
\mathbb{I}\{x\in F,i=j\} & I(x,v,i)=\emptyset,\\
\frac{\mathbb{I}\{x\notin F,j\in I\circ\tilde{\phi}_{1}(x,v,i)\}}{|I|\circ\tilde{\phi}_{1}(x,v,i)} & I(x,v,i)\neq\emptyset.
\end{cases}
\end{align*}
where. Therefore, using that $\rho(x,v,i,j),\rho\circ\phi_{1}(x,v,i,j)\in\{0,\kappa(v)/m\}$
we obtain
\begin{align*}
r_{1}(x,v,i,j) & =\mathbb{I}\{x\in F,(x+\mathbf{e}_{i}\varoast v)\in F,i=j\}\\
 & =\mathbb{I}\{I\circ\tilde{\phi}_{1}(x,v,i)=I(x,v,i)=\emptyset,i=j\}.
\end{align*}
For the second stage observe that for $i,j\in\llbracket m\rrbracket$
we have $\Vert x_{i}+v-x_{j}\Vert=\Vert x_{i}-(x_{j}-v)\Vert$ and
therefore $\mathbb{I}\{j\in I(x,v,i)\neq\emptyset\}=\mathbb{I}\{i\in I(x,-v,j)\neq\emptyset\}$
and therefore
\[
\rho\circ\phi_{1}(x,v,i,j)=0\iff\rho\circ\phi_{1}\circ\phi_{2}(x,v,i,j)=0
\]
in which case
\[
\rho\circ\phi_{2}(x,v,i,j)=\mathbb{I}\{x\in F\}\frac{\mathbb{I}\{i\in I(x,-v,j)\neq\emptyset\}}{|I|(x,-v,j)}\frac{\kappa(v)}{m}.
\]
Therefore we conclude that
\[
r_{2}(x,v,i,j)=\mathbb{I}\{x\in F,i\in I(x,-v,j)\neq\emptyset,j\in I(x,v,i)\neq\emptyset\}\frac{|I|(x,v,i)}{|I|(x,-v,j)},
\]
and easy counterexamples show that there is no reason for the equality
$|I|(x,v,i)=|I|(x,-v,j)\neq0$ to hold in general, when the indicator
function is one. In practice, one can implement the combination
of this kernel with the refreshments for $(x,v,i)$ in various ways
to save time. In particular, one may be able to determine the first
time at which either a refreshment occurs or there is a collision.
We do not consider these details here.

We now show that Alg. \ref{alg:BKM} can be straightforwardly adapted
to accommodate ``soft potentials'' (or constraints) by using a slice
sampler strategy and hence the introduction of instrumental variables.
For example, assume that
\[
\gamma(x)=\prod_{1\leq i<j\leq m}\gamma_{ij}(x),
\]
where $\gamma_{ij}(x)=\Gamma(\Vert x_{i}-x_{j}\Vert)$ with $\Gamma\colon\mathbb{R}_{+}\rightarrow\mathbb{R}_{+}$
is non-decreasing and such that $\gamma$ is a probability density
on $\mathsf{X}$ for the relevant dominating measure. Then one can
consider the instrumental distribution, with $u=(u_{ij})\in\mathbb{R}_{+}^{m(m-1)/2}$
for $i\in\llbracket m-1\rrbracket$ and $j\in\llbracket i+1,m\rrbracket$,
\begin{align*}
\varpi(x,v,u) & :=\gamma(x)\prod_{1\leq i<j\leq m}\frac{\mathbb{I}\{u_{ij}\leq\gamma_{ij}(x)\}}{\gamma_{ij}(x)}\\
 & =\prod_{1\leq i<j\leq m}\mathbb{I}\{u_{ij}\leq\gamma_{ij}(x)\}\\
 & =\prod_{1\leq i<j\leq m}\mathbb{I}\left\{ \Vert x_{i}-x_{j}\Vert\geq\Gamma^{-1}(u_{ij})\right\} ,
\end{align*}
where $\Gamma^{-1}(u):=\inf\{y:\Gamma(y)\geq u\}$ . Hence, for a
fixed $u$, we have $\pi_{u}(x,v)$ of the same form as (\ref{eq:hard-sphere-model})
with $\delta_{ij}=\Gamma^{-1}(u_{ij})$, suggesting the use of a MwG
strategy to sample from $\pi$. It is naturally possible to consider
more general forms for the $\gamma_{ij}$ and adaptation of the algorithm
is straightforward.

\section{Acknowledgements}

CA and SL acknowledge support from EPSRC ``Intractable Likelihood:
New Challenges from Modern Applications (ILike)'' (EP/K014463/1).
CA and AL acknowledge support of EPSRC grant CoSInES (EP/R034710/1)
and CA acknowledges support of EPSRC grant Bayes4Health (EP/R018561/1).

%% file: appendices.tex
\section{Proofs\label{sec:appendix-some-proofs}}
\begin{proof}[Proof of Theorem~\ref{thm:invo-rev}]
 Let $\lambda=\mu+\mu^{\phi}$ and define $\rho={\rm d}\mu/{\rm d}\lambda$.
We observe that $\lambda^{\phi}=\lambda$. Then for any $A\in\mathscr{E}$,
we find using Theorems~\ref{thm:change-of-variables}--\ref{thm:Radon-Nikodym},
\begin{align*}
\int{\bf 1}_{A}(\xi)\rho\circ\phi(\xi)\lambda({\rm d}\xi) & =\int{\bf 1}_{A}\circ\phi(\xi)\rho(\xi)\lambda^{\phi}({\rm d}\xi)\\
 & =\int{\bf 1}_{A}\circ\phi(\xi)\mu({\rm d}\xi)\\
 & =\int{\bf 1}_{A}(\xi)\mu^{\phi}({\rm d}\xi),
\end{align*}
so $\rho\circ\phi={\rm d}\mu^{\phi}/{\rm d}\lambda$. Define $S=\{\xi\in E:\rho(\xi)\wedge\rho\circ\phi(\xi)>0\}$,
which satisfies $\phi(S)=S$. Since $S$ is the intersection of two
measurable sets, it is measurable and so the restrictions of $\mu$
and $\mu^{\phi}$ to $S$ are well defined, and for $\xi\in S$,
\[
\frac{{\rm d}\mu_{S}^{\phi}}{{\rm d}\mu_{S}}(\xi)=\frac{\rho\circ\phi}{\rho}(\xi),\qquad\frac{{\rm d}\mu_{S}}{{\rm d}\mu_{S}^{\phi}}(\xi)=\frac{\rho}{\rho\circ\phi}(\xi),
\]
so $\mu_{S}\equiv\mu_{S}^{\phi}$. Let $A=\{\xi:\rho(\xi)=0\}$ and
$B=\{\xi:\rho(\xi)>0,\rho\circ\phi(\xi)=0\}$. We deduce that $\mu(A)=\int_{A}\rho(\xi)\lambda({\rm d}\xi)=0$
and $\mu^{\phi}(B)=\int_{B}\rho\circ\phi(\xi)\lambda({\rm d}\xi)=0$.
Since $A\cap B=\emptyset$, $A\cup B=S^{\complement}$ and $\mu(A)=\mu^{\phi}(B)=0$
we conclude that $\mu$ and $\mu^{\phi}$ are mutually singular on
$S^{\complement}$.

For part b(i), $\left(\rho\circ\phi/\rho\right)\circ\phi=\rho/\rho\circ\phi$
on $S$, since $\phi$ is an involution. Hence, $r\circ\phi=1/r$
on $S$ and since $a(0)=0$, $\alpha=r\cdot\alpha\circ\phi$ on $S$
and $\alpha=0$ on $S^{\complement}$ by the definition of $\alpha$
and condition on $a$.

For part b(ii), combining the first part with $a(0)=0$, Theorems~\ref{thm:change-of-variables}--\ref{thm:Radon-Nikodym}
and $\phi$ an involution, we obtain
\begin{align*}
\int_{E}F(\xi)G\circ\phi(\xi)\alpha(\xi)\mu({\rm d}\xi) & =\int_{S}F(\xi)G\circ\phi(\xi)\alpha(\xi)\mu_{S}({\rm d}\xi)\\
 & =\int_{S}F(\xi)G\circ\phi(\xi)r(\xi)\alpha\circ\phi(\xi)\mu_{S}({\rm d}\xi)\\
 & =\int_{S}F(\xi)G\circ\phi(\xi)\alpha\circ\phi(\xi)\mu_{S}^{\phi}({\rm d}\xi)\\
 & =\int_{S}F\circ\phi(\xi)G(\xi)\alpha(\xi)\mu_{S}({\rm d}\xi)\\
 & =\int_{E}F\circ\phi(\xi)G(\xi)\alpha(\xi)\mu({\rm d}\xi).
\end{align*}

For the part b(iii), we define the sub-Markov kernels
\[
T(\xi,A)=\alpha(\xi){\bf 1}_{A}(\phi(\xi)),\qquad\xi\in E,A\in\mathscr{E},
\]
and
\[
R(\xi,A)=\left\{ 1-\alpha(\xi)\right\} {\bf 1}_{A}(\xi),\qquad\xi\in E,A\in\mathscr{E},
\]
so that $\Pi=T+R$. First we observe that
\[
\int F(\xi)G(\xi')\mu({\rm d}\xi)R(\xi,{\rm d}\xi')=\int F(\xi)G(\xi)\left\{ 1-\alpha(\xi)\right\} \mu({\rm d}\xi)=\int G(\xi)F(\xi')\mu({\rm d}\xi)R(\xi,{\rm d}\xi').
\]

Then from the second part,
\begin{align*}
\int F(\xi)G(\xi')\mu({\rm d}\xi)T(\xi,{\rm d}\xi') & =\int F(\xi)G\circ\phi(\xi)\alpha(\xi)\mu({\rm d}\xi)\\
 & =\int F\circ\phi(\xi)G(\xi)\alpha(\xi)\mu({\rm d}\xi)\\
 & =\int G(\xi)F(\xi')\mu({\rm d}\xi)T(\xi,{\rm d}\xi').
\end{align*}
\end{proof}

\begin{proof}[Proof of Proposition~\ref{cor:rev-coordinate}]
Let $f,g:\mathsf{Z}\to[0,1]$ be measurable and let $F,G\colon\Xi\rightarrow[0,1]$
such that $F(\xi)=f(\xi_{0})$ and $G(\xi)=g(\xi_{0})$. Using reversibility
of $\Pi$, we find
\begin{align*}
\int f(\xi_{0})g(\xi_{0}')\pi({\rm d}\xi_{0})P(\xi_{0},{\rm d}\xi_{0}') & =\int f(\xi_{0})g(\xi_{0}')\pi({\rm d}\xi_{0})\mu_{\xi_{0}}({\rm d}\xi_{-0})\Pi(\xi;{\rm d}\xi')\\
 & =\int F(\xi)G(\xi')\mu({\rm d}\xi)\Pi(\xi;{\rm d}\xi')\\
 & =\int G(\xi)F(\xi')\mu({\rm d}\xi)\Pi(\xi;{\rm d}\xi')\\
 & =\int g(\xi_{0})f(\xi_{0}')\pi({\rm d}\xi_{0})P(\xi_{0},{\rm d}\xi_{0}').
\end{align*}
\end{proof}
\begin{proof}[Proof of Lemma~\ref{lemnu:leb-count-combi}]
Since $g$ is an involution and $g(\mathsf{Y})\subseteq\mathsf{Y}$,
we have $\lambda_{Y}=\lambda_{Y}^{g}$ as explained in Remark~\ref{rem:d-lambda-phi-common}.
Let $\psi$ be a non-negative function and $\lambda({\rm d}x,{\rm d}y)=\lambda_{X}({\rm d}x)\lambda_{Y}({\rm d}y)$.
Then we find
\begin{align*}
\int\psi(x,y)\lambda({\rm d}x,{\rm d}y) & =\int\psi(x,y)\lambda_{Y}({\rm d}y)\lambda_{X}({\rm d}x)\\
 & =\int\psi(x,y)\lambda_{Y}^{g}({\rm d}y)\lambda_{X}({\rm d}x)\\
 & =\int\psi(x,g(y))\lambda_{Y}({\rm d}y)\lambda_{X}({\rm d}x)\\
 & =\int\psi(f(x,y),g(y))\left|{\rm det}f_{y}'(x)\right|\lambda_{Y}({\rm d}y)\lambda_{X}({\rm d}x)\\
 & =\int\psi\circ\phi(x,y)\left|{\rm det}f_{y}'(x)\right|\lambda({\rm d}x,{\rm d}y)
\end{align*}
where $f_{y}(x)=x\mapsto f(x,y)$ for each $y\in\mathsf{Y}$. Since
\[
\int\psi\circ\phi(x,y)\lambda^{\phi}({\rm d}x,{\rm d}y)=\int\psi(x,y)\lambda({\rm d}x,{\rm d}y),
\]
and for an arbitrary, measurable, non-negative $g:E\to\mathbb{R}$
we can take $\psi=g\circ\phi^{-1}$ to obtain $g=\psi\circ\phi$,
we obtain that ${\rm d}\lambda^{\phi}/{\rm d}\lambda=\left|{\rm det}f_{y}'(x)\right|$.
\end{proof}
\begin{proof}[Proof of Proposition~\ref{prop:psi-is-sigma-phi}]
 For part \ref{enu:norev-prop-psi}, the identity $\psi^{-1}=\sigma\circ\psi\circ\sigma$
is verified by observing that, since $\psi=\sigma\circ\phi$, $\sigma\circ\psi\circ\sigma=\phi\circ\sigma$
and that indeed $\phi\circ\sigma\circ\psi=\psi\circ\phi\circ\sigma={\rm Id}$.
We then note that $\phi=\sigma\circ\psi$ and so for $A\in\mathscr{E}$,
\[
\lambda^{\phi}(A)=\lambda(\phi^{-1}(A))=\lambda(\phi(A)=\lambda(\sigma\circ\psi(A))=\lambda^{\sigma}(\psi(A))=\lambda(\psi(A))=\lambda^{\psi^{-1}}(A).
\]
Since $\lambda^{\sigma}=\lambda$, for $f$ integrable w.r.t. $\mu^{\sigma}$,
\begin{align*}
\int f(\xi)\mu^{\sigma}({\rm d}\xi) & =\int f\circ\sigma(\xi)\mu({\rm d}\xi)\\
 & =\int f\circ\sigma(\xi)\rho(\xi)\lambda({\rm d}\xi)\\
 & =\int f\circ\sigma(\xi)\rho(\xi)\lambda^{\sigma}({\rm d}\xi)\\
 & =\int f(\xi)\rho\circ\sigma(\xi)\lambda({\rm d}\xi),
\end{align*}
and so ${\rm d}\mu^{\sigma}/{\rm d}\lambda=\rho\circ\sigma$. Since
$\mu^{\sigma}=\mu$, we have $\rho\circ\sigma=\rho$. Hence, $\rho\circ\phi=\rho\circ\sigma\circ\psi=\rho\circ\psi$.
We proceed to part \ref{enu:non-rev-def-kernel}, and note that
\begin{align*}
\Pi(\xi,{\rm d}\xi') & =\alpha(\xi)\delta_{\phi(\xi)}({\rm d}\xi')+[1-\alpha(\xi)]\delta_{\xi}({\rm d}\xi'),
\end{align*}
from which for $\xi\in E$ and $f\colon E\rightarrow[0,1]$, with
$\Phi f:=f\circ\phi$,
\[
\Pi\mathfrak{S}f=\alpha(\xi)\cdot\Phi\mathfrak{S}f+[1-\alpha(\xi)]\mathfrak{S}^{2}f
\]
and use that $\mathfrak{S}^{2}={\rm Id}$ and $\Phi\mathfrak{S}f(\xi)=\phi\left(\mathfrak{S}f\right)(\xi)=\phi\left(f\circ\sigma\right)(\xi)=f\circ\sigma\circ\phi(\xi)$.
Using the identities from part \ref{enu:norev-prop-psi}, we can use
the general acceptance ratio for $\Pi$ from Proposition~\ref{prop:r-density}
to obtain,
\[
r(\xi)=\frac{\rho\circ\phi}{\rho}(\xi)\frac{{\rm d}\lambda^{\phi}}{{\rm d}\lambda}(\xi)=\frac{\rho\circ\psi}{\rho}(\xi)\frac{{\rm d}\lambda^{\psi^{-1}}}{{\rm d}\lambda}(\xi),\qquad\xi\in S.
\]
For part \ref{enu:nonrev-def-muS-reversible}, for $f,g\in E\rightarrow[0,1]$
we use that $\varPi\mathfrak{S}:=\Pi$ satisfies detailed balance
and $\mu\mathfrak{S}=\mu$ 
\begin{align*}
\int f(\xi)g(\xi')\mu({\rm d}\xi)\varPi(\xi,{\rm d}\xi') & =\int f(\xi)\mathfrak{S}g(\xi')\mu({\rm d}\xi)\varPi\mathfrak{S}(\xi,{\rm d}\xi')\\
 & =\int f(\xi)\mathfrak{S}g(\xi')\mu({\rm d}\xi')\varPi\mathfrak{S}(\xi',{\rm d}\xi)\\
 & =\int f(\xi)\mathfrak{S}g(\xi')\mu\mathfrak{S}({\rm d}\xi')\varPi\mathfrak{S}(\xi',{\rm d}\xi)\\
 & =\int f(\xi)g(\xi')\mu({\rm d}\xi')\mathfrak{S}\varPi\mathfrak{S}(\xi',{\rm d}\xi).
\end{align*}
For part \ref{enu:nonrev-SvarPi} we proceed as above and for $\xi\in S$
notice that
\begin{align*}
\mathfrak{S}\left[\frac{\rho\circ\phi}{\rho}\frac{{\rm d}\lambda^{\phi}}{{\rm d}\lambda}\right](\xi) & =\frac{\rho\circ\phi\circ\sigma}{\rho\circ\sigma}(\xi)\frac{{\rm d}\lambda^{\phi}}{{\rm d}\lambda}\circ\sigma(\xi)\\
 & =\frac{\rho\circ\phi\circ\sigma}{\rho}(\xi)\frac{{\rm d}\lambda^{\phi\circ\sigma}}{{\rm d}\lambda}(\xi),
\end{align*}
where we have used that
\begin{align*}
\int f(\xi)\frac{{\rm d}\lambda^{\phi}}{{\rm d}\lambda}\circ\sigma(\xi)\lambda({\rm d}\xi) & =\int f(\xi)\frac{{\rm d}\lambda^{\phi}}{{\rm d}\lambda}\circ\sigma(\xi)\lambda^{\sigma}({\rm d}\xi)\\
 & =\int f(\xi)\frac{{\rm d}\lambda^{\phi}}{{\rm d}\lambda}(\xi)\lambda({\rm d}\xi)\\
 & =\int f\circ\phi(\xi)\lambda({\rm d}\xi)\\
 & =\int f\circ\phi(\xi)\lambda^{\sigma}({\rm d}\xi)\\
 & =\int f\circ\phi\circ\sigma(\xi)\lambda({\rm d}\xi)\\
 & =\int f(\xi)\frac{{\rm d}\lambda^{\phi\circ\sigma}}{{\rm d}\lambda}\lambda({\rm d}\xi),
\end{align*}
where we have used that $\lambda^{\phi\circ\sigma}(A)=\lambda\big(\sigma\circ\phi(A)\big)=\lambda\big(\phi(A)\big)=\lambda^{\phi}\big(A\big)$.
\end{proof}
\begin{proof}[Proof of Lemma~\ref{lem:update-one-preserve}]
By the definition of $\psi$, for measurable $A\in\mathsf{X}\times\mathsf{Y}$
\[
\psi^{-1}(A)=\{(x,y):(x,y+f(x))\in A.
\]
Let $\lambda_{X}$ and $\lambda_{Y}$ be, respectively, the Lebesgue
measures on $\mathbb{R}^{d_{X}}$ and $\mathbb{R}^{d_{Y}}$. Using
the translation-invariance of the Lebesgue measure, we obtain that
for arbitrary, measurable $A$,
\begin{align*}
\lambda(\psi^{-1}(A)) & =\int_{\mathbb{R}^{d_{X}}}\int_{\{y:(x,y+f(x))\in A\}}\lambda_{Y}({\rm d}y)\lambda_{X}({\rm d}x)\\
 & =\int_{\mathbb{R}^{d_{X}}}\int_{\{y:(x,y)\in A\}}\lambda_{Y}({\rm d}y)\lambda_{X}({\rm d}x)\\
 & =\lambda(A),
\end{align*}
from which we can conclude that $\lambda^{\psi}=\lambda$.
\end{proof}
\begin{proof}[Proof of Lemma~\ref{lem:alternating-maps-preserve-reversible}]
Let $\lambda$ denote the Lebesgue measure on $\mathbb{R}^{2d}$.
By Lemma~\ref{lem:update-one-preserve}, $\psi_{A}$ and $\psi_{B}$
each preserve $\lambda$, and hence $\psi$ preserves $\lambda$ as
a composition of $\lambda$-preserving maps. We observe that 
\[
\sigma\circ\psi(x,v)=\big(x+\jmath(v+\imath(x)),-v-\imath(x)-\imath[x+\jmath(v+\imath(x)]\big),
\]
so that 
\[
\psi_{B}\circ\sigma\circ\psi(x,v)=\big(x+\jmath[v+\imath(x)],-v-\imath(x)\big),
\]
and using $\jmath(-v')=-\jmath(v')$,
\[
\psi_{A}\circ\psi_{B}\circ\sigma\circ\psi(x,v)=\big(x,-v-\imath(x)\big).
\]
It follows that 
\[
\psi\circ\sigma\circ\psi(x,v)=(x,-v),
\]
and so $\sigma\circ\psi\circ\sigma\circ\psi(x,v)=(x,v)$, from which
we conclude that $\psi^{-1}=\sigma\circ\psi\circ\sigma$.
\end{proof}
\begin{proof}[Proof of Lemma~\ref{lem:stop-IMH}]
Let $n=\tau(\mathtt{Z})$, so $s_{k}(\mathtt{Z})=0$ for $k\in\llbracket n-1\rrbracket$
and $s_{n}(\mathtt{Z})=1$. From (\ref{eq:generic-def-tau}) it is
sufficient to show that $s_{k}\circ\sigma_{l}(\mathtt{Z})=s_{k}(\mathtt{Z})$
for 
\[
(k,l)\in\llbracket n\rrbracket\times\llbracket n-1\rrbracket=\left\{ 1\leq l\leq k\leq n,l<n-1\right\} \cup\left\{ 1\leq k<l\leq n-1\right\} =:S_{1}\cup S_{2},
\]
to establish the result. From condition~\ref{enu:IMH-lem-tau-constant},
for $(k,l)\in S_{1}$, $s_{k}\circ\sigma_{l}(\mathtt{Z})=s_{k}(\mathtt{Z})$
and in particular $s_{n-1}\circ\sigma_{l}(\mathtt{Z})=0$ for $l\in\llbracket n-1\rrbracket$
while $s_{n}\circ\sigma_{l}(\mathtt{Z})=1$. From condition~\ref{enu:sk-nondecreasing}
and then condition~\ref{enu:IMH-lem-tau-constant}, for $(k,l)\in S_{2}$,
$s_{k}\circ\sigma_{l}(\mathtt{Z})\leq s_{l}\circ\sigma_{l}(\mathtt{Z})=s_{l}(\mathtt{Z})=0$
and so $s_{k}\circ\sigma_{l}(\mathtt{Z})=0=s_{k}(\mathtt{Z})$.
\end{proof}
\begin{proof}[Proof of Proposition \ref{prop:measure-preserving-reversal}]
From Theorem \ref{thm:change-of-variables} and the fact that $\nu^{\psi}=\nu$,
for any $f,g\colon\mathsf{Z}\rightarrow[0,1]$
\begin{align*}
\int f(z)g(z')\nu({\rm d}z)Q(z,{\rm d}z') & =\int f(z)g\circ\psi(z)\nu({\rm d}z)\\
 & =\int f\circ\psi^{-1}(z)g(z)\nu^{\psi}({\rm d}z)\\
 & =\int f\circ\psi^{-1}(z)g(z)\nu({\rm d}z)\\
 & =\int f(z')g(z)\nu({\rm d}z)Q^{*}(z,{\rm d}z'),
\end{align*}
from which one can conclude.
\end{proof}
\begin{proof}[Proof of Lemma \ref{lemma:markov-chain-change-of-measure}]
Let $k\in\mathbb{Z}$, then the probability measure $\Lambda^{k}$
has finite dimensional distributions satisfying, for $n\geq|k|$,
\[
\Lambda_{n}^{k}({\rm d}\mathtt{Z})=\pi({\rm d}z_{k})\prod_{i=k+1}^{n}Q(z_{i-1},{\rm d}z_{i})\prod_{i=-n}^{k-1}Q^{*}(z_{i+1},{\rm d}z_{i}),
\]
which also guarantee the existence of $\Lambda^{k}$ by Kolmogorov's
Extension Theorem \citep{billingsley1995probability}. Notice that
for $n\geq|k|$ and $\mathtt{Z}\in S_{k}$, since $(\nu,Q,Q^{*})$
is a reversible triplet,
\begin{align*}
\nu({\rm d}z_{k})\prod_{i=k+1}^{n}Q(z_{i-1},{\rm d}z_{i})\prod_{i=-n}^{k-1}Q^{*}(z_{i+1},{\rm d}z_{i})\\
=\nu({\rm d}z_{k-{\rm sign}(k)}) & \prod_{i=k+1-{\rm sign}(k)}^{n}Q(z_{i-1},{\rm d}z_{i})\prod_{i=-n}^{k-1-{\rm sgn}(k)}Q^{*}(z_{i+1},{\rm d}z_{i}),
\end{align*}
implying,
\begin{align*}
\Lambda_{n}^{k}({\rm d}\mathtt{Z}) & =\varpi(z_{k})\nu({\rm d}z_{k})\prod_{i=k+1}^{n}Q(z_{i-1},{\rm d}z_{i})\prod_{i=-n}^{k-1}Q^{*}(z_{i+1},{\rm d}z_{i})\\
 & =\varpi(z_{k})\nu({\rm d}z_{0})\prod_{i=1}^{n}Q(z_{i-1},{\rm d}z_{i})\prod_{i=-n}^{-1}Q^{*}(z_{i+1},{\rm d}z_{i})\\
 & =\frac{\varpi(z_{k})}{\varpi(z_{0})}\pi({\rm d}z_{0})\prod_{i=1}^{n}Q(z_{i-1},{\rm d}z_{i})\prod_{i=-n}^{-1}Q^{*}(z_{i+1},{\rm d}z_{i})\\
 & =\frac{\varpi(z_{k})}{\varpi(z_{0})}\Lambda_{n}^{0}({\rm d}\mathtt{Z}),
\end{align*}
from which we conclude by application of \citet[Theorem~4.3.5]{durrett2019probability},
which is a mild generalization of \citet{Engelbert1980}.
\end{proof}
\begin{proof}[Proof of Lemma~\ref{lem:non-measure-preserve-map}]
For $f,g\colon\mathsf{Z}\rightarrow[0,1]$ we have
\begin{align*}
\int f(z)g(z')\nu({\rm d}z)\Psi(z,{\rm d}z') & =\int f(z)g\circ\psi(z)\nu({\rm d}z)\\
 & =\int f\circ\psi^{-1}(z)g(z)\nu^{\psi}({\rm d}z)\\
 & =\int f(z')g(z)\nu^{\psi}({\rm d}z)\Psi^{*}(z,{\rm d}z').
\end{align*}
\end{proof}
\begin{proof}[Proof of Lemma \ref{lem:sn-properties-1}]
 Part~\ref{enu:lem:sn-enu1} is clear from the definition of $s_{n}$.
To establish parts \ref{enu:lem:sn-enu2} and \ref{enu:lem:sn-enu3}
we use the decomposition
\begin{align*}
s_{n}(\mathtt{Z},b) & =g_{n-1}(z_{-\ell_{n}(b)},\ldots,z_{m_{n}(b)})\vee f_{n-2}(z_{-\ell_{n}(b)},\ldots,z_{-\ell_{n}(b)+2^{n-2}-1})\vee f_{n-2}(z_{-\ell_{n}(b)+2^{n-2}},\ldots,z_{m_{n}(b)})\vee\\
 & \qquad g_{n-1}(z_{m_{n}(b)+1},\ldots,z_{r_{n}(b)})\vee f_{n-2}(z_{m_{n}(b)+1},\ldots,z_{m_{n}(b)+2^{n-2}})\vee f_{n-2}(z_{m_{n}(b)+2^{n-2}+1},\ldots,z_{r_{n}(b)}).
\end{align*}
It follows that $s_{n}(\mathtt{Z},b)\geq s_{n-1}(\mathtt{Z},b)$,
since if $b_{n}=0$ then $\llbracket-\ell_{n}(b),m_{n}(b)\rrbracket=\llbracket-\ell_{n-1}(b),r_{n-1}(b)\rrbracket$
and so
\[
s_{n-1}(\mathtt{Z},b)=f_{n-2}(z_{-\ell_{n}(b)},\ldots,z_{-\ell_{n}(b)+2^{n-2}-1})\vee f_{n-2}(z_{-\ell_{n}(b)+2^{n-2}},\ldots,z_{m_{n}(b)})
\]
and if $b_{n}=1$ then $\llbracket m_{n}(b)+1,r_{n}(b)\rrbracket=\llbracket-\ell_{n-1}(b),r_{n-1}(b)\rrbracket$
and so
\[
s_{n-1}(\mathtt{Z},b)=f_{n-2}(z_{m_{n}(b)+1},\ldots,z_{m_{n}(b)+2^{n-2}})\vee f_{n-2}(z_{m_{n}(b)+2^{n-2}+1},\ldots,z_{r_{n}(b)}).
\]
For the same reasons $s_{n}(\mathtt{Z},b)\geq g_{n-1}(z_{-\ell_{n-1}(b)},\ldots,z_{r_{n-1}(b)})$.
\end{proof}
\begin{proof}[Proof of Lemma~\ref{lem:nuts-tau-equal-1}]
 The uniqueness of $b'_{1:n-1}$ follows from the fact that $\ell$
is uniquely determined by $b_{1:n-1}$ and $b'_{1:n-1}$ is uniquely
determined by $\ell+k$. We have $\ell_{n-1}(b')=\ell+k$ and $r_{n-1}(b')=r-k$
by construction. Since for $i\in\mathbb{Z}$ $z'_{i}=z_{k+i}$, it
follows that $z'_{-\ell_{n-1}(b')}=z_{k-\ell-k}=z_{-\ell}$ and $z'_{r_{n-1}(b')}=z_{k+r-k}=z_{r}$,
so that indeed $\llbracket z'_{-(\ell+k)},z'_{r-k}\rrbracket=\llbracket z{}_{-\ell},z{}_{r}\rrbracket$.
Since $b'_{n}=b_{n}$ and $\ell_{n-1}(b')=\ell_{n-1}(b)+k$, we also
have $\ell_{n}(b')=\ell_{n-1}(b')+b_{n}2^{n-1}=\ell_{n}(b)+k$ and
similarly $r_{n}(b')=r_{n}(b)-k$, so $\llbracket z'_{-\ell_{n}(b')},z'_{r_{n}(b')}\rrbracket=\llbracket z{}_{-\ell_{n}(b)},z{}_{r_{n}(b)}\rrbracket$.
Since $\tau(\mathtt{Z},b)=n$, $s_{n}(\mathtt{Z},b)=1$ and $s_{n-1}(\mathtt{Z},b)=0$.
Since $s_{n-1}(\mathtt{Z},b)$ and $s_{n}(\mathtt{Z},b)$ depend only
on the values and the order of their inputs, and not the way they
are indexed, we have $s_{n-1}(\mathtt{Z},b)=s_{n-1}(\mathtt{Z}',b')=0$
and $s_{n}(\mathtt{Z},b)=s_{n}(\mathtt{Z}',b')=1$. To conclude that
$\tau(\mathtt{Z}',b')=n$, it remains only to show that $s_{i}(\mathtt{Z}',b')=0$
for all $i\in\llbracket1,n-2\rrbracket$, but this is implied by Lemma~\ref{lem:sn-properties-1}-\ref{enu:lem:sn-enu2}.
\end{proof}
\begin{proof}[Proof of Proposition \ref{prop:DR-random}]
 Let $\xi=(k,\mathtt{Z}^{k})\in S_{k}\cap\{\mathtt{Z}^{k}\in\mathsf{Z}^{k}\colon\alpha_{k}(\mathtt{Z}^{k})\wedge\alpha_{k}\circ\phi_{k}(\mathtt{Z}^{k})>0\}$.
Then
\begin{align*}
r(\xi) & =\frac{{\rm d}\eta_{k,S_{k}}^{\phi_{k}}}{{\rm d}\eta_{k,S_{k}}}(\mathtt{Z}^{k})\frac{\beta_{k}\circ\phi_{k}}{\beta_{k}}(\mathtt{Z}^{k})\frac{\alpha_{k}\circ\phi_{k}}{\alpha_{k}}(\mathtt{Z}^{k})\\
 & =r_{k}(\mathtt{Z}^{k})\frac{\alpha_{k}\circ\phi}{\alpha_{k}}(\mathtt{Z}^{k})\\
 & =1,
\end{align*}
where we have used that $\alpha_{k}=r_{k}\cdot\alpha_{k}\circ\phi_{k}$
on $S_{k}$ by part~\ref{enu:invo-rev-alpha-decomp} of Theorem~\ref{thm:invo-rev},
applied with $\mu=\eta_{k}\cdot\beta_{k}$ and $\phi=\phi_{k}$.
\end{proof}

\section{Measure theory tools\label{sec:Measure-theory-tools}}

\subsection{Standard results}
\begin{thm}[Change of variables formula for Lebesgue measure]
\label{thm:cov-leb}Let $\phi$ be a continuously differentiable,
invertible function. If $f:\mathbb{R}^{d}\to\mathbb{R}$ is integrable
then, with $\lambda$ the Lebesgue measure,
\[
\int_{\mathbb{R}^{d}}f\circ\phi(\xi)\left|{\rm det}\phi'(\xi)\right|\lambda({\rm d}\xi)=\int_{\mathbb{R}^{d}}f(\xi)\lambda({\rm d}\xi),
\]
where $\phi'(\xi)$ is the Jacobian matrix with entries $\phi'(\xi)_{ij}=\partial\phi_{i}/\partial\xi_{j}(\xi)$.
\end{thm}
This is covered by \citet[Theorem 17.2]{billingsley1995probability}.
\begin{example}[Jacobian of a linear mapping]
 Consider the Lebesgue measure on $\big(\mathbb{R},\mathscr{B}(\mathbb{R})\big)$
such that for any $a,b\in\mathbb{R}$, $a<b$ $\lambda\big((a,b]\big)=b-a$
and consider the scenario $\phi(\xi)=\alpha\xi$ where, without lost
of generality, $\alpha>0$. Recalling the definition $\lambda^{\phi^{-1}}(A):=\lambda\big(\phi(A)\big)$
for any $A\in\mathscr{B}(\mathbb{R})$ we have for $a,b\in\mathbb{R},a<b$
\[
\lambda^{\phi^{-1}}\big((a,b]\big)=\lambda\big((\alpha a,\alpha b]\big)=\alpha(b-a)=\alpha\lambda\big((a,b]\big),
\]
that is $\lambda^{\phi^{-1}}=\alpha\lambda$ and $\lambda\equiv\lambda^{\phi^{-1}}$.
We deduce on the one hand that
\begin{align*}
\int f(\xi)\lambda({\rm d}\xi) & =\int f\circ\phi(\xi)\lambda^{\phi^{-1}}({\rm d}\xi)\\
 & =\int f\circ\phi(\xi)\alpha\lambda({\rm d}\xi)\\
 & =\int f\circ\phi(\xi){\rm det}\phi'(\xi)\lambda({\rm d}\xi)
\end{align*}
and using the Radon-Nikodym theorem (Theorem~\ref{thm:Radon-Nikodym})
we also have
\begin{align*}
\int f(\xi)\lambda({\rm d}\xi) & =\int f\circ\phi(\xi)\lambda^{\phi^{-1}}({\rm d}\xi)\\
 & =\int f\circ\phi(\xi)\frac{{\rm d}\lambda^{\phi^{-1}}}{{\rm d}\lambda}(\xi)\lambda({\rm d}\xi)
\end{align*}
and we deduce that, $\lambda-$almost everywhere,
\[
\frac{{\rm d}\lambda^{\phi^{-1}}}{{\rm d}\lambda}(\xi)=\left|{\rm det}\phi'(\xi)\right|.
\]
This result can be generalised to the multivariate scenario but also
to nonlinear invertible and smooth mappings $\phi\colon\mathbb{R}^{d}\rightarrow\mathbb{R}^{d}$
by local linearisation.
\end{example}

\subsection{Proofs\label{subsec:app-proofs-measure}}
\begin{proof}[Proof of Proposition \ref{prop:r-density} ]
 First observe that ${\rm d}\mu^{\phi}/{\rm d}\lambda^{\phi}=\rho\circ\phi$:
for any $A\in\mathscr{E}$,
\begin{align*}
\int{\bf 1}_{A}(\xi)\rho\circ\phi(\xi)\lambda^{\phi}({\rm d}\xi) & =\int{\bf 1}_{A}\circ\phi(\xi)\rho(\xi)\lambda({\rm d}\xi)\\
 & =\int{\bf 1}_{A}\circ\phi(\xi)\mu({\rm d}\xi)\\
 & =\int{\bf 1}_{A}(\xi)\mu^{\phi}({\rm d}\xi).
\end{align*}
Then we find for $A\in\mathscr{E}$, $A\subseteq S$,
\begin{align*}
\int{\bf 1}_{A}(\xi)\frac{\rho\circ\phi}{\rho}(\xi)\frac{{\rm d}\lambda^{\phi}}{{\rm d}\lambda}(\xi)\mu_{S}({\rm d}\xi) & =\int{\bf 1}_{A}(\xi)\frac{\rho\circ\phi}{\rho}(\xi)\frac{{\rm d}\lambda^{\phi}}{{\rm d}\lambda}(\xi)\mu_{S}({\rm d}\xi)\\
 & =\int{\bf 1}_{A}(\xi)\frac{\rho\circ\phi}{\rho}(\xi)\frac{{\rm d}\lambda^{\phi}}{{\rm d}\lambda}(\xi)\rho(\xi)\lambda({\rm d}\xi)\\
 & =\int{\bf 1}_{A}(\xi)\rho\circ\phi(\xi)\lambda^{\phi}({\rm d}\xi)\\
 & =\int{\bf 1}_{A}(\xi)\mu_{S}^{\phi}({\rm d}\xi),
\end{align*}
so that indeed ${\rm d}\mu_{S}^{\phi}/{\rm d}\mu_{S}=\frac{\rho\circ\phi}{\rho}\cdot\frac{{\rm d}\lambda^{\phi}}{{\rm d}\lambda}$.
The proof that $\mu$ and $\mu^{\phi}$ are mutually singular on $S^{\complement}$
follows the same arguments as in the proof of Theorem~\ref{thm:invo-rev}.
\end{proof}

\section{X-tra chance proof\label{sec:X-tra-chance-proof}}

This is a proof of the claims in Remark~\ref{eg:xtra-chance}. Fix
$u\in\mathbb{R}_{+}$, we show the result by induction. First we have
$\beta_{1}(z)=1\times1$, $\beta_{1}\circ\phi_{1}(z)=1$ and by considering
$z\in S(\varpi_{u},\varpi_{u}^{\phi_{1}})$ and $z\in S^{\complement}(\varpi_{u},\varpi_{u}^{\phi_{1}})$
separately, and Theorem~\ref{def:slice-sampler}-\ref{def:slice-sampler}
we obtain $r_{1}(z)=\mathbb{I}\{u\leq\varpi\}\mathbb{I}\{u\leq\varpi\circ\phi_{1}(z)\}=\mathbb{I}\{u\leq\varpi(z)\wedge\varpi\circ\phi_{1}(z)\}$
and therefore with $a(r)=1\wedge r$ we deduce $\beta_{2}(z)=\beta_{1}(z)[1-\mathbb{I}\{u\leq\varpi(z)\wedge\varpi\circ\phi_{1}(z)\}]=\mathbb{I}\{\varpi(z)\wedge\varpi\circ\phi_{1}(z)<u\}$.
Assume that for some $k\in\left\llbracket 2,n\right\rrbracket $ and
any $z\in\mathsf{Z}$
\[
\beta_{k}(z)=\mathbb{I}\{\varpi(z)\wedge\vee_{i=1}^{k-1}\varpi\circ\phi_{i}(z)<u\}.
\]
From the assumption $\varpi\circ\phi_{i}\circ\phi_{k}=\varpi\circ\phi_{k-i}$
this implies
\begin{align*}
\beta_{k}\circ\phi_{k}(z) & =\mathbb{I}\{\varpi\circ\phi_{k}(z)\wedge\vee_{i=1}^{k-1}\varpi\circ\phi_{k-i}(z)<u\}=\mathbb{I}\{\varpi\circ\phi_{k}(z)\wedge\vee_{i=1}^{k-1}\varpi\circ\phi_{i}(z)<u\}.
\end{align*}
Therefore, proceeding as for $r_{1}(z)$ above and taking advantage
of the fact that $\beta_{k}(\xi),\beta_{k}\circ\phi_{k}(z)\in\{0,1\}$
we obtain
\begin{align*}
r_{k}(z) & =\beta_{k}(z)\beta_{k}\circ\phi_{k}(z)\mathbb{I}\{u\leq\varpi(z)\wedge\varpi\circ\phi_{k}(z)\}\\
 & =\mathbb{I}\{\varpi(z)\wedge\vee_{i=1}^{k-1}\varpi\circ\phi_{i}(z)<u\leq\varpi(z)\}\mathbb{I}\{\varpi\circ\phi_{k}(z)\wedge\vee_{i=1}^{k-1}\varpi\circ\phi_{i}(z)<u\leq\varpi\circ\phi_{k}(z)\},\\
 & =\mathbb{I}\{\vee_{i=1}^{k-1}\varpi\circ\phi_{i}(z)<u\leq\varpi(z)\}\mathbb{I}\{\vee_{i=1}^{k-1}\varpi\circ\phi_{i}(z)<u\leq\varpi\circ\phi_{k}(z)\}\\
 & =\mathbb{I}\{\vee_{i=1}^{k-1}\varpi\circ\phi_{i}(z)<u\leq\varpi(z)\wedge\varpi\circ\phi_{k}(z)\}.
\end{align*}
When $\vee_{i=1}^{k-1}\varpi\circ\phi_{i}(z)<\varpi(z)\wedge\varpi\circ\phi_{k}(z)$
\[
1-r_{k}(z)=\mathbb{I}\{u\leq\vee_{i=1}^{k-1}\varpi\circ\phi_{i}(z)\}+\mathbb{I}\{\varpi(z)\wedge\varpi\circ\phi_{k}(z)<u\}
\]
therefore
\begin{align*}
\beta_{k+1}(z) & =\mathbb{I}\{\varpi(z)\wedge\vee_{i=1}^{k-1}\varpi\circ\phi_{i}(z)<u\}\left[\mathbb{I}\{\varpi(z)\wedge\varpi\circ\phi_{k}(z)<u\}+\mathbb{I}\{u\leq\vee_{i=1}^{k-1}\varpi\circ\phi_{i}(z)\}\right]\\
 & =\mathbb{I}\{\varpi(z)\wedge\vee_{i=1}^{k}\varpi\circ\phi_{i}(z)<u\}+\mathbb{I}\{\varpi(z)\wedge\vee_{i=1}^{k-1}\varpi\circ\phi_{i}(z)<u\leq\vee_{i=1}^{k-1}\varpi\circ\phi_{i}(z)\}\\
 & =\mathbb{I}\{\varpi(z)\wedge\vee_{i=1}^{k}\varpi\circ\phi_{i}(z)<u\},
\end{align*}
where we have used that $\mathbb{I}\{\varpi(z)\wedge\vee_{i=1}^{k-1}\varpi\circ\phi_{i}(z)<u\}=\mathbb{I}\{\vee_{i=1}^{k-1}\big(\varpi(z)\wedge\varpi\circ\phi_{i}(z)\big)<u\}$. 

When $\vee_{i=1}^{k-1}\varpi\circ\phi_{i}(z)\geq\varpi(z)\wedge\varpi\circ\phi_{k}(z)$,
using the same argument,
\begin{align*}
\beta_{k+1}(\xi) & =\mathbb{I}\{\varpi(z)\wedge\vee_{i=1}^{k-1}\varpi\circ\phi_{i}(z)<u\}\\
 & =\mathbb{I}\{\vee_{i=1}^{k}\big(\varpi(z)\wedge\varpi\circ\phi_{i}(z)\big)<u\}\\
 & =\mathbb{I}\{\varpi(z)\wedge\vee_{i=1}^{k}\varpi\circ\phi_{i}(z)<u\},
\end{align*}
which completes the proof. 

\section{\label{sec:app-NUTS-motivation}NUTS motivation}

The criterion consists of stopping when $\left\Vert x_{r}-x_{\ell}\right\Vert ^{2}$
reaches a stationary point, in the hope that it is a maximum. This
requires the computation of a differential, that is the first order
linear approximation of variations of $\left\Vert x_{r}-x_{\ell}\right\Vert ^{2}$when
$x_{\ell}$ (resp. $x_{r}$) is perturbed linearly $x_{\ell}+\epsilon v_{\ell}$
(resp. $x_{r}+\epsilon v_{r}$). This leads to
\[
\left\Vert x_{r}-(x_{\ell}+\epsilon v_{\ell})\right\Vert _{2}^{2}-\left\Vert x_{r}-x_{\ell}\right\Vert _{2}^{2}=-2\epsilon(x_{r}-x_{\ell})^{\top}v_{\ell}+\epsilon^{2}\left\Vert v_{\ell}\right\Vert _{2}^{2},
\]
and
\[
\left\Vert x_{r}+\epsilon v_{r}-x_{\ell}\right\Vert _{2}^{2}-\left\Vert x_{r}-x_{\ell}\right\Vert _{2}^{2}=2\epsilon(x_{r}-x_{\ell})^{\top}v_{r}+\epsilon^{2}\left\Vert v_{r}\right\Vert _{2}^{2},
\]
As $\epsilon\downarrow0$ the dominant and linear term has coefficient
\[
\lim_{\epsilon\downarrow0}\frac{1}{\epsilon}\left\{ \left\Vert x_{r}-(x_{\ell}+\epsilon v_{\ell})\right\Vert _{2}^{2}-\left\Vert x_{r}-x_{\ell}\right\Vert _{2}^{2}\right\} <0\iff(x_{r}-x_{\ell})^{\top}v_{\ell}>0,
\]
and
\[
\lim_{\epsilon\downarrow0}\frac{1}{\epsilon}\left\{ \left\Vert x_{r}+\epsilon v_{r}-x_{\ell}\right\Vert _{2}^{2}-\left\Vert x_{r}-x_{\ell}\right\Vert _{2}^{2}\right\} <0\iff(x_{r}-x_{\ell})^{\top}v_{r}<0.
\]

\section{Event chain algorithms \label{sec:Event-chain-algorithms}}

We briefly describe standard event chain processes for soft potentials
and pairwise interactions. Define $x=(x_{1},x_{2},\ldots,x_{m})\in\mathsf{X}^{m}$
with $\mathsf{X}=\mathbb{T}^{d}:=-1/2+\mathbb{R}^{d}/\mathbb{Z}^{d}$
and $v\in\mathsf{V}\subset\mathbb{R}^{d}$. The target distribution
of interest has density 
\[
\gamma(x,v,i)=\gamma(x)\kappa(v)\propto\exp\big(-U(x)\big)\kappa(v)\mathbb{I}\{i\in\llbracket m\rrbracket\}
\]
where $\gamma$ has density with respect to the measure induced by
the Lebesgue measure on $[-1/2,1/2)^{d}$ \citet[Chapter 6]{folland1999real}
and $\kappa$ is the density with respect to the Lebesgue measure
on $\mathsf{V}=\mathbb{R}^{d}$ or the Hausdorff measure on $\mathsf{V}=\mathbb{S}^{d-1}$.
It is further assumed that $\kappa(-v)=\kappa(v)$ for $v\in\mathsf{V}$
and we focus on the scenario involving pairwise interactions,
\[
U(x):=\sum_{1\leq i<j\leq m}V(x_{i}-x_{j}),
\]
where $V:(-1,1)^{d}\rightarrow\mathbb{R}_{+}$ is continuously differentiable
and such that $V(x_{*})=V(-x_{*})$ for all $x_{*}\in\mathsf{X}$.
This leads to a probability density with exchangeability properties.
The generator corresponding to event chain processes is given by
\[
Lf(x,v,i)=\langle\nabla_{x}f,\mathbf{e}_{i}\varoast v\rangle+\lambda(x,v,i)\cdot\big[Rf(x,v,i)-f(x,v,i)\big]+\lambda_{\text{{\rm ref}}}\cdot\left[\int f(x,w,i)\kappa({\rm d}w)-f(x,v,i)\right],
\]
for $\lambda_{{\rm ref}}>0$, $\{\mathbf{e}_{i},i\in\llbracket m\rrbracket\}$
the canonical basis vectors and here $\varoast$ the Kronecker product.
The intensity of the process is taken to be of the form
\[
\lambda(x,v,i)=\sum_{j=1}^{m}\lambda_{j}(x,v,i)
\]
with the convention $\lambda_{i}(x,v,i)=0$ and for $j\neq i$, with
$\langle\cdot,\cdot\rangle_{+}:=\max\left\{ 0,\langle\cdot,\cdot\rangle\right\} $
,
\[
\lambda_{j}\big(x,v,i\big):=\bigl\langle\nabla V_{*}(x_{i}-x_{j}),v\bigr\rangle_{+},
\]
and for $(x,v,i),(y,w,j)\in\mathsf{X}\times\mathsf{V}\times\llbracket m\rrbracket$,
\[
R\big((x,v,i),{\rm d}(y,w,j)\big):=\sum_{k=1,k\neq i}^{m}\frac{\lambda_{k}(x,v,i)}{\lambda(x,v,i)}\delta_{(x,v,k)}\big({\rm d}(y,w,j)\big).
\]
This means that we follow trajectories of the form $t\mapsto(x_{1},\ldots,x_{i-1},x_{i}+tv,x_{i+1},\ldots,x_{m},v,i)$
with $t\geq0$ for a random time arising from an inhomogeneous Poisson
process of intensity $t\mapsto\lambda(x+t\,\mathbf{e}_{i}\otimes v,v,i)+\lambda_{{\rm ref}}$,
a time at which one chooses between refreshing the velocity or selecting
a new active particle randomly.

We check now that the corresponding process leaves the correct distribution
invariant. We know that it is sufficient to show that $\mu(Lf)=0$
for all functions $f\colon\mathsf{X}\times\mathsf{V}\times\llbracket m\rrbracket\rightarrow\mathbb{R}$
in a core of $\big(L,D(L)\big)$. Using \citet{2018arXiv180705421D},
it can be shown that the functions $f\colon\mathsf{X}\times\mathsf{V}\times\llbracket m\rrbracket\rightarrow\mathbb{R}$
such that for $i\in\llbracket m\rrbracket$, $f(\cdot,i)\in\mathsf{C}_{b}^{2}(\mathsf{X\times\mathsf{V}})$
(bounded support and twice continuously differentiable) define such
a core. In fact with the isometric involution $\mathfrak{S}f(x,v,i)=f(x,-v,i)$,
we can show the stronger property $\bigl\langle Lf,g\bigr\rangle_{\mu}=\bigl\langle f,\mathfrak{S}L\mathfrak{S}g\bigr\rangle_{\mu}$,
for $f,g\colon\mathsf{X}\times\mathsf{V}\times\llbracket m\rrbracket\rightarrow\mathbb{R}$
such that the integral exists and where $\bigl\langle f,g\bigr\rangle_{\mu}:=\int fg{\rm d}\mu$,
which is the continuous time formulation of $(\mu,\mathfrak{S})-$reversibility
\citet{Andrieu2019}. The property $\mu(Lf)=0$ can be deduced by
setting $g=\mathbf{1}$. We establish an intermediate result from
which this latter property can be deduced.
\begin{lem}
\label{lem:result-intensities}Let\textup{ $V:\mathsf{X}=\mathbb{T}^{d}\rightarrow\mathbb{R}_{+}$
be continuously differentiable }and such that $V(x_{*})=V(-x_{*})$
for all $x_{*}\in\mathsf{X}$. Then for $i,j\in\llbracket m\rrbracket$,
$i\neq j$
\begin{enumerate}
\item $\lambda(x,v,i)-\mathfrak{S}\lambda(x,v,i)=\bigl\langle\nabla_{x}U(x),\mathbf{e}_{i}\varoast v\bigr\rangle$,
\item $\mathfrak{S}\lambda_{j}\big(x,v,i\big)=\lambda_{i}\big(x,v,j\big).$
\end{enumerate}
\end{lem}
\begin{proof}
The first relation follows, for $i\in\llbracket m\rrbracket$, from
\begin{align*}
\sum_{j\neq i}\lambda_{j}\big(x,v,i\big)-\mathfrak{S}\lambda_{j}\big(x,v,i\big) & =\sum_{j\neq i}\bigl\langle\nabla_{*}V(x_{i}-x_{j}),v\bigr\rangle\\
 & =\bigl\langle\sum_{j\neq i}\nabla_{*}V(x_{i}-x_{j}),v\bigr\rangle\\
 & =\bigl\langle\nabla_{x}U(x),\mathbf{e}_{i}\varoast v\bigr\rangle.
\end{align*}
The second property follows from the assumption $V(x_{*})=V(-x_{*})=V\circ s(x_{*})$
where $s(x_{*})=-x_{*}$. Indeed, in this scenario the chain rule
leads to $\nabla_{*}V(x_{*})=(\nabla_{*}\varoast s)(\nabla V)\circ s(x_{*})=-\nabla_{*}V(-x_{*})$
and consequently for $i,j\in\left\llbracket 1,m\right\rrbracket $,
$i\neq j$ 
\begin{align*}
\mathfrak{S}\lambda_{j}\big(x,v,i\big) & =\bigl\langle-\nabla_{*}V(x_{i}-x_{j}),v\bigr\rangle_{+}\\
 & =\bigl\langle\nabla_{*}V(x_{j}-x_{i}),v\bigr\rangle_{+}\\
 & =\lambda_{i}\big(x,v,j\big).
\end{align*}
\end{proof}
We now prove $\mu(Lf)=0$. We can clearly ignore the refreshment component
of the generator. An integration by part and Lemma~\ref{lem:result-intensities}
establish that
\begin{align*}
\int\langle\nabla_{x}f(x,v,i),\mathbf{e}_{i}\otimes v\rangle\mu\big({\rm d}(x,v,i)\big) & =\int\langle\nabla_{i}f(x,v,i),v\rangle\mu\big({\rm d}(x,v,i)\big)\\
 & =\int f(x,v,i)\langle\nabla_{i}U(x),v\rangle\mu\big({\rm d}(x,v,i)\big).\\
 & =\int f(x,v,i)\left[\lambda(x,v,i)-\mathfrak{S}\lambda(x,v,i)\right]\mu\big({\rm d}(x,v,i)\big)\\
 & =\int f(x,v,i)\sum_{j\neq i}\big[\lambda_{j}\big(x,v,i\big)-\mathfrak{S}\lambda_{j}\big(x,v,i\big)\big]\mu\big({\rm d}(x,v,i)\big)\\
 & =\int\sum_{j\neq i}\lambda_{j}\big(x,v,i\big)\big[f(x,v,i)-f(x,v,j)\big]\mu\big({\rm d}(x,v,i)\big).
\end{align*}
and we conclude by noting that
\[
\int\lambda(x,v,i)\cdot\big[Rf(x,v,i)-f(x,v,i)\big]\mu\big({\rm d}(x,v,i)\big)=\int\sum_{j\neq i}\lambda_{j}(x,v,i)\big[f(x,v,j)-f(x,v,i)\big]\mu\big({\rm d}(x,v,i)\big).
\]
 The same calculations can be used for higher order interactions \citet{harland2017event}.